\documentclass[a4paper,11pt]{article}
\usepackage[a4paper,hmargin=2.8cm,vmargin=3cm]{geometry}
\usepackage[utf8x]{inputenc}
\usepackage{amsmath,amsthm,amsfonts,amssymb}
\usepackage{authblk}
\usepackage[bookmarksnumbered=true]{hyperref}
\usepackage{bbm}
\usepackage{bigints}

\usepackage{color}
\usepackage{appendix}

\usepackage{graphicx}

\numberwithin{equation}{section}

\newtheorem{theorem}{Theorem}[section]

\newtheorem{lemma}[theorem]{Lemma}
\newtheorem{proposition}[theorem]{Proposition}
\theoremstyle{definition}

\newtheorem{remark}{Remark}

\let\oldmarginpar\marginpar
\renewcommand\marginpar[1]{\-\oldmarginpar[\raggedleft\footnotesize #1]%
  {\raggedright\footnotesize #1}}

\def\eps{\varepsilon}
\def\N{\mathbb{N}}

\def\R{\mathbb{R}}
\let\e=\varepsilon

\let\.=\cdot
\let\0=\emptyset

\def\tr{\text{\rm tr}}

\def\square{\hbox{$\sqcap\kern-7pt\sqcup$}}

\def\be{\begin{equation}}
\def\ee{\end{equation}}
\def\bea{\begin{eqnarray}}
\def\eea{\end{eqnarray}}

\def\ˆ{^{•}}

\title{Semiclassical limit to the Vlasov equation\\ with inverse power law potentials}

\author{Chiara Saffirio}

\def\adresse{
\begin{description}
\item[C. Saffirio:] Institute of Mathematics, \\ University of Z\"urich, Winterthurerstrasse 190, CH-8057 Z\"urich, Switzerland\\
E-mail: \texttt{chiara.saffirio@math.uzh.ch}

\end{description}
}

\date{\today}

\begin{document}

\maketitle

\begin{abstract}
We consider mixed quasi-free states describing $N$ fermions in the mean-field limit. In this regime, the time evolution is governed by the nonlinear Hartree equation. In the large $N$ limit, we study the convergence towards the classical Vlasov equation. Under integrability and regularity assumptions on the initial state, we prove strong convergence in trace and Hilbert-Schmidt norm and provide explicit bounds on the convergence rate for a class of singular potentials of the form $V(x)=|x|^{-\alpha}$, for $\alpha\in(0,1/2)$.   
\end{abstract}


\section{Introduction and main results}

We consider a system of $N$ fermions in three dimensions and we are interested in their many-body time evolution in a weakly interacting regime. A state of the system is represented by a wave function $\psi_N\in L^2_a(\R^{3N})$, where
$$
L^2_a(\R^{3N})=\{\psi_N\in L^2(\R^{3N})\ |\ \psi_N(x_{\pi(1)},\dots,x_{\pi(N)})=\sigma_\pi\psi_N(x_1,\dots,x_N),\ \mbox{for all } \pi\in\mathfrak{G}_N\}
$$
is the space of square integrable functions antisymmetric in the exchange of particles, $\mathfrak{G}_N$ is the group of permutation of $N$ elements  and $\sigma_\pi=\{\pm\}$ denotes the sign of the permutation $\pi$.
The Hamilton operator acting on $L^2_a(\R^{3N})$ is given by
\be\label{eq:H}
H_N^{\rm ext} =\sum_{j=1}^N(-\Delta_{x_j}+V_{\rm ext}(x_j))+\lambda\sum_{i<j}^N V(x_i-x_j),
\ee
where $\Delta_{x_j}$ is the standard Laplace operator, $V_{\rm ext}$ is an external potential confining the system in a volume of order one, $\lambda\geq 0$ is a coupling constant to be determined according to the regime we are interested in and $V$ is a two-body interaction potential that we assume for the moment to be smooth.  

The antisymmetry of the wave function (which reflects the Pauli's principle) implies that the kinetic energy of $N$ particles trapped in a volume of order one is at least of the order $N^{5/3}$. Since we are interested in a truly interacting effective picture, we choose $\lambda=N^{-1/3}$ to make sure the potential energy is of the same order of magnitude of the kinetic one. We now assume our initial datum describes equilibrium states of the trapped system and we look at its evolution resulting from a change in the external field. In other words, as $V_{\rm ext}=0$, the system starts to evolve in time.  Pauli's principle also implies that the typical velocity of the particles is of the oder $N^{1/3}$, thus the choice $\lambda=N^{-1/3}$ captures the time evolution for times of the order $N^{1/3}$. Since we are interested in times of order one, we rescale the time variable and we obtain the following $N$-body Schr\"{o}dinger equation
\be\label{eq:SE-N}
i\,N^{1/3}\partial_t\psi_{N,t}=\left[\sum\limits_{j=1}^N(-\Delta_{x_j})+\frac{1}{N^{1/3}}\sum\limits_{i<j}^N V(x_i-x_j)\right] \psi_{N,t}
\ee
where $\psi_{N,t}$ denotes the evolution in time of the state $\psi_N$. To rewrite \eqref{eq:SE-N} in a handier way, 
we introduce the small parameter 
\be\label{eq:scaling}
\e_N=N^{-1/3}
\ee
and multiply by $\e_N^2$ both sides of Eq. \eqref{eq:SE-N}. Hence 
\be\label{eq:SE}
i\e_N\partial_t\psi_{N,t}=\left[\sum\limits_{j=1}^N(-\e_N^2\Delta_{x_j})+\frac{1}{N}\sum\limits_{i<j}^N V(x_i-x_j)\right] \psi_{N,t}
\ee

In Eq. \eqref{eq:SE-N} the parameter $\e_N$ defined in \eqref{eq:scaling} plays the role of Planck constant $\hbar$. Therefore in this setting the mean-field scaling for fermions, characterised by the factor $N^{-1}$ in front of the interaction, comes coupled with a semiclassical limit, that is the focus of this paper. Throughout the paper the dependence of $\e_N$ on $N$ given in \eqref{eq:scaling} holds, though we will omit the subscript $N$ and denote $\e_N$ simply by $\e$ to shorten the notation. 

\medskip

\noindent {\it Evolution of quasi-free states.}  As already mentioned, we are interested in initial data $\psi_N$ describing equilibrium states of trapped systems. These are approximated in the mean-field regime by quasi-free states. In the mean-field regime, quasi-free states are completely characterised by their one-particle reduced density matrix $\omega_N$, the nonnegative trace class operator defined as
\be\label{eq:1pdm}
\omega_N=N\,\tr_{2\dots N}\ |\psi_N\rangle\langle \psi_N|\quad \mbox{ s.t. }\quad \tr\ \omega_N=N,\quad 0\leq\omega_N\leq 1,
\ee
where $\tr_{2\dots N}$ denotes the trace w.r.t. the $N-1$ coordinates $(x_2,\dots,x_N)$.
Moreover, Shale-Stinespring condition (see \cite{LL}) implies that every $\omega_N$ satisfying \eqref{eq:1pdm} is the one-particle reduced density matrix of a quasi-free state with $N$ particles.

At zero temperature, this corresponds to select a Slater determinant
\be\label{eq:slater}
\psi_{\rm Slater}(x_1,\dots,x_N)=\frac{1}{\sqrt{N!}}{\rm det}(f_j(x_i))_{1\leq i,j\leq N},
\ee
where $\{f_j\}_{j=1}^N$ is an orthogonal system in the one-particle space $L^2(\R^3)$. Plugging \eqref{eq:slater} into \eqref{eq:1pdm} it is easy to see that the one-particle reduced density matrix associated to $\psi_{\rm Slater}$ is an orthogonal projection onto the space spanned by $\{f_1,\dots,f_N\}$. In particular, $\omega_N^2=\omega_N$. Though the time evolution of a Slater determinant is not a Slater determinant because the interaction among particles in principle destroys such a structure, it is known that, in the mean-field regime and under further assumptions on the initial data, $\omega_{N,t}$ remains close to a Slater determinant and solves the time-dependent Hartree-Fock equation
\be\label{eq:HF}
i\e\partial_t\omega_{N,t}=[-\e^2\Delta+(V*\rho_t)-X_t\,,\,\omega_{N,t}],
\ee
where, for every $x\in\R^3$, 
$$\rho_t(x)=N^{-1}\omega_{N,t}(x;x)$$
 is the density associated to the one-particle reduced density matrix $\omega_{N,t}$, $(~V~*~\rho_t~)$ represents the so-called direct term, while $X_t$ is the exchange term defined through its operator kernel
$$X_t(x;y)=\frac{1}{N}V(x-y)\omega_{N,t}(x;y).$$
The correctness of the approximation of the many-body Schr\"{o}dinger evolution by the Hartree-Fock dynamics for states close to a Slater determinant has been proved in the case of smooth interaction potentials $V$ in \cite{EESY} for short time intervals and in \cite{BPS13} for arbitrary time intervals, providing in addition an explicit estimate of the convergence rate. In \cite{PRSS}, the Coulomb interaction has been addressed in the aforementioned regime, providing convergence of the time evolution of a Slater determinant towards a solution to the time-dependent  Hartree-Fock equation with Coulomb interaction. The result holds under severe assumptions on the solution of the Hartree-Fock equation, that has to satisfy some integrability and regularity conditions at positive times. In particular, translation invariant states fulfil such assumptions.  

At positive temperature, equilibria are expected to be approximated by mixed states, i.e. quasi-free states (see \eqref{eq:1pdm}) whose one-particle reduced density matrix is not a projection, that is $\omega_N^2\neq \omega_N$. 
The result in \cite{BPS13} has been extended to mixed states in \cite{BJPSS} for regular interaction potentials.

\medskip

\noindent {\it Mean-field in presence of singular interactions.} When dealing with singular interactions $V(x)={1}/{|x|^\alpha}$, $\alpha\in(0,1]$, the Hamiltonian takes the form
\be\label{eq:H-sing}
H_N=\sum_{i=1}^N(-\e^2\Delta_{x_i})+\frac{1}{N}\sum_{i<j}^N\frac{1}{|x_i-x_j|^{\alpha}}.
\ee
In particular, the case $\alpha=1$  treated in \cite{PRSS} represents a system of $N$ fermions interacting through a Coulomb potential, which describes for instance the dynamics of large atoms and molecules. In this case, the choice $\e=N^{-1/3}$ is justified by a rescaling of the space variables at a scale $O(N^{-1/3})$ (the typical distance of the electrons from the nucleus) as suggested by the Thomas-Fermi theory (see \cite{L}, \cite{LSi}). An analogous reasoning applies to the case of inverse power law potentials and, by appropriately scaling the time variable,  it leads to
\be\label{eq:SE-sing}
i\e\partial_t\psi_{N,t}=\left[\sum_{i=1}^N(-\e^2\Delta_{x_i})+\frac{1}{N}\sum_{i<j}^N\frac{1}{|x_i-x_j|^{\alpha}}\right]\psi_{N,t}.
\ee   
More details on the rigorous justification of the mean-field scaling coupled to the semiclassical one in the case of inverse power law potentials can be found in \cite{PRSS, S18}.    We stress that for both regular and singular potentials, the exchange term in the Hartree-Fock equation represents a sub-leading contribution (see Appendix A in \cite{BPS13}). For this reason we will drop it in the rest of the paper.
\medskip

\noindent {\it Semiclassical limit.} The Hartree-Fock Eq. \eqref{eq:HF} still depends on $N$. It is therefore natural to ask what happens in the large $N$ limit, which is equivalent to $\e\to 0$ because of \eqref{eq:scaling}. In this sense, we are here considering a semiclassical limit. To answer this question, we introduce the {\it Wigner transform} of the one-particle reduced density $\omega_{N,t}$, defined as
\be\label{eq:Wigner}
W_{N,t}(x,v)=\frac{\e^3}{(2\pi)^3}\int \omega_{N,t}\left(x+\e\frac{y}{2};x-\e\frac{y}{2}\right)\,e^{-iv\cdot y}\,dy.
\ee
The Wigner transform is therefore a function on the phase space $\R^3\times\R^3$ and it is normalised to 1, i.e.
\begin{equation*}
\int W_{N,t}(x,v)\,dx\,dv=\e^3\,\tr\ \omega_{N,t}=N\,\e^3=1,
\end{equation*} 
but it is not a probability density on the phase space, since in general it is not positive. Nevertheless, we can give physical meaning to the quantities which would correspond to the marginals, i.e. $\int W_{N,t}(x,v)\,dv$ represents the density of fermions at the point $x\in\R^3$ and $\int W_{N,t}(x,v)\,dx$ is the density of particles with velocity $v\in\R^3$.

The Wigner transform can be inverted by means of the {\it Weyl quantization}, defined by
\be\label{eq:Weyl}
\omega_{N,t}(x;y)=N\,\int W_{N,t}\left(\frac{x+y}{2},v\right)\,e^{iv\cdot(x-y)/\e}\,dv\,. 
\ee
We notice that 
$$
\|\omega_{N,t}\|_{\rm HS}=\sqrt{N}\|W_{N,t}\|_{L^2},
$$
where $\|W_{N,t}\|_{L^2}$ is the usual $L^2$ norm on $\R^3\times\R^3$, $\|A\|_{\rm HS}:=\tr\ \sqrt{A^*A}$ denotes the Hilbert-Schmidt norm of a compact operator A and $A^*$ is the adjoint of $A$.

Using \eqref{eq:Wigner} in the Hartree equation \eqref{eq:HF}, it is easy to get an evolution equation for the Wigner transform which suggests that, as $N\to\infty$, $W_{N,t}$ should converge to a probability density $W_t$, solution to the classical Vlasov equation
\be\label{eq:VL}
\partial_t W_t+2v\cdot\nabla_x W_t=\nabla(V*\rho_t)\cdot\nabla_v W_t,
\ee
where $\rho_t(x)=\int W_t(x,v)\,dv$ is the spatial density. 

The well-posedness of the Vlasov Eq. \eqref{eq:VL} for smooth interaction potentals $V$ goes back to Dobrushin \cite{Dobr}. As $V$ is taken to be the Coulomb interaction, then the Vlasov Eq. \eqref{eq:VL} reads 
\be\label{eq:VP}
\left\{\begin{array}{l}
\partial_t W_t(x,v)+v\cdot\nabla_x W_t(x,v)+E(t,x)\cdot\nabla_v W_t(x,v)=0,\\\\
E(t,x)=\nabla\left(\frac{1}{|  \cdot  |}*{\rho}_t\right)(x)\,, \quad\quad\quad \rho_t(x)=\int W_t(x,v)\,dv\,.
\end{array}\right.
\ee
Eq. \eqref{eq:VP} is referred to as
the Vlasov-Poisson equation, since it couples the Vlasov Eq. \eqref{eq:VL} with a Poisson equation (encoded in the definition of $E$). We remark that the factor 2 appearing in \eqref{eq:VL} is consistent with the choice of $-\Delta$ for the quantum kinetic energy instead of $-\frac{1}{2}\Delta$. Hence Eq. \eqref{eq:VL} and \eqref{eq:VP} coincide when the interaction $V$ is chosen to be the Coulomb potential.

Existence of classical solutions to the Cauchy problem associated with \eqref{eq:VP} under regularity assumptions on the initial data has been enstablished in \cite{Iordanski} in dimension one and in \cite{OkabeUkai} in dimension two. Well-posedness in dimension three has been treated by Bardos and Degond in \cite{BD} for small data and by Pfaffelmoser \cite{Pf} and  Lions and Perthame \cite{LP} for more general initial densities. More recently, less stringent uniqueness criteria have been provided in \cite{Loeper, Miot, HoldingMiot}.    
\medskip

\noindent {\it State of art.} The derivation of the Vlasov Eq. \eqref{eq:VL} from quantum evolution equations has been addressed by several authors in the last 40 years. The first results were obtained by Narnhofer and Sewell \cite{NS} and by Spohn \cite{Spohn} in the '80s. In these works the Vlasov Eq. \eqref{eq:VL} is derived from the many-body Schr\"{o}dinger Eq. \eqref{eq:SE} for fermions and bosons in the mean-field regime combined with a semiclassical limit for analytic interactions $V$ in \cite{NS} and for $V\in\mathcal{C}^2(\R^3)$ in \cite{Spohn}. See also \cite{GMP} for a WKB approach for bosons. 

A different viewpoint has been considered in \cite{GIMS, LionsPaul, MM, FigalliLigaboPaul}, where the Vlasov Eq. \eqref{eq:VL} is obtained in the semiclassical limit directly from the Hartree dynamics. The convergence is established in an abstract sense, namely without control on the rate of convergence. In particular, in \cite{LionsPaul} and \cite{FigalliLigaboPaul}, the authors deal with singular interactions, here included the case of Coulomb potential. Such a result provides the first proof of the derivation of the Vlasov-Poisson Eq. \eqref{eq:VP}
from quantum dynamics. As a limitation of the result, the notion of convergence obtained in \cite{LionsPaul, FigalliLigaboPaul} is very weak. Moreover, the authors make use of compactness techniques, which prevent them to establish explicit bounds on the convergence rates.

The issue of exhibiting explicit bounds has been addressed in \cite{APPP, PP, AKN1, AKN2, GolsePaul1, GolsePulvirentiPaul, GolsePaul2}. More precisely, bounds on the rate of convergence of the Hartree evolution towards the Vlasov equation have been first obtained in \cite{APPP}, where the convergence is established in Hilbert-Schmidt norm with a relative rate $\e^{2/7}$ for regular initial data and smooth potentials $V$. Moreover, for smooth interactions, it has been proven in \cite{PP, AKN1, AKN2} that the solution to the Hartree equation can be written as an expansion - with no control on the remainder - in powers of $\e$. In the same spirit of \cite{APPP}, it has been shown in \cite{BPSS} that the convergence holds in trace norm for mixed states and interaction potentials $V$ such that $\nabla V\in {\rm Lip}(\R^3)$, where ${\rm Lip}(\R^3)$ is the space of Lipschitz functions on $\R^3$. Explicit bounds on the convergence rate are provided. By requiring stronger integrability assumptions on $V$, the authors prove convergence in Hilbert-Schmidt norm for mixed states. Moreover, they prove convergence for the expectation of a class of semiclassical observables at zero temperature (i.e. pure states), thus providing the first rigorous results concerning convergence from the Hartree dynamics towards the Vlasov equation that can be applied to reasonable approximations of ground states.

More recently, a new notion of pseudo-distance reminiscent of the Monge-Kantorovich distance in classical mechanics has been introduced in \cite{GolsePaul1,GolsePaul2}. This new technique allows to substantially relax the assumptions on the interaction potential $V$. More precisely, in \cite{GolsePaul1,GolsePaul2} the convergence of the Hartree dynamics towards the Vlasov Eq. holds with an explicit rate of convergence for a special class of bosonic states (related to T\"{o}plitz operators)  and for interaction potentials $V$ such that $\nabla V\in {\rm Lip}(\R^3)$. In the bosonic setting, the same authors in collaboration with Pulvirenti have shown in \cite{GolsePulvirentiPaul} that the convergence is uniform in the Planck constant.  
\medskip

\noindent {\it Main results, notations and strategy.} The aim of this paper is twofold: on the one hand we extend the results by Lions and Paul \cite{LionsPaul} and  Figalli, Ligab\`o and Paul \cite{FigalliLigaboPaul} for mixed states providing a strong convergence statement (in trace and Hilbert-Schmidt norms) for a certain class of singular interaction potentials, namely $V(x)=\frac{1}{|x|^\alpha}$, for $\alpha\in(0,1/2)$; on the other hand we exhibit explicit bounds on the convergence rate, that is important for applications to real physical systems. Indeed, in real applications the number of particles $N$ is large but finite and the knowledge of explicit convergence rates is crucial to determine how large should $N$ be in order for the Vlasov equation to represent a good approximation of the Hartree dynamics, and therefore of the many-body quantum dynamics given by \eqref{eq:SE}.   

Before stating precisely our main result, we introduce some notations. For $s\in\N$, we denote by $H^s$ the space of real-valued functions $f$ on $\R^3\times\R^3$  such that the Sobolev norm
\be\label{eq:sobolev-norm}
\|f\|_{H^s}^2:=\sum_{|\beta|\leq s}\int |\nabla^\beta f(x,v)|^2dxdv
\ee
is finite. In \eqref{eq:sobolev-norm}, $\beta$ is a multi-index and $\nabla^\beta$ can act on both space and velocity variables.\\
For $s,a\in\N$, we denote by $H^s_a$ the weighted Sobolev space of real-valued functions on $\R^3\times\R^3$ such that
\begin{equation*}
\|f\|_{H^s_a}^2:=\sum_{|\beta|\leq s}\int (1+x^2+v^2)^a|\nabla^\beta f(x,v)|^2 dxdv
\end{equation*}
is finite.\\
For $s\in\N$, we denote by $\mathcal{W}^{s,1}$ the Sobolev space of real-valued functions $f$ on $\R^3\times\R^3$ with $\nabla^\beta f \in L^1$ for all multi-index $\beta$ such that $|\beta|\leq s$. Moreover, we will use the shorthand notation $L^p_x(L^q_v)$ to indicate the space $L^p(\R^3_x;L^q(\R^3_v))$, i.e. the space of functions $f$ on $\R^3\times\R^3$ such that
\begin{equation*}
\|f\|_{L^p_x(L^q_v)}:=\left(\int \left|\int |f(x,v)|^q\,dv\right|^{\frac{p}{q}}\,dx\right)^{\frac{1}{p}}.
\end{equation*}
Lastly, we denote by $\mathcal{H}_t$ the total energy associated with a solution $\widetilde{W}_t$ to the Vlasov Eq. \eqref{eq:VL}. More precisely,
\be\label{eq:energy} 
\mathcal{H}_t:=\frac{1}{2}\iint |v|^2\widetilde{W}_t(x,v)\,dv\,dx+\frac{1}{2}\iint V(x-y)\,\widetilde{\rho}_t(x)\,\widetilde{\rho}_t(y)\,dx\,dy
\ee
is a conserved quantities w.r.t. time, i.e. 
\be\label{eq:energy-vlasov}
\mathcal{H}_t=\mathcal{H}_0
\ee
for every $\widetilde{W}_t$, solution to the Vlasov Eq. \eqref{eq:VL}.

We are now ready to state our main result 
\begin{theorem}\label{thm:trace}
For $\delta\in(0,5/12)$, let $V(x)=1/|x|^{\alpha}$, for $\alpha\in(0, \frac{1}{2}-\frac{6}{5}\delta)$. Let $\omega_N$ be a sequence of reduced densities on $L^2(\R^3)$ with $\tr\ \omega_N=N$, $0\leq \omega_N\leq 1$ such that $\tr\ (-\e^2\Delta)\omega_N\leq CN$, for some positive constant $C$. Let $W_N$, the Wigner transform of $\omega_N,$ satisfy 
the following assumptions uniformly in $N$:
\begin{itemize}
\item[(H1)] $W_N\in L^1\cap L^\infty(\R^3\times\R^3)$ and $\mathcal{H}_0$ finite; 
\item[(H2)] Let $m_0>\frac{3\alpha}{2-\alpha}$. For $m<m_0$ assume there exists a positive constant $C$ such that
\begin{equation*}
\iint |v|^m\,W_N(x,v)\,dx\,dv\leq C;
\end{equation*}
\item[(H3)] For all $R,\,T>0$, $m=0,2$ and $l=0,1,\dots,5$, 
\begin{equation*}\begin{split}
\sup\{(1+x^2+v^2)^m|\nabla^l W_N|(y+tv,w)\,:\,|y-x|\leq &Rt^2,\,|w-v|\leq Rt\}\\
&\in L^\infty((0,T)\times\R^3_x;L^1\cap L^2(\R^3_v))  
\end{split}
\end{equation*}
\item[(H4)] There exist $C>0$ such that for $l=0,1,3$ and $k=0,\dots,6$
\begin{equation*}
\|W_N\|_{H_4^k}\leq C,\quad\quad\quad \|(1+|x|^8+|v|^5)W_N\|_{\mathcal{W}^{l,1}}\leq C.
\end{equation*}
\end{itemize}
We denote by $\omega_{N,t}$ the solution to the Cauchy problem associated with the Hartree equation \eqref{eq:H-sing} with initial data $\omega_N$. Fix $T>0$ and let $\widetilde{W}_{N,t}$, $t\in(0,T)$, be the solution of the Vlasov equation 
\be\label{eq:Valpha}
\partial_t W_t+2v\cdot\nabla_x W_t=\nabla\left(\frac{1}{|\cdot|^\alpha}*\rho_t\right)\cdot\nabla_v W_t,
\ee
with initial data $\widetilde{W}_{N,0}=W_N$. By $\widetilde{\omega}_{N,t}$ we indicate the Weyl quantization of $\widetilde{W}_{N,t}$ as defined in \eqref{eq:Weyl}. 

Then, there exist positive constants $C_i'$, $i=1,\dots,8$, $D_j'$, $j=0,\dots,4$, depending only on the initial data and on $T>0$such that 
\be\label{eq:trace-thm}
\begin{split}
\tr\ |\omega_{N,t}-\widetilde{\omega}_{N,t}|\leq e^{C_1't}\,N\,\e^\frac{2}{5}&\left[C_2'\e^{\frac{3}{5}}+C_3'\e^{\frac{3}{5}\delta}+C_4'\e^{2+\frac{3}{5}\delta}\right.\\
&\ +C_5'\e^{\frac{9}{10}}+C_6'\e^{\frac{7}{5}}+C_7'\e^{\frac{29}{10}}+C_8'\e^{\frac{17}{5}}\\
&\ \left.+D_0'\e^{\frac{3}{5}}+D_1'\e^{\frac{8}{5}}+D_2'\e^{\frac{13}{5}}+D_3'\e^{\frac{18}{5}}+D_4'\e^{\frac{23}{5}}     \right]
\end{split}
\ee
\end{theorem}
Some remarks are in order:

\begin{enumerate}
\item We recall that $\tr\ \omega_{N,t}=\tr\ \widetilde{\omega}_{N,t}=N$. The bound \eqref{eq:trace-thm} is therefore non trivial. Indeed the difference of the reduced densities $\omega_{N,t}$ and $\widetilde{\omega}_{N,t}$ is smaller than their trace norm by a factor which is a fractional power of $\e$.

\item The assumption on the kinetic energy of $\omega_N$, namely $\tr\ (-\e^2\Delta)\omega_N\leq CN$, is needed to ensure $\rho_t\in L^{5/3}(\R^3)$ for all $t\geq 0$. This is shown via kinetic energy estimates and a proof can be found in Appendix A of \cite{S18}.
\item The assumption of the finiteness of the energy $\mathcal{H}_0$ ensures $\widetilde{\rho}_t\in L^{5/3}(\R^3)$, thanks to standard interpolation estimates.
\item Assumptions $(H1),\,(H2)$ and $(H3)$ guarantee the existence and uniqueness of the solution to the Vlasov Eq. \eqref{eq:Valpha}. This follows by a simple adaptation of \cite{LP} to the case of inverse power law potentials $V(x)=\frac{1}{|x|^\alpha}$ (see Appendix \ref{appendix:regularity}). Hypothesis $(H4)$ is needed to prove smoothness properties of the solution obtained from $(H1)-(H3)$. These imply that  Theorem \ref{thm:trace} is expected to hold for fermionic mixed states. Assumptions in $(H4)$ are not optimal. However, looking for minimal assumptions falls out of the scope of this paper. Indeed, even though we would be able to relax the regularity assumptions on $W_N$, we would need at least bounds on $\|W_N\|_{H^2_4}$. This is not enough to consider pure states, for which only one derivative is allowed. 
\item The constants $C_i'$ and $D_j'$ only depend on the bounds in $(H3)$ and $(H4)$. More precisely, $C_1'$ depends on assumptions $(H3)$-$(H4)$ for $l=k=0,1$; $C_2'$ depends on $(H3)$-$(H4)$ for $l,k\leq 3$; $C_3'$ depends on $(H3)$-$(H4)$ for $l,k\leq 3$ and on $\mathcal{H}_0$; $C_4'$ depends on $(H3)$-$(H4)$ for $l,k\leq 5$; $C_5'$ depends on $(H3)$-$(H4)$ for $k\leq 2$; $C_6'$ depends on $(H3)$-$(H4)$ for $k\leq 3$; $C_7'$ depends on $(H3)$-$(H4)$ for $k\leq 4$; $C_8'$ depends on $(H3)$-$(H4)$ for $k\leq 5$; $D_0'$ depends on $(H4)$ for $k\leq 2$; $D_1'$ depends on $(H4)$ for $k\leq 3$; $D_2'$ depends on $(H4)$ for $k\leq 4$; $D_3'$ depends on $(H4)$ for $k\leq 5$; $D_4'$ depends on $(H4)$ for $k\leq 6$. Moreover, all the constants above depend on $T$ at most as $\exp(\exp(\exp T)$. This is due to the Gr\"{o}nwall type argument in the proof combined with estimates on the solution of the Vlasov equation obtained  following on from \cite{LP}. The exact same time dependence can be proven in Theorem \ref{thm:HS} and Theorem \ref{thm:L2}.

\item Examples of initial data satisfying the assumptions of Theorem \ref{thm:trace} can be constructed in the same spirit of Remark 3 in\cite{BPSS}. 

\item Theorem \ref{thm:trace} holds true for repulsive and attractive interactions $$V(x)=\frac{\gamma}{\ \ |x|^{\alpha}},\quad\quad \gamma\in\{+1,-1\}.$$ 
In addition, more general potentials can be considered as soon as they satisfy the assumptions of Proposition \ref{prop:FDLL} below, although they cannot behave worst than $|x|^{-\alpha}$, for $\alpha\in(0,1/2)$, as $x$ is close to $0$.  In particular, we cannot deal with the Coulomb singularity at zero. The main obstruction is that the bound \eqref{eq:trick} is not sharp. This matter will be addressed in the forthcoming paper \cite{S:H-VP}.  
This remark applies also to Theorem \ref{thm:HS} and Theorem \ref{thm:L2} below. 
\end{enumerate}

Let us informally describe the main ingredients of the proof. Our strategy relies on the approach initiated in \cite{BPSS}, where a comparison between the Hartree and the Vlasov dynamics is performed at the level of the Weyl quantization. More precisely, we consider the Weyl quantization of the solution to the Vlasov equation \eqref{eq:Valpha} with initial data $W_N$ and compare it with a solution to the Cauchy problem associated to the Hartree equation with the same initial data $W_N$. The main difficulty in our analysis is the singularity of the interaction potential. Indeed, we recall that in \cite{BPSS} $V$ is at least such that $\nabla V\in {\rm Lip}(\R^3)$. The singularity of the potential has been already an issue in \cite{LionsPaul}, \cite{FigalliLigaboPaul} and here the situation is made even more complicated by the fact that we cannot rely on compactness methods, since we look for explicit bounds on the rate of convergence.\\
To tackle this problem, a key ingredient is the use of a generalized version of Feffermann-de la Llave representation formula, which allows to rewrite a radially symmetric function $V$ with certain decay and regularity properties at infinity (see Proposition \ref{prop:FDLL}) as
\be\label{eq:GFDL}
V(x)=\int_0^\infty dr\,g(r)\,\chi_{(r/2)}*\chi_{(r/2)}(x),
\ee
where $\chi_{(r)}(x)$ is the characteristic function of a ball of radius $r$ centred at the origin, the symbol $*$ denotes the usual convolution operation on $\R^d$ and $g$ is a function that can be explicitly computed in terms of higher order derivatives of $V$.   
Eq. \eqref{eq:GFDL} has been first introduced by Feffermann and de la Llave in \cite{FDLL} for the Coulomb potential $V(x)=1/|x|$, and then generalized by Hainzl and Seiringer in \cite{HS}. \\
Through \eqref{eq:GFDL}, the l.h.s. of \eqref{eq:trace-thm} can be estimated by a sum of terms, eventually leading to a Gr\"{o}nwall type estimate. The most relevant aspect of such a rewriting consists in isolating the singularity of the interaction and making clear that it can be faced by requiring enough regularity on the solution to the Vlasov equation and thus on its initial data. To characterise the class of admissible initial data, we rely on an adaptation of the well known result by Lions and Perthame on the well-posedness of the Vlasov-Poisson system \cite{LP}. In our context, the Vlasov system can be written as
\be\label{eq:V-P-alpha}
\left\{\begin{array}{l}
\partial_t W_t(x,v)+v\cdot\nabla_x W_t(x,v)+E(t,x)\cdot\nabla_v W_t(x,v)=0,\\\\
E(t,x)=\nabla(\frac{\gamma}{|  \cdot  |^\alpha}*{\rho}_t)(x)\,,\\\\
\rho_t(x)=\int W_t(x,v)\,dv
\end{array}\right.
\ee
where $\gamma\in\{+1,\,-1\}$ models the repulsive or attractive nature of the interaction.
Notice that the second equation in \eqref{eq:V-P-alpha} can be rewritten as a modified Poisson equation in which the Laplace operator is replaced by the $p$-Laplacian, namely $\Delta^p$ with $p=(3-\alpha)/2$, which is clearly bigger than 1 if $\alpha\in(0,1/2)$. We underline that the most physical interesting case is the one of the Vlasov-Poisson equation, in which $p=1$, achieved when $\alpha=1$. Unfortunately, the techniques developed in this paper do not allow so far to consider $\alpha\geq 1/2$. Nevertheless, Theorem \ref{thm:trace} is, up to our knowledge, the only result proving strong convergence for potentials $V$ such that $\nabla V\notin {\rm Lip}(\R^3)$. 

Albeit Theorem \ref{thm:trace} provides trace norm convergence with explicit bounds on the reduced density $\omega_{N,t}$ towards $\widetilde{\omega}_{N,t}$, it does not give any information on the convergence of the Wigner transform $W_{N,t}$ towards a solution to the Vlasov Eq. \eqref{eq:Valpha} $\widetilde{W}_{N,t}$. This is due to the fact that there is no equivalence between the trace norm of a reduced density and the $L^1$ norm of its Wigner transform, as it is instead the case for the Hilbert-Schmidt norm of a reduced density and the $L^2$ norm of its Wigner transform. Motivated by this observation, we provide Hilbert-Schmidt norm convergence of the reduced density $\omega_{N,t}$ towards $\widetilde{\omega}_{N,t}$ in the following  Theorem, whose proof strongly relies on the result obtained in Theorem \ref{thm:trace}.

\begin{theorem}\label{thm:HS}
For $\delta\in(0,5/12)$, let $V(x)=1/|x|^{\alpha}$, for $\alpha\in(0, \frac{1}{2}-\frac{6}{5}\delta)$.
Let $\omega_N$ a sequence of reduced density matrices on $L^2(\R^3)$ such that $\tr\ \omega_N=N$, $0\leq \omega_N\leq 1$, $\tr\ (-\e^2\Delta)\omega_N\leq CN$, for some $C>0$, and its Wigner transform $W_N$ satisfies $(H1)-(H5)$ uniformly in $N$.\\
For $T>0$, denote by $\omega_{N,t}$ the solution of the Hartree Eq. with initial data $\omega_N$ and by $\widetilde{\omega}_{N,t}$ the Weyl quantization of the solution $\widetilde{W}_{N,t}$ of the Vlasov Eq. \eqref{eq:Valpha} with initial data $W_N$ in $[0,T]$. \\
Then, there exist positive constants $C_i'$, $D_j'$, $E_i'$, $i=1,\dots,8$, $j=0,\dots,4$, depending on the assumptions on the initial data $W_N$ and on $T>0$, such that 
\be\label{eq:HS-thm}
\begin{split}
\|\omega_{N,t}-\widetilde{\omega}_{N,t}\|_{\rm HS}\leq e^{C_1't}\,\sqrt{N}\,\e^{\frac{2}{5}}&\left[ C_2'\e^{\frac{3}{5}}+C_3'\e^{\frac{3}{5}\delta}+C_4'\e^{2+\frac{3}{5}\delta}+C_5'\e^{\frac{9}{10}}+C_6'\e^{\frac{7}{5}}+C_7'\e^{\frac{29}{10}}+C_8'\e^{\frac{17}{5}}\right.\\
&\quad+D_0'\e^{\frac{3}{5}}+D_1'\e^{\frac{8}{5}}+D_2'\e^{\frac{13}{5}}+D_3'\e^{\frac{18}{5}}+D_4'\e^{\frac{23}{5}}\\
 &\quad+E_1\e^{4+\frac{3}{5}\delta}+E_2\e^{\frac{49}{10}}+E_3\e^{\frac{27}{5}}+E_4\e^{\frac{13}{5}}\\
 &\quad\left.+E_5'\e^{\frac{18}{5}}+E_6'\e^\frac{23}{5}+E_7'\e^\frac{28}{5}+E_8'\e^\frac{33}{5}
     \right]
\end{split}
\ee   
\end{theorem}
\begin{remark}
Analogously to \cite{BPSS}, it is possible to relax the assumptions on the initial data asking for $(H3)$ to hold true only for $l\leq 3$ and assuming $(H4)$ only for $k\leq 3$. However, this procedure is rather technical and does not allow neither to state Theorem \ref{thm:HS} in a simpler way getting rid of   higher order corrections nor to consider pure states. 
\end{remark}

The Hilbert-Schmidt norm has the advantage of being proportional to the $L^2$ norm at the level of the Wigner transform. Therefore Theorem \ref{thm:HS} can be used to prove the following

\begin{theorem}\label{thm:L2}
For $\delta\in(0,5/12)$, let $V(x)=1/|x|^{\alpha}$, for $\alpha\in(0, \frac{1}{2}-\frac{6}{5}\delta)$. Let $\omega_N$ be a sequence of reduced densities on $L^2(\R^3)$ with $\tr\omega_N=N$, $0~\leq~\omega_N\leq 1$, $\tr\ (-\e^2\Delta)\omega_N\leq CN$, for some $C>0$, and denote by $W_N$ the Wigner transform of $\omega_N$. Assume $(H1)-(H5)$ and let $W_0$ be a probability density on $\R^3\times\R^3$ such that 
$$\|W_N-W_0\|_{L^1}\leq C\,k_{N,1},\quad\quad\quad \|W_N-W_0\|_{L^2}\leq C\,k_{N,2}$$
where $k_{N,j}$ are nonnegative sequences such that $k_{N,j}\to 0$ as $N\to\infty$, $j=1,2$.
Then, the Wigner transform $W_{N,t}$ of the solution $\omega_{N,t}$ to the Hartree Eq. \eqref{eq:HF} with initial data $\omega_N$ converges in $L^2(\R^3\times\R^3)$ to $W_t$, solution to the Vlasov Eq. \eqref{eq:Valpha} with initial data $W_0$, i.e. 
\be\label{eq:L2-thm}
\begin{split}
\|{W}_{N,t}-W_{t}\|_{L^2}\leq &e^{C_1't}\,\sqrt{N}\,\e^{\frac{2}{5}}\left[ C_2'\e^{\frac{3}{5}}+C_3'\e^{\frac{3}{5}\delta}+C_4'\e^{2+\frac{3}{5}\delta}+C_5'\e^{\frac{9}{10}}+C_6'\e^{\frac{7}{5}}+C_7'\e^{\frac{29}{10}}+C_8'\e^{\frac{17}{5}}\right.\\
&\quad\quad\quad\quad\quad+D_0'\e^{\frac{3}{5}}+D_1'\e^{\frac{8}{5}}+D_2'\e^{\frac{13}{5}}+D_3'\e^{\frac{18}{5}}+D_4'\e^{\frac{23}{5}}+E_1\e^{4+\frac{3}{5}\delta}\\
&\quad\quad\quad\quad\quad\left.+E_2\e^{\frac{49}{10}}+E_3\e^{\frac{27}{5}}+E_4\e^{\frac{13}{5}}
 +E_5'\e^{\frac{18}{5}}+E_6'\e^\frac{23}{5}+E_7'\e^\frac{28}{5}+E_8'\e^\frac{33}{5}
     \right]\\
 &+C(k_{N,1}+k_{N,2})
\end{split}
\ee
where $C_i'$, $D_j'$, $E_i'$, $i=1,\dots,8$, $j=0,\dots,4$,  are constants depending on the assumptions on the initial data $W_N$ and on $T>0$.   
\end{theorem}

We underline that Theorem \ref{thm:HS} uses strongly the trace norm bound \eqref{eq:trace-thm}. As already mentioned, there is no simple relation between the trace norm of a reduced density and the $L^1$ norm of its Wigner transform. This is why working at the level of Weyl quantization, instead of looking at the Wigner transform, is so crucial here and represents the strength of the approach adopted in this paper.  In other words, there is no simple way of reproducing estimate \eqref{eq:L2-thm} directly at the level of the Wigner transform, without passing through the Weyl quantization. 

The paper is organised as follows. Section \ref{section:trace} is devoted to the proof of Theorem \ref{thm:trace}. In particular, in Subsection \ref{subsect:notations} some useful notations are given, in Subsection \ref{subsect:dominant} the term containing the dominant part is estimated, while the subleading term  is analysed in Subsection \ref{subsect:subleading}.  Subsection \ref{subsect:end} gathers together the estimates from Subsections \ref{subsect:dominant} and \ref{subsect:subleading} to prove convergence of the Hartree dynamics towards the Vlasov Eq. at the level of Weyl quantization in  trace norm. Section \ref{section:HS} heavily makes use of Theorem \ref{thm:trace} to establish convergence of the Hartree dynamics towards the Vlasov Eq. in Hilbert-Schmidt norm, thus leading to the proof of Theorem \ref{thm:HS}. In Section \ref{section:L2} we give the proof of Theorem \ref{thm:L2}. Lastly Appendix \ref{appendix:regularity} reviews the theory of existence, uniqueness and regularity for the Vlasov Eq., adapting the well-know result by Lions and Perthame \cite{LP} to the case of singular interaction $V(x)=\frac{1}{|x|^\alpha}$, for $\alpha\in(0,1/2)$, while Appendix \ref{ap:int} contains some estimates on Gaussian integrals used throughout the paper. 


\section{Trace norm convergence: proof of Theorem \ref{thm:trace}}\label{section:trace}

This section is devoted to prove Theorem \ref{thm:trace}.  
To this end, we will make use of the following simple Lemma
\begin{lemma}\label{lem:est-tr}
Let $\omega_{N,t}$ be a solution to the Hartree Eq. with initial datum $\omega_N$
\be\label{eq:HF}
i\,\e\,\partial_t\omega_{N,t}=\left[-\e^2\Delta+\frac{1}{|\cdot|^\alpha}*\rho_t\,,\,\omega_{N,t}\right]
\ee
for some $\alpha\in (0,1/2)$.
Let $\widetilde{\omega}_{N,t}$ be the Weyl transform of $\widetilde{W}_{N,t}$, solution to the Vlasov Eq. \eqref{eq:Valpha} with initial data $W_N$, the Wigner transform of $\omega_N$. Then
\be\label{eq:est-tr}
\tr\ |\omega_{N,t}-\tilde{\omega}_{N,t}|\leq\frac{1}{\e}\int_0^t \tr\ \left|\left[ \frac{1}{|\cdot|^{\alpha}}*(\rho_s-\tilde{\rho}_s),\tilde{\omega}_{N,s} \right]\right|\,ds +\frac{1}{\e}\int_0^t \tr\ |B_s|\,ds
\ee
where, for every $s\in [0,t]$, $B_s$ is the operator with kernel
\be\label{eq:def-B}
B_s(x;y)=\left[\left(\frac{1}{|\cdot|^{\alpha}}*\widetilde{\rho}_s\right)(x)-\left(\frac{1}{|\cdot|^{\alpha}}*\widetilde{\rho}_s\right)(y)-\nabla\left(\frac{1}{|\cdot|^{\alpha}}*\widetilde{\rho}_s\right)\left(\frac{x+y}{2}\right)\cdot(x-y)\right]\,\widetilde{\omega}_{N,t}(x;y)\,.
\ee
\end{lemma}
\begin{proof}
Straightforward computations show that $\widetilde{\omega}_{N,t}$ solves 
\be\label{eq:vlasov-weyl}
i\,\e\,\partial_t\,\widetilde{\omega}_{N,t}=[-\e^2\,\Delta,\widetilde{\omega}_{N,t}]+A_t
\ee
where $A_t$ is the operator with integral kernel
\begin{equation*}
A_t(x;y)=\nabla\left(\frac{1}{|\cdot|^{\alpha}}*\widetilde{\rho}_t\right)\left(\frac{x+y}{2}\right)\cdot(x-y)\,\widetilde{\omega}_{N,t}(x;y)\,.
\end{equation*}
To compare $\omega_{N,t}$ and $\widetilde{\omega}_{N,t}$, we introduce the two-parameter group of unitary transformations $\mathcal{U}(t;s)$ generated by the Hartree Hamiltonian $h(t):=-\e^2\Delta+\frac{1}{\ |\cdot|^\alpha}*\rho_t$
\begin{equation*}
\left\{
\begin{array}{l}
i\,\e\,\partial_t\,\mathcal{U}(t;s)=h(t)\,\mathcal{U}(t;s)\\
\mathcal{U}(s;s)=1 
\end{array}
\right.
\end{equation*}
Notice that $\mathcal{U}(t;s)$ is such that $\omega_{N,t}=\mathcal{U}(t;0)\,\omega_{N}\,\mathcal{U}^*(t;0)$. In order to  get rid of the kinetic part of the Hamiltonian $h(t)$, we conjugate the difference between $\omega_{N,t}$ and $\widetilde{\omega}_{N,t}$ with $\mathcal{U}(t;0)$ and perform the time derivative. This leads to
\begin{equation*}
\begin{split}
i\,\e\,\partial_t\,(\mathcal{U}^*(t;0)\,(\omega_{N,t}-\widetilde{\omega}_{N,t})\,\mathcal{U}(t;0))&=\mathcal{U}^*(t;0)\,[h(t),\omega_{N,t}-\widetilde{\omega}_{N,t}]\,\mathcal{U}(t;0)\\
&+\mathcal{U}^*(t;0)\,([h(t),\omega_{N,t}]-[-\e^2\Delta,\widetilde{\omega}_{N,t}]-A_t)\,\mathcal{U}(t;0)\\
&=\mathcal{U}^*(t;0)\,\left(\left[\frac{1}{|\cdot|^{\alpha}}*\rho_t,\widetilde{\omega}_{N,t} \right]-A_t\right)\,\mathcal{U}(t;0)\\
&=\mathcal{U}^*(t;0)\,\left(\left[\frac{1}{|\cdot|^{\alpha}}*(\rho_t-\widetilde{\rho}_t),\widetilde{\omega}_{N,t}\right]+B_t\right)\,\mathcal{U}(t;0)\,
\end{split}
\end{equation*}
where $B_t$ is defined through its kernel as in \eqref{eq:def-B}. Recalling that $\widetilde{\omega}_{N,0}=\omega_N$, Duhamel's formula yields
 \begin{equation}\label{eq:omega-omegatilde}
 \begin{split}
\mathcal{U}^*(t;0)\,(\omega_{N,t}-\widetilde{\omega}_{N,t})\,\mathcal{U}(t;0)&=\frac{1}{i\,\e}\int_0^t \mathcal{U}^*(t;s)\,\left[\frac{1}{|\cdot|^{\alpha}}*(\rho_s-\widetilde{\rho}_s),\widetilde{\omega}_{N,s}\right]\,\mathcal{U}(t;s)\,ds\\
&+\frac{1}{i\,\e}\int_0^t \mathcal{U}^*(t;s)\,B_s\,\mathcal{U}(t;s)\,ds
\end{split}
 \end{equation}
 Taking the trace norm in the above expression and recalling that $\mathcal{U}(t;s)$ is a family of unitary operators, the Lemma is proved.  
\end{proof}
In order to prove Theorem \ref{thm:trace}, we will bound the r.h.s. in \eqref{eq:est-tr} and conclude by Gr\"{o}nwall's Lemma. From $\rho_s-\widetilde{\rho}_s$ in the first term on the r.h.s. we extract the Gr\"{o}nwall term $\omega_{N,s}-\widetilde{\omega}_{N,s}$, while the second term turns out to give subleading contributions. To deal with the singularity of the potential, we will use the properties of $\tilde{\omega}_{N,s}$, which translate into regularity assumptions on its Wigner transform $\widetilde{W}_{N,s}$, solution to the Vlasov Eq. \eqref{eq:Valpha}. 

\subsection{Notations} \label{subsect:notations}

We start by introducing some notations that will help us to shorten the exposition. \\
We underline  that the notation $C$ will refer to constants, possibly depending on $T>0$. When a constant depends on another parameters, we state it explicitly. 

We define the operators $\mathcal{J}_k$ for $k=1,2,3$ through their kernels $\mathcal{J}_k(x;x')$ for $x,x'\in\R^3$, as follows:
\begin{equation}\label{eq:J1-J2-J3}
\begin{split}
\mathcal{J}_1(x;x')&:=-N\int_0^1 ds\,(1+x^2)\,\chi_{(r/\sqrt{s},y)}(x)\,\chi_{(r/\sqrt{1-s},y)}(x')\,\left[\frac{(x-y)}{r^2}+\frac{(x'-y)}{r^2}\right]\cdot (x-x')\\
&\hspace{7.5cm}\times\int dv\,\widetilde{W}_{N,s}\left(\frac{x+x'}{2},v\right)\,e^{iv\cdot\frac{(x-x')}{\e}} \\
\mathcal{J}_2(x;x')&:= -N\e^2\,(1+x^2)\,\int ds\,\Delta_x\left\{ \chi_{(r/\sqrt{s},y)}(x)\,\chi_{(r/\sqrt{1-s},y)}(x') \right.   \\
&\hspace{3cm}\left. \times\frac{(x-y)}{r^2}\cdot (x-x')\int dv\,\widetilde{W}_{N,s}\left(\frac{x+x'}{2},v\right)\,e^{iv\cdot\frac{(x-x')}{\e}}   \right\}\\
\mathcal{J}_3(x;x')&:= -N\e^2\,(1+x^2)\,\int ds\,\Delta_x\left\{ \chi_{(r/\sqrt{s},y)}(x)\,\chi_{(r/\sqrt{1-s},y)}(x') \right.   \\
&\hspace{3cm}\left. \times\frac{(x'-y)}{r^2}\cdot (x-x')\int dv\,\widetilde{W}_{N,s}\left(\frac{x+x'}{2},v\right)\,e^{iv\cdot\frac{(x-x')}{\e}}   \right\}.\\
\end{split}
\end{equation} 
Moreover, we will denote by $\mathcal{J}_{2,j}$ for $j=1,\dots,5$ and $\mathcal{J}_{3,i}$ for $i=1,\dots,4$ the operators whose associated kernels have absolute values defined by
\be\label{eq:J2}
\begin{split}
|\mathcal{J}_{2,1}(x;x')|&:= \frac{N\e}{r}\int_0^1 \frac{ds}{\sqrt{s}}\,(1+x^2)\,\chi_{(r/\sqrt{s},y)}(x)\frac{(x-y)}{r/\sqrt{s}}\chi_{(r/\sqrt{1-s},y)}(x')\\
&\quad\quad\quad\times\int dv\,\left[2v\,\widetilde{W}_{N,s}\left(\frac{x+x'}{2},v\right)+v^2\nabla_v\widetilde{W}_{N,t}\left(\frac{x+x'}{2},v\right)\right]\,e^{iv\cdot\frac{(x-x')}{\e}}\\
|\mathcal{J}_{2,2}(x;x')|&:=\frac{4\,N\e^2}{r^2}\int_0^1 ds\,(1+x^2)\,\chi_{(r/\sqrt{s},y)}(x)\chi_{(r/\sqrt{1-s},y)}(x')\\
&\quad\quad\quad\quad\times\left[\frac{(x-y)^3}{r^3/\sqrt{s^3}}+2\frac{(x-y)}{r/\sqrt{s}}\right]\cdot\frac{(x-x')}{r/\sqrt{s}}\int dv\,\widetilde{W}_{N,t}\left(\frac{x+x'}{2},v\right)\,e^{iv\cdot\frac{(x-x')}{\e}}\\
|\mathcal{J}_{2,3}(x;x')|&:=\frac{4\,N\e^2}{r^2}\int_0^1 ds\,(1+x^2)\,\chi_{(r/\sqrt{s},y)}(x)\frac{(x-y)^2}{r^2/s}\chi_{(r/\sqrt{1-s},y)}(x')\,\int dv\,e^{iv\cdot\frac{(x-x')}{\e}}\\
&\quad\quad\quad\quad\times\left[\widetilde{W}_{N,t}\left(\frac{x+x'}{2},v\right)+\e\nabla^2_{v,x}\widetilde{W}_{N,t}\left(\frac{x+x'}{2},v\right)+v\cdot\nabla_v\widetilde{W}_{N,t}\left(\frac{x+x'}{2},v\right)\right]\,\\
|\mathcal{J}_{2,4}(x;x')|&:=\frac{2\,N\e^2}{r^2}\int_0^1 ds\,(1+x^2)\,\chi_{(r/\sqrt{s},y)}(x)\chi_{(r/\sqrt{1-s},y)}(x')\,\int dv\,e^{iv\cdot\frac{(x-x')}{\e}}\\
&\quad\quad\quad\quad\times\left[\widetilde{W}_{N,t}\left(\frac{x+x'}{2},v\right)+\e\nabla^2_{v,x}\widetilde{W}_{N,t}\left(\frac{x+x'}{2},v\right)+v\cdot\nabla_v\widetilde{W}_{N,t}\left(\frac{x+x'}{2},v\right)\right]\,\\
|\mathcal{J}_{2,5}(x;x')|&:=\frac{4\,N\e^2}{r}\int_0^1 \frac{ds}{\sqrt{s}}\,(1+x^2)\,\chi_{(r/\sqrt{s},y)}(x)\frac{(x-y)}{r/\sqrt{s}}\chi_{(r/\sqrt{1-s},y)}(x')\,\int dv\,e^{iv\cdot\frac{(x-x')}{\e}}\\
&\quad\quad\quad\times\left[\nabla_x\widetilde{W}_{N,s}\left(\frac{x+x'}{2},v\right)+\frac{\e}{2}\nabla^3_{v,x,x}\widetilde{W}_{N,t}\left(\frac{x+x'}{2},v\right)+v\nabla^2_{v,x}\widetilde{W}_{N,t}\left(\frac{x+x'}{2},v\right)\right]\,\\
\end{split}
\ee
\be\label{eq:J3}
\begin{split}
|\mathcal{J}_{3,1}(x;x')|&:=\frac{N\e}{r}\left|\int_0^1 \frac{ds}{\sqrt{1-s}}\,(1+x^2)\,\chi_{(r/\sqrt{s},y)}(x)\,\chi_{(r/\sqrt{1-s},y)}(x')\,\frac{(x'-y)}{r/\sqrt{1-s}}\int dv\,e^{iv\cdot\frac{(x-x')}{\e}}\right.\\
&\quad\quad\quad\left.\times\left[2v\,\widetilde{W}_{N,t}\left(\frac{x+x'}{2},v\right)+\nabla_{v}\left(v^2\,\widetilde{W}_{N,t}\left(\frac{x+x'}{2},v\right)\right)\right]\,\right|\\
|\mathcal{J}_{3,2}(x;x')|&:=\frac{N\e^2}{r^2}\left|\int_0^1 ds\,2s\,(1+x^2)\,\chi_{(r/\sqrt{s},y)}(x)\left[2\frac{|x-y|^2}{r^2/s}+1\right]\,\chi_{(r/\sqrt{1-s},y)}(x')\,\frac{(x'-y)}{r}\cdot\frac{(x-x')}{r}\right.\\
&\quad\quad\quad\quad\left.\times\int dv\,\widetilde{W}_{N,t}\left(\frac{x+x'}{2},v\right)\,e^{iv\cdot\frac{(x-x')}{\e}}\right|\\
|\mathcal{J}_{3,3}(x;x')|&:=\frac{4\,N\e^2}{r^2}\left|\int_0^1 ds\,\frac{\sqrt{s}}{\sqrt{1-s}}\,(1+x^2)\,\chi_{(r/\sqrt{s},y)}(x)\frac{(x-y)}{r/\sqrt{s}}\,\chi_{(r/\sqrt{1-s},y)}(x')\,\frac{(x'-y)}{r/\sqrt{1-s}}\int dv\,e^{iv\cdot\frac{(x-x')}{\e}}\right.\\
&\quad\quad\quad\left.\times\left[\widetilde{W}_{N,t}\left(\frac{x+x'}{2},v\right)+\e\nabla^2_{v,x}\widetilde{W}_{N,t}\left(\frac{x+x'}{2},v\right)+\nabla_{v}\left(v\,\widetilde{W}_{N,t}\left(\frac{x+x'}{2},v\right)\right)\right]\,\right|\\
|\mathcal{J}_{3,4}(x;x')|&:=\frac{2\,N\e^2}{r}\left|\int_0^1\frac{ds}{\sqrt{1-s}}\,(1+x^2)\,\chi_{(r/\sqrt{s},y)}(x)\,\chi_{(r/\sqrt{1-s},y)}(x')\,\frac{(x'-y)}{r/\sqrt{1-s}}\int dv\,e^{iv\cdot\frac{(x-x')}{\e}}\right.\\
&\quad\quad\quad\left.\times\left[\nabla_x\widetilde{W}_{N,t}\left(\frac{x+x'}{2},v\right)+\frac{\e}{2}\nabla^3_{v,x,x}\widetilde{W}_{N,t}\left(\frac{x+x'}{2},v\right)+\nabla^2_{v,x}\left(v\,\widetilde{W}_{N,t}\left(\frac{x+x'}{2},v\right)\right)\right]\,\right|\\
\end{split}
\ee

\subsection{Estimates on the dominant term}\label{subsect:dominant}

We give the bound on the first term in the r.h.s. of \eqref{eq:est-tr} in the following Proposition

\begin{proposition}\label{prop:1term-tr}
Under the same assumptions of Lemma \ref{lem:est-tr} and Theorem \ref{thm:trace}, it holds
\be\label{eq:1term-tr}
\begin{split}
\tr\ \left|\left[\frac{1}{|\cdot|^\alpha}*(\rho_s-\widetilde{\rho}_s)\,,\,\widetilde{\omega}_{N,s}\right]\right|\leq &\e\,C_1\,\tr\ |\omega_{N,s}-\widetilde{\omega}_{N,s}| \\
&+C_2\,N\,\e+ C_3\,N\,\e^{\frac{7}{5}+\frac{3}{5}\delta}+C_4\,N\,\e^{\frac{17}{5}+\frac{3}{5}\delta}\\
&+C_5\,N\,\e^{\frac{23}{10}}+C_6\,N\,\e^{\frac{14}{5}}+C_7\,N\,\e^{\frac{43}{10}}+C_8\,N\,\e^{\frac{24}{5}}
\end{split}
\ee
where $C_i$, $i=1,\dots,8$ are positive constants depending on weighted Sobolev norms of $\widetilde{W}_{N,s}$ as follows: 
\begin{equation*}
\begin{split}
C_1&=C_1\left(\|(1+x^2+v^2)\nabla^j\widetilde{W}_{N,s}\|_{L^2_v(L^\infty_x)},\,\|\widetilde{W}_{N,s}\|_{H_2^j}\right)\quad  j=0,1;
\\ 
C_2&=C_2\left(\|v^2\nabla^j\widetilde{W}_{N,s}\|_{L^2_v(L^\infty_x)},\,\|\widetilde{W}_{N,s}\|_{H_2^j}\right)\quad  j\leq 3;
\\
C_3&=C_3\left(\|\rho_s\|_{L^{5/3}},\,\|\widetilde{\rho}_s\|_{L^{5/3}},\,\|\widetilde{W}_{N,s}\|_{\mathcal{W}^{j,1}_8},\,\|\widetilde{W}_{N,s}\|_{H_2^j},\,\|(1+x^2+v^2)\nabla^j\widetilde{W}_{N,s}\|_{L^2_v(L^\infty_x)}\right)\quad j\leq 3;
\\
C_4&=C_4\left(\|\rho_s\|_{L^{5/3}},\,\|\widetilde{\rho}_s\|_{L^{5/3}},\,\|\widetilde{W}_{N,s}\|_{\mathcal{W}^{j,1}},\,\|\widetilde{W}_{N,s}\|_{H^j},\,\|\nabla^j\widetilde{W}_{N,s}\|_{L^2_v(L^\infty_x)}\right)\quad  j\leq 5;
\\ 
C_5&=C_5\left(\|\rho_s\|_{L^{5/3}},\,\|\widetilde{\rho}_s\|_{L^{5/3}},\,\|\widetilde{W}_{N,s}\|_{H_2^j}\right)\quad  j\leq 2;
\\
C_6&=C_6\left(\|\rho_s\|_{L^{5/3}},\,\|\widetilde{\rho}_s\|_{L^{5/3}},\,\|\widetilde{W}_{N,s}\|_{H_2^j}\right)\quad  j\leq 3;
\\
C_7&=C_7\left(\|\rho_s\|_{L^{5/3}},\,\|\widetilde{\rho}_s\|_{L^{5/3}},\,\|\widetilde{W}_{N,s}\|_{H_2^j}\right)\quad  j\leq 4;\\
C_8&=C_8\left(\|\rho_s\|_{L^{5/3}},\,\|\widetilde{\rho}_s\|_{L^{5/3}},\,\|\widetilde{W}_{N,s}\|_{H_2^j}\right)\quad  j\leq 5.
\end{split}
\end{equation*}
\end{proposition}

The above Proposition relies on the following Lemma, whose proof is postponed at the end of the proof of Proposition \ref{prop:1term-tr}. 
 
\begin{lemma}\label{lemma:fund-est}
For every $r>0$, $y,z\in\R^3$, denote by $$\chi_{r,y}(z):=\exp\{-|y-z|^2/r^2\}$$
Under the same assumptions of Proposition \ref{prop:1term-tr}, there exists a constant $C>0$, such that
\be\label{eq:fund-est}
\|(1+x^2)\,(1-\e^2\Delta)\,[\chi_{(r,y)},\widetilde{\omega}_{N,s}]\|_{\rm HS}\leq C\sum_{j=1}^3\|\mathcal{J}_j\|_{\rm HS} 
\ee
where $\mathcal{J}_1,\ \mathcal{J}_2$ and $\mathcal{J}_3$ have been defined in \eqref{eq:J1-J2-J3}.
\end{lemma}
\medskip 

We recall here the generalization to radially symmetric functions of the Fefferman - de la Llave representation formula established in \cite{FDLL} for the Coulomb potential. This generalization has been proposed by Hainzl and Seiringer in \cite{HS} and we will make use of it in the proof of Proposition \ref{prop:1term-tr}.

\begin{proposition}[Theorem 1 in \cite{HS}]\label{prop:FDLL}
For $n\geq 2$, let $V:\R^n\to\R$ be a radial function that is $[n/2]+2$ times differentiable away from $x=0$. For $m\in\N_0$ denote $V^{(m)}(|x|)=d^m/d|x|^m\,V(x)$. Assume that $\lim_{|x|\to\infty}|x|^mV^{(m)}(|x|)=0$ for all $0\leq m\leq [n/2]+1$ and let $\chi_r(x)={\bf 1}_{\{|x|\leq r\}}$. Then
\be\label{eq:FDLL}
V(x)=\int_0^\infty dr\,g(r)\,\chi_{r/2}*\chi_{r/2}(x)
\ee
where 
\begin{equation*}
\begin{split}
g(r)=\frac{(-1)^{[n/2]}}{\Gamma\left(\frac{n-1}{2}\right)}\frac{2}{(\pi\,r^2)^{(n-1)/2}}&\left( \int_r^\infty ds\,V^{([n/2]+2)}(s)\,\left(\frac{d}{ds}\right)^{n-1-[n/2]}s(s^2-r^2)^{\frac{1}{2}(n-3)}\right.\\
&\quad\quad\quad\quad\quad\quad+\left.\delta_{\rm odd}V^{([n/2]+2)}(r)\,r(2r)^{\frac{1}{2}(n-3)}\Gamma\left(\frac{n-1}{2}\right)\right)
\end{split}
\end{equation*}
where $$\delta_{\rm odd}=\left\{\begin{array}{ll} 1, & n \mbox{ \rm odd, }\\  0, & n \mbox{ \rm even. } \end{array}\right.$$
\end{proposition}
\begin{remark}
Inverse power law potentials $V(x)=1/|x|^\alpha$,  $\alpha>0$, obviously verify the assumptions in Proposition \ref{prop:FDLL} and the representation formula \eqref{eq:FDLL} takes a very simple coincise form in this case (see \eqref{eq:FDL}). Notice however that Proposition \ref{prop:FDLL} applies to more general functions, hence our result can be generalized to all such interaction potentials which fullfill the hypotheses of Proposition \ref{prop:FDLL} and do not exhibit a singularity at $x=0$ worst than  $|x|^{-k}$ for $k\in(0,1/2)$. 
\end{remark}
\begin{proof}[Proof of Proposition \ref{prop:1term-tr}]
We are now concerned with the estimate of the first term in the r.h.s. of \eqref{eq:est-tr}.
%
%
To deal with this term, we use a smooth version of the generalized Fefferman-de la Llave representation formula for radial potentials in Proposition \ref{prop:FDLL}
\be\label{eq:FDL}
\frac{1}{|x-y|^\alpha}=\frac{4}{\pi^2}\int_0^\infty \frac{1}{r^{4+\alpha}}\int_{\R^3}\chi_{(r,z)}(x)\,\chi_{(r,z)}(y)\,dz\,dr\,,
\ee
where the characteristic function ${\bf 1}_{\{|x-z|\leq r\}}$ is replaced by $\chi_{(r,z)}(\cdot)=e^{-|\cdot\,-z|^2/r^2}$.

We plug \eqref{eq:FDL} in the convolution term, thus obtaining 
\begin{equation}\label{eq:convolution}
\frac{1}{|\cdot|^\alpha}*(\rho_s-\tilde{\rho}_s)(x)=\frac{4}{\pi^2}\int_0^\infty\iint \frac{1}{r^{4+\alpha}}\,\chi_{(r,z)}(x)\,\chi_{(r,z)}(y)\,(\rho_s(y)-\tilde{\rho}_s(y))\,dz\,dy\,dr\,.
\end{equation} 
We recall that
\begin{equation*}
\int \chi_{(r,z)}(x)\,\chi_{(r,z)}(y)\,dz=C\,r^3\,\chi_{(r,y)}(x)
\end{equation*}
for some positive constant $C$. Thus the integral in the $z$ variable absorbs part of the power of $r$ and \eqref{eq:convolution} reduces to
\begin{equation*}
\frac{1}{|\cdot|^\alpha}*(\rho_s-\tilde{\rho}_s)(x)=C\int_0^\infty\int \frac{1}{r^{1+\alpha}}\,\chi_{(r,x)}(y)\,(\rho_s(y)-\tilde{\rho}_s(y))\,dy\,dr\,,
\end{equation*} 
for some positive constant $C$.\\
Therefore, we obtain the following expression for the kernel of the commutator in \eqref{eq:1term-tr}
\be
\begin{split}
&\left[\frac{1}{|\cdot|^\alpha}*(\rho_s-\tilde{\rho}_s)\,,\,\tilde{\omega}_{N,s}\right](x;x')\\ 
&\quad\quad =
C\int_0^\infty \int \frac{1}{r^{1+\alpha}}\,(\rho_s(y)-\tilde{\rho}_s(y))\,[\chi_{(r,y)}\,,\,\tilde{\omega}_{N,s}](x;x')\,dy\,dr\,.
\end{split}
\ee
Then, \eqref{eq:1term-tr} can be estimated by
\be
\eqref{eq:1term-tr}\leq C\int_0^\infty\frac{1}{r^{1+\alpha}}\int |\rho_s(y)-\tilde{\rho}_s(y)|\,\tr\ |[\chi_{(r,y)}\,,\,\tilde{\omega}_{N,s}]|\,dy\,dr\,.
\ee
To bound the trace norm of the commutator between $\tilde{\omega}_{N,s}$ and the multiplication operator $\chi_{(r,y)}$, we borrow an idea from \cite{BPSS}. Namely, we insert the identity operator 
$$
\mathbf{1}=(1-\e^2\Delta)^{-1}(1+x^2)^{-1}(1+x^2)\,(1-\e^2\Delta)
$$  
and perform Cauchy-Schwarz inequality
\be\label{eq:trick}
\tr\ |[\chi_{(r,y)},\tilde{\omega}_{N,s}]|\leq \|(1-\e^2\Delta)^{-1}(1+x^2)^{-1}\|_{\rm HS}\|(1+x^2)\,(1-\e^2\Delta)[\chi_{(r,y)},\tilde{\omega}_{N,s}]\|_{\rm HS}
\ee
Writing down explicitly the kernel of $(1-\e^2\Delta)^{-1}$, straightforward computations (see \cite{Evans}, Chapter 4.3) together with Eq. \eqref{eq:scaling} lead to
\be\label{eq:bound-sqrtN}
\|(1-\e^2\Delta)^{-1}(1+x^2)^{-1}\|_{\rm HS}\leq C\sqrt{N}
\ee
for some positive constant $C$. Hence
\begin{equation*}
\begin{split}
\tr\ &\left|\left[\frac{1}{|\cdot|^{\alpha}}*(\rho_s-\tilde{\rho}_s)\,,\,\tilde{\omega}_{N,s}\right]\right|\\
&\quad\quad\quad\leq C\sqrt{N}\int_0^\infty \frac{dr}{r^{1+\alpha}}\int dy\,|\rho_s(y)-\tilde{\rho}_s(y)|\,\|(1+x^2)(1-\e^2\Delta)[\chi_{(r,y)},\tilde{\omega}_{N,s}]\|_{\rm HS}
\end{split}
\end{equation*}
We make use of Lemma \ref{lemma:fund-est}
\be\label{eq:est-dom-err}
\begin{split}
\tr\ &\left|\left[\frac{1}{|\cdot|^{\alpha}}*(\rho_s-\tilde{\rho}_s)\,,\,\tilde{\omega}_{N,s}\right]\right|\\
&\quad\quad\quad\leq C\sqrt{N}\int_0^\infty \frac{dr}{r^{1+\alpha}}\int dy\,|\rho_s(y)-\tilde{\rho}_s(y)|\,(\|\mathcal{J}_1\|_{\rm HS}+\|\mathcal{J}_{2,1}\|_{\rm HS}+\|\mathcal{J}_{3,1}\|_{\rm HS})\\
&\quad\quad\quad+ C\sqrt{N}\int_0^\infty \frac{dr}{r^{1+\alpha}}\int dy\,|\rho_s(y)-\tilde{\rho}_s(y)|\,\|\mathcal{J}^{\rm err}\|_{\rm HS}
\end{split}
\ee
where $\mathcal{J}_1,\,\mathcal{J}_{2,1},\,\mathcal{J}_{\rm 3,1}$ are defined in \eqref{eq:J1}, \eqref{eq:J2}, \eqref{eq:J3}, while $\mathcal{J}^{\rm err}$ is defined as
\be\label{eq:J-err}
\mathcal{J}^{\rm err}:=\sum_{k=2}^5\mathcal{J}_{2,k}+\sum_{j=2}^4\mathcal{J}_{3,j}
\ee
where $\mathcal{J}_{2,k}$ and $\mathcal{J}_{3,j}$ are defined in \eqref{eq:J2}, \eqref{eq:J3}. 

We now divide the integrals in the $r$ variables into two parts: one close to zero and one far away from zero. More precisely, let us fix $k>0$ and consider:
\begin{itemize}
\item $r\in[0,k]$:
the difficulty here is to deal with the singularity at $r\sim 0$ in the above expression. Indeed we need to extract powers of $r$ from the Hilbert-Schmidt norms of $\mathcal{J}_i$, $i=1,2,3$, making the bounds as sharp as possible. 
In the first term in the r.h.s. of \eqref{eq:est-dom-err} we use Lemma \ref{lemma:J1} and Lemma \ref{lemma:J21-J31} to get
\be
\begin{split}
C&\sqrt{N}\int_0^k \frac{dr}{r^{1+\alpha}}\int dy\,|\rho_s(y)-\tilde{\rho}_s(y)|\,(\|\mathcal{J}_1\|_{\rm HS}+\|\mathcal{J}_{2,1}\|_{\rm HS}+\|\mathcal{J}_{3,1}\|_{\rm HS})\\
&\leq C\,N\,\e\int_0^k\frac{dr}{r^{\frac{1}{2}+\alpha}}\int dy\,|\rho_s(y)-\tilde{\rho}_s(y)|\,\left\{ \int dv\,\sup_x(1+x^2)^2\left[|\nabla_v\widetilde{W}_{N,s}(x,v)|^2\right.\right.\\
&\quad\quad\quad\quad\quad\quad\quad\quad\quad\quad\quad\quad\quad\quad\quad\quad\quad\left.\left.+|v\widetilde{W}_{N,s}(x,v)|^2+|\nabla_v(v^2\widetilde{W}_{N,s}(x,v))|^2\right]\right\}^\frac{1}{2}\\
&+C\,N\e^3\int_0^k\frac{dr}{r^{\frac{1}{2}+\alpha}}\int dy\,|\rho_s(y)-\tilde{\rho}_s(y)|\left\{ \int dv\,\sup_x\left[|\nabla^3_v\widetilde{W}_{N,s}(x,v)|^2\right.\right.\\
&\quad\quad\quad\quad\quad\quad\quad\quad\quad\quad\quad\quad\quad\quad\quad\quad\left.\left.+|v\nabla_v^2\widetilde{W}_{N,s}(x,v)|^2+|\nabla_v(v^2\nabla_v^2\widetilde{W}_{N,s}(x,v))|^2\right]\right\}^\frac{1}{2}
\end{split}
\ee
Moreover, we recall that
\be\label{eq:L1-tr}
\int dy\,|\rho_s(y)-\tilde{\rho}_s(y)|\leq \frac{C}{N}\tr\ |\omega_{N,s}-\tilde{\omega}_{N,s}|
\ee
Indeed we have that
\begin{equation*}
\int dy\,|\rho_s(y)-\tilde{\rho}_s(y)|\leq \sup_{\substack{O\in L^\infty(\R^3)\\ \|O\|_{L^{\infty}}\leq 1 }}\left|\int O(y)\,(\rho_s(y)-\widetilde{\rho}_s(y))\,dy\,\right|\leq \frac{1}{N}\sup_{\|O\|\leq 1} |\tr\ O(\omega_{N,s}-\widetilde{\omega}_{N,s})|,
\end{equation*}
where the supremum on the r.h.s. is taken over all bounded operators with operator norm less or equal than one. 

Therefore, the first term on the r.h.s. of \eqref{eq:est-dom-err} is bounded by
\be\label{eq:bound-main-k}
\begin{split}
C&\sqrt{N}\int_0^k \frac{dr}{r^{1+\alpha}}\int dy\,|\rho_s(y)-\tilde{\rho}_s(y)|\,(\|\mathcal{J}_1\|_{\rm HS}+\|\mathcal{J}_{2,1}\|_{\rm HS}+\|\mathcal{J}_{3,1}\|_{\rm HS})\\
&\leq C\,\e\,\tr\ |\omega_{N,s}-\tilde{\omega}_{N,s}|\,\left\{ \int dv\,\sup_x(1+x^2)^2\left[|\nabla_v\widetilde{W}_{N,s}(x,v)|^2\right.\right.\\
&\quad\quad\quad\quad\quad\quad\quad\quad\quad\quad\quad\quad\quad\quad\quad\quad\quad\quad\left.\left.+|v\widetilde{W}_{N,s}(x,v)|^2+|\nabla_v(v^2\widetilde{W}_{N,s}(x,v))|^2\right]\right\}^\frac{1}{2}\\
&+C\e^3\,\tr\ |\omega_{N,s}-\tilde{\omega}_{N,s}|\,\left\{ \int dv\,\sup_x\left[|\nabla^3_v\widetilde{W}_{N,s}(x,v)|^2\right.\right.\\
&\quad\quad\quad\quad\quad\quad\quad\quad\quad\quad\quad\quad\quad\quad\quad\quad\left.\left.+|v\nabla_v^2\widetilde{W}_{N,s}(x,v)|^2+|\nabla_v(v^2\nabla_v^2\widetilde{W}_{N,s}(x,v))|^2\right]\right\}^\frac{1}{2}
\end{split}
\ee
for $\alpha\in(0,1/2)$. 
\item  $r\in(k,\infty)$: as for this part of the integral in the $r$ variable, we do not need to extract further powers of $r$. We therefore use the point-wise bounds
$$
\chi_{(r,\sqrt{s},y)}(x)\frac{|x-y|^j}{r^j/\sqrt{s^j}}\leq 1,\quad\quad \chi_{(r,\sqrt{1-s},y)}(x')\frac{|x'-y|^j}{r^j/\sqrt{(1-s)^j}}\leq 1
$$ 
in $\mathcal{J}_1,\,\mathcal{J}_{2,1}$ and $\mathcal{J}_{3,1}$
and we are left with 
\be\label{eq:J1-infty}
\begin{split}
\|\mathcal{J}_1\|_{\rm HS}&\leq\frac{C\,\sqrt{N}\,\e}{r}\left[\iint dx\,dv\,(1+x^2)^2|\nabla_v\widetilde{W}_{N,s}(x,v)|^2\right]^{\frac{1}{2}}\\
&+\frac{C\,\sqrt{N}\,\e^3}{r}\left[\iint dx\,dv\,|\nabla^3_v\widetilde{W}_{N,s}(x,v)|^2\right]^{\frac{1}{2}}
\end{split}
\ee
where we have used \eqref{eq:mom-variables} and the change of variables \eqref{eq:change-var}.\\
The other terms can be handled analogously, thus producing the bounds
\be
\begin{split}
\|&\mathcal{J}_{i,1}\|_{\rm HS}\\
&\leq\frac{C\,\sqrt{N}\,\e}{r}\left\{\iint dx\,dv\,(1+x^2)^2\left[|2v\,\widetilde{W}_{N,s}(x,v)|^2+|v^2\nabla_v\widetilde{W}_{N,s}(x,v)|^2\right]\right\}^{\frac{1}{2}}\\
&+\frac{C\,\sqrt{N}\,\e^3}{r}\left\{\iint dx\,dv\,\left[|\nabla^2_v(2v\,\widetilde{W}_{N,s}(x,v))|^2+|\nabla_v^3(v^2\,\widetilde{W}_{N,s}(x,v))|^2\right]\right\}^{\frac{1}{2}}
\end{split}
\ee
for $i=2,3$. Therefore, for the first integral in \eqref{eq:est-dom-err} with $r\in(k,\infty)$ and $\alpha>0$ we obtain the bound
\be
\begin{split}
C&\sqrt{N}\int_k^\infty \frac{dr}{r^{1+\alpha}}\int dy\,|\rho_s(y)-\tilde{\rho}_s(y)|\,(\|\mathcal{J}_1\|_{\rm HS}+\|\mathcal{J}_{2,1}\|_{\rm HS}+\|\mathcal{J}_{3,1}\|_{\rm HS})\\
&\leq C\,\e\,\tr\ |\omega_{N,s}-\tilde{\omega}_{N,s}|\,\left\{ \iint dx\,dv\,(1+x^2)^2\left[|\nabla_v\widetilde{W}_{N,s}(x,v)|^2+|2v\,\widetilde{W}_{N,s}(x,v)|^2\right.\right.\\
&\left.\left.\quad\quad\quad\quad\quad\quad\quad\quad\quad\quad\quad\quad\quad\quad\quad\quad\quad\quad\quad\quad\quad\quad\quad\quad\quad\quad+|v^2\nabla_v\widetilde{W}_{N,s}(x,v)|^2\right]\right\}^\frac{1}{2}\\
&+C\e^3\,\tr\ |\omega_{N,s}-\tilde{\omega}_{N,s}|\,\left\{ \iint dx\,dv\,\left[|\nabla^3_v\widetilde{W}_{N,s}(x,v)|^2\right.\right.\\
&\quad\quad\quad\quad\quad\quad\quad\quad\quad\quad\quad\quad\quad\quad\quad\quad\left.\left.+|\nabla_v^2(2v\,\widetilde{W}_{N,s}(x,v))|^2+|\nabla^3_v(v^2\,\widetilde{W}_{N,s}(x,v))|^2\right]\right\}^\frac{1}{2}
\end{split}
\ee
\end{itemize}
As for the second integral in \eqref{eq:est-dom-err} we proceed differently. Indeed, being the singularity at $r= 0$ worst than the one in the first integral, we need to use also the integral in the $y$ variable to extract a sufficiently high power of $r$, which cancels the singularity at $r=0$. We therefore first make use of the triangular inequality to bound
$$
|\rho_s(y)-\tilde{\rho}_s(y)|\leq |\rho_s(y)|+|\tilde{\rho}_s(y)|
$$
so that we are left with the bound on
\be\label{eq:error-term}
C\sqrt{N}\int_0^\infty\frac{dr}{r^{1+\alpha}}\int dy\,(|\rho_s(y)|+|\tilde{\rho}_s(y)|)\,\|\mathcal{J}^{\rm err}\|_{\rm HS}.
\ee
We perform H\"{o}lder inequality in the $y$ variable with exponents $p=5/3$ and $p'=5/2$. Thus
\be
\begin{split}
\eqref{eq:error-term}&\leq C\sqrt{N}\int_0^\infty \frac{dr}{r^{1+\alpha}}\,(\|\rho_s\|_{L^{5/3}}+\|\tilde{\rho}_s\|_{L^{5/3}})\,\left(\int dy\,\|\mathcal{J}^{\rm err}\|_{\rm HS}^{5/2}\right)^\frac{2}{5}
\end{split}
\ee
We now split the integral in the $r$ variable into two parts. More precisely, for a fixed positive constant $k$, we consider
\begin{itemize}
\item $r\in[0,k]$: to cancel the singularity in $r$, we apply Lemmas \ref{lemma:J-2k}, \ref{lemma:J-25}, \ref{lemma:J-22}, \ref{lemma:J-32}, \ref{lemma:J-34} and \ref{lemma:J-33}, which lead together with Proposition \ref{prop:interpolation} and Young inequality to
\be\label{eq:bound-error-norms}
\begin{split}
\eqref{eq:error-term}&\leq C\,N\,\e^{\frac{7}{5}+\frac{3}{5}\delta}\int_0^k \frac{dr}{r^{\frac{1}{2}+\frac{6}{5}\delta+\alpha}}\,(\|\rho_s\|_{L^{\frac{5}{3}}}+\|\tilde{\rho}_s\|_{L^{\frac{5}{3}}})\\
&\quad\quad\quad\quad\quad\quad\quad\left(\|\widetilde{W}_{N,s}\|_{\mathcal{W}_8^{0,1}}^{\frac{3}{10}-\frac{2}{5}\delta}\|\widetilde{W}_{N,s}\|_{H_2^0}^{\frac{1}{10}+\frac{2}{5}\delta}\|(1+x^2)^2\widetilde{W}_{N,s}\|^{\frac{1}{10}}_{L_v^2(L^\infty_x)} \right.\\
&\left.\quad\quad\quad\quad\quad\quad\quad\quad\quad\quad\quad+\|f_{s}\|_{\mathcal{W}_8^{0,1}}^{\frac{3}{10}-\frac{2}{5}\delta}\|f_{s}\|_{H_2^0}^{\frac{1}{10}+\frac{2}{5}\delta}\|(1+x^2)^2f_{s}\|^{\frac{1}{10}}_{L_v^2(L^\infty_x)} \right)\\
&+ C\,N\,\e^{\frac{17}{5}+\frac{3}{5}\delta}\int_0^k \frac{dr}{r^{\frac{1}{2}+\frac{6}{5}\delta+\alpha}}\,(\|\rho_s\|_{L^\frac{5}{3}}+\|\tilde{\rho}_s\|_{L^\frac{5}{3}})\\
&\quad\quad\quad\quad\quad\quad\quad\left(\|\widetilde{W}_{N,s}\|_{\mathcal{W}^{2,1}}^{\frac{3}{10}-\frac{2}{5}\delta}\|\widetilde{W}_{N,s}\|_{H^2}^{\frac{1}{10}+\frac{2}{5}\delta}\|\widetilde{W}_{N,s}\|^{\frac{1}{10}}_{H_v^2(L^\infty_x)} \right.\\
&\left.\quad\quad\quad\quad\quad\quad\quad\quad\quad\quad\quad+\|f_{s}\|_{\mathcal{W}^{2,1}}^{\frac{3}{10}-\frac{2}{5}\delta}\|f_{s}\|_{H^2}^{\frac{1}{10}+\frac{2}{5}\delta}\|f_{s}\|^{\frac{1}{10}}_{H_v^2(L^\infty_x)} \right)\\
&+C\,N\,\e^{\frac{7}{5}+\frac{3}{5}\delta}\int_0^k \frac{dr}{r^{-\frac{1}{2}+\frac{6}{5}\delta+\alpha}}\,(\|\rho_s\|_{L^\frac{5}{3}}+\|\tilde{\rho}_s\|_{L^\frac{5}{3}})\\
&\quad\quad\quad\quad\quad\quad\quad\quad\quad\quad\quad\left(\|g_{s}\|_{\mathcal{W}_8^{0,1}}^{\frac{3}{10}-\frac{2}{5}\delta}\|g_{s}\|_{H_2^0}^{\frac{1}{10}+\frac{2}{5}\delta}\|(1+x^2)^2g_{s}\|^{\frac{1}{10}}_{L_v^2(L^\infty_x)} \right)\\
&+C\,N\,\e^{\frac{17}{5}+\frac{3}{5}\delta}\int_0^k \frac{dr}{r^{-\frac{1}{2}+\frac{6}{5}\delta+\alpha}}\,(\|\rho_s\|_{L^\frac{5}{3}}+\|\tilde{\rho}_s\|_{L^\frac{5}{3}})\\
&\quad\quad\quad\quad\quad\quad\quad\quad\quad\quad\quad\left(\|g_{s}\|_{\mathcal{W}_8^{2,1}}^{\frac{3}{10}-\frac{2}{5}\delta}\|g_{s}\|_{H^2}^{\frac{1}{10}+\frac{2}{5}\delta}\|(1+x^2)^2g_{s}\|^{\frac{1}{10}}_{H_v^2(L^\infty_x)} \right)
\end{split}
\ee
where
\be\label{eq:f-s}
f_s(x,v):=\widetilde{W}_{N,s}(x,v)+\e\nabla^2\widetilde{W}_{N,s}(x,v)+v\cdot\widetilde{W}_{N,s}(x,v),
\ee
\be\label{eq:g-s}
g_s(x,v):=\nabla\widetilde{W}_{N,s}(x,v)+\e\nabla^3\widetilde{W}_{N,s}(x,v)+v\nabla^2\widetilde{W}_{N,s}(x,v).
\ee   
We recall that $\rho_s,\,\tilde{\rho}_s\in L^{5/3}(\R^3)$, therefore there exists $C>0$ such that 
$$\|\rho_s\|_{L^{5/3}}\leq C\quad\quad\quad \|\tilde{\rho_s}\|_{L^{5/3}}\leq C$$
uniformly in $N$ (see observation 2. after Theorem \ref{thm:trace} and Remark \ref{rk:kin-energy} below).
Moreover for $\alpha\in(0,1/2)$ and $\delta$ very small, the integrals in the $r$ variable converge.  

%

\item $r\in(k,\infty)$: as for this term we proceed analogously to the proof of the bound on the first integral in the r.h.s. of \eqref{eq:est-dom-err}. We estimate the integral in the $y$ variable of the Hilbert-Schmidt norms of the operators $\mathcal{J}_{2,j}$, $j=2,\dots,5$ and $\mathcal{J}_{3,i}$, $i=2,\dots,4$, by first integrating in the $y$ variable. This integration gives rise to $r^{3/2}$. Then, instead of gaining other powers of $r$ from the integral in the $x$ variable as done for instance in \eqref{eq:interp-1}, we use the point-wise bound on the Gaussian \eqref{eq:interp-2}. Therefore the integral in the $r$ variable converges for $\alpha\in(0,1/2)$. Gathering together all the estimates and applying Proposition \ref{prop:interpolation} and Young inequality, the second integral in the r.h.s. of \eqref{eq:est-dom-err} turns out to be bounded by
\be\label{eq:bound-main-infty}
\begin{split}
\eqref{eq:error-term}\leq & CN\e^{\frac{23}{10}} \int_k^{\infty} \frac{dr}{r^{\frac{9}{5}+\alpha}}(\|\rho_s\|_{L^{\frac{5}{3}}}+\|\tilde{\rho}_s\|_{L^\frac{5}{3}})\left(\|\widetilde{W}_{N,s}\|_{H^0_2}^{\frac{9}{10}}+\|f_s\|_{H_2^0}^{\frac{9}{10}} \right) \\
&+ CN\e^\frac{14}{5}\int_k^\infty \frac{dr}{r^{1+\alpha}}(\|\rho_s\|_{L^{\frac{5}{3}}}+\|\tilde{\rho}_s\|_{L^\frac{5}{3}})\,\|g_s\|_{H^0_2}^{\frac{9}{10}}\\
&+ CN\e^\frac{43}{10}\int_0^k \frac{dr}{r^{\frac{9}{5}+\alpha}}(\|\rho_s\|_{L^\frac{5}{3}}+\|\tilde{\rho}_s\|_{L^\frac{5}{3}})\left(\|\widetilde{W}_{N,s}\|_{H^2_2}^\frac{9}{10}+\|f_s\|_{H^2_2}^\frac{9}{10}\right)\\
&+CN\e^\frac{24}{5}\int_k^\infty \frac{dr}{r^{1+\alpha}}(\|\rho_s\|_{L^\frac{5}{3}}+\|\tilde{\rho}_s\|_{L^\frac{5}{3}})\,\|g_s\|_{H_0^2}^\frac{9}{10} 
\end{split}
\ee
\end{itemize}
Gathering together \eqref{eq:bound-main-k}, \eqref{eq:bound-error-norms} and \eqref{eq:bound-main-infty}, we conclude the proof.
\end{proof}

\begin{remark}
We recall definitions of $f_s$ \eqref{eq:f-s} and $g_s$ \eqref{eq:g-s}
and observe that the $\e$-dependence in \eqref{eq:bound-error-norms} can be  made more explicit by using the factor $\e$ in \eqref{eq:f-s}-\eqref{eq:g-s}.  
\end{remark}

\begin{proof}[Proof of Lemma \ref{lemma:fund-est}]
The Hilbert-Schmidt norm of the operator $(1+x^2)(1-\e^2\Delta)[\chi_{(r,z)},\tilde{\omega}_{N,t}]$ is given by the $L^2$ norm of its integral kernel.
We notice that 
\begin{equation}
\begin{split}
[\chi_{(r,y)},&\tilde{\omega}_{N,t}](x;x')\\
&=-\sum_{k=1}^3\int_0^1 ds\,\chi_{(r/\sqrt{s},y)}(x)\,\left(\frac{(x-y)_k}{r^2}+\frac{(x'-y)_k}{r^2}\right)\,\chi_{(r/\sqrt{1-s},y)}(x')\,[x_k,\widetilde{\omega}_{N,t}](x;x') \\
\end{split}
\end{equation}
Hence, 
\begin{equation}
\|(1+x^2)(1-\e^2\Delta)[\chi_{(r,y)},\tilde{\omega}_{N,t}]\|_{HS}^2=\int dx\int dx'\left| (1+x^2)(1-\e^2\Delta_x)[\chi_{(r,y)},\tilde\omega_{N,t}](x;x') \right|^2
\end{equation}
We observe that 
\begin{equation}\label{eq:kernel}
(1+x^2)(1-\e^2\Delta_x)[\chi_{(r,y)},\tilde\omega_{N,s}](x;x')=\mathcal{J}_1+\mathcal{J}_2+\mathcal{J}_3
\end{equation}
where we have introduced the shorthand notation introduced in \eqref{eq:J1-J2-J3}.

Straightforward computations together with the observation  
\be\label{eq:v-der}
(x-x')\int dv\,\widetilde{W}_{N,t}\left(\frac{x+x'}{2},v\right)\,e^{iv\cdot\frac{(x-x')}{\e}}=-i\e\int dv\,\nabla_v\widetilde{W}_{N,t}\left(\frac{x+x'}{2},v\right)\,e^{iv\cdot\frac{(x-x')}{\e}}
\ee
leads to the following estimates for each term in \eqref{eq:kernel}:
\be\label{eq:J1}
\begin{split}
|\mathcal{J}_1|&\leq \frac{N}{r}\left|\int_0^1 ds\,(1+x^2)\,\chi_{(r/\sqrt{s},y)}(x)\chi_{(r/\sqrt{1-s},y)}(x')\right.\\
&\quad\quad\quad\left.\times\left[\frac{1}{\sqrt{s}}\frac{(x-y)}{r/\sqrt{s}}+\frac{1}{\sqrt{1-s}}\frac{(x'-y)}{r/\sqrt{1-s}}\right]\cdot(x-x')\int dv\,\widetilde{W}_{N,t}\left(\frac{x+x'}{2},v\right)\,e^{iv\cdot\frac{(x-x')}{\e}}\right|\quad\quad\quad\quad\quad\quad\quad\\
\end{split}
\ee
\be\label{eq:J2+}
\begin{split}
|\mathcal{J}_2|&\leq \frac{N\e}{r}\left|\int_0^1 \frac{ds}{\sqrt{s}}\,(1+x^2)\,\chi_{(r/\sqrt{s},y)}(x)\frac{(x-y)}{r/\sqrt{s}}\chi_{(r/\sqrt{1-s},y)}(x')\right.\\
&\quad\quad\quad\left.\times\int dv\,\left[2v\,\widetilde{W}_{N,s}\left(\frac{x+x'}{2},v\right)+v^2\nabla_v\widetilde{W}_{N,t}\left(\frac{x+x'}{2},v\right)\right]\,e^{iv\cdot\frac{(x-x')}{\e}}\right|\\
&+\frac{4\,N\e^2}{r^2}\left|\int_0^1 ds\,(1+x^2)\,\chi_{(r/\sqrt{s},y)}(x)\chi_{(r/\sqrt{1-s},y)}(x')\right.\\
&\quad\quad\quad\quad\left.\times\left[\frac{(x-y)^3}{r^3/\sqrt{s^3}}+2\frac{(x-y)}{r/\sqrt{s}}\right]\cdot\frac{(x-x')}{r/\sqrt{s}}\int dv\,\widetilde{W}_{N,t}\left(\frac{x+x'}{2},v\right)\,e^{iv\cdot\frac{(x-x')}{\e}}\right|\\
&+\frac{4\,N\e^2}{r^2}\left|\int_0^1 ds\,(1+x^2)\,\chi_{(r/\sqrt{s},y)}(x)\frac{(x-y)^2}{r^2/s}\chi_{(r/\sqrt{1-s},y)}(x')\,\int dv\,e^{iv\cdot\frac{(x-x')}{\e}}\right.\\
&\quad\quad\quad\quad\left.\times\left[\widetilde{W}_{N,t}\left(\frac{x+x'}{2},v\right)+\e\nabla^2_{v,x}\widetilde{W}_{N,t}\left(\frac{x+x'}{2},v\right)+v\cdot\nabla_v\widetilde{W}_{N,t}\left(\frac{x+x'}{2},v\right)\right]\,\right|\\
&+\frac{2\,N\e^2}{r^2}\left|\int_0^1 ds\,(1+x^2)\,\chi_{(r/\sqrt{s},y)}(x)\chi_{(r/\sqrt{1-s},y)}(x')\,\int dv\,e^{iv\cdot\frac{(x-x')}{\e}}\right.\\
&\quad\quad\quad\quad\left.\times\left[\widetilde{W}_{N,t}\left(\frac{x+x'}{2},v\right)+\e\nabla^2_{v,x}\widetilde{W}_{N,t}\left(\frac{x+x'}{2},v\right)+v\cdot\nabla_v\widetilde{W}_{N,t}\left(\frac{x+x'}{2},v\right)\right]\,\right|\\
&+\frac{4\,N\e^2}{r}\left|\int_0^1 \frac{ds}{\sqrt{s}}\,(1+x^2)\,\chi_{(r/\sqrt{s},y)}(x)\frac{(x-y)}{r/\sqrt{s}}\chi_{(r/\sqrt{1-s},y)}(x')\,\int dv\,e^{iv\cdot\frac{(x-x')}{\e}}\right.\\
&\quad\quad\quad\left.\times\left[\nabla_x\widetilde{W}_{N,s}\left(\frac{x+x'}{2},v\right)+\frac{\e}{2}\nabla^3_{v,x,x}\widetilde{W}_{N,t}\left(\frac{x+x'}{2},v\right)+v\nabla^2_{v,x}\widetilde{W}_{N,t}\left(\frac{x+x'}{2},v\right)\right]\,\right|\\
\end{split}
\ee
\be\label{eq:J3+}
\begin{split}
|\mathcal{J}_3|&\leq\frac{N\e}{r}\left|\int_0^1 \frac{ds}{\sqrt{1-s}}\,(1+x^2)\,\chi_{(r/\sqrt{s},y)}(x)\,\chi_{(r/\sqrt{1-s},y)}(x')\,\frac{(x'-y)}{r/\sqrt{1-s}}\int dv\,e^{iv\cdot\frac{(x-x')}{\e}}\right.\\
&\quad\quad\quad\left.\times\left[2v\,\widetilde{W}_{N,t}\left(\frac{x+x'}{2},v\right)+\nabla_{v}\left(v^2\,\widetilde{W}_{N,t}\left(\frac{x+x'}{2},v\right)\right)\right]\,\right|\\
&+\frac{N\e^2}{r^2}\left|\int_0^1 ds\,2s\,(1+x^2)\,\chi_{(r/\sqrt{s},y)}(x)\left[2\frac{|x-y|^2}{r^2/s}+1\right]\,\chi_{(r/\sqrt{1-s},y)}(x')\,\frac{(x'-y)}{r}\cdot\frac{(x-x')}{r}\right.\\
&\quad\quad\quad\quad\left.\times\int dv\,\widetilde{W}_{N,t}\left(\frac{x+x'}{2},v\right)\,e^{iv\cdot\frac{(x-x')}{\e}}\right|\\
&+\frac{4\,N\e^2}{r^2}\left|\int_0^1 ds\,\frac{\sqrt{s}}{\sqrt{1-s}}\,(1+x^2)\,\chi_{(r/\sqrt{s},y)}(x)\frac{(x-y)}{r/\sqrt{s}}\,\chi_{(r/\sqrt{1-s},y)}(x')\,\frac{(x'-y)}{r/\sqrt{1-s}}\int dv\,e^{iv\cdot\frac{(x-x')}{\e}}\right.\\
&\quad\quad\quad\left.\times\left[\widetilde{W}_{N,t}\left(\frac{x+x'}{2},v\right)+\e\nabla^2_{v,x}\widetilde{W}_{N,t}\left(\frac{x+x'}{2},v\right)+\nabla_{v}\left(v\,\widetilde{W}_{N,t}\left(\frac{x+x'}{2},v\right)\right)\right]\,\right|\\
&+\frac{2\,N\e^2}{r}\left|\int_0^1\frac{ds}{\sqrt{1-s}}\,(1+x^2)\,\chi_{(r/\sqrt{s},y)}(x)\,\chi_{(r/\sqrt{1-s},y)}(x')\,\frac{(x'-y)}{r/\sqrt{1-s}}\int dv\,e^{iv\cdot\frac{(x-x')}{\e}}\right.\\
&\quad\quad\quad\left.\times\left[\nabla_x\widetilde{W}_{N,t}\left(\frac{x+x'}{2},v\right)+\frac{\e}{2}\nabla^3_{v,x,x}\widetilde{W}_{N,t}\left(\frac{x+x'}{2},v\right)+\nabla^2_{v,x}\left(v\,\widetilde{W}_{N,t}\left(\frac{x+x'}{2},v\right)\right)\right]\,\right|\\
\end{split}
\ee
$\mathcal{J}_2$ and $\mathcal{J}_3$ are bounded respectively  by five and four terms which correspond to $\mathcal{J}_{2,k}$, for $k=1,\dots,5$, and $\mathcal{J}_{3,j}$, for $j=1,\dots,4$ defined in \eqref{eq:J2} and \eqref{eq:J3}.

\end{proof}

The rest of this section is devoted to prove bounds on the Hilbert-Schmidt norm of the operators $\mathcal{J}_i$, for $i=1,2,3$, defined in \eqref{eq:J1-J2-J3}.

\begin{lemma}\label{lemma:J1}
\be
\begin{split}
\|\mathcal{J}_1\|_{\rm HS}&\leq C\,\sqrt{N}\,\e\,\sqrt{r}\,\left\{ \int dv\,\sup_X\,(1+X^{2})^2|\nabla_v\widetilde{W}_{N,t}(X,v)|^2\right\}^{1/2} \\
&+C\,\sqrt{N}\,\e^3\,\sqrt{r}\,\left\{ \int dv\,\sup_X\,|\nabla^3_v\widetilde{W}_{N,t}(X,v)|^2\right\}^{1/2}
\end{split}
\ee
\end{lemma}
\begin{proof}
We observe that 
\be\label{eq:J1-est}
\begin{split}
\|\mathcal{J}_1\|_{\rm HS}&\leq \frac{C\,N\,\e}{r}\left\{ \int dx\int dx'\left|\int_0^1\frac{ds}{\sqrt{s}}(1+x^2)\,\chi_{(r/\sqrt{s},y)}(x)\,\frac{(x-y)}{r/\sqrt{s}}\,\chi_{(r/\sqrt{1-s},y)}(x')   \right.\right.\\
&\quad\quad\quad\quad\quad\quad\quad\quad\quad\quad\left.\left.  \cdot\int dv\,\nabla_v\widetilde{W}_{N,t}\left(\frac{x+x'}{2},v\right)\,e^{iv\cdot\frac{(x-x')}{\e}}\,\right|^2   \right\}^{1/2}\\
&+ \frac{C\,N\,\e}{r}\left\{ \int dx\int dx'\left|\int_0^1\frac{ds}{\sqrt{1-s}}(1+x^2)\,\chi_{(r/\sqrt{s},y)}(x)\,\chi_{(r/\sqrt{1-s},y)}(x')\frac{(x'-y)}{r/\sqrt{1-s}}   \right.\right.\\
&\quad\quad\quad\quad\quad\quad\quad\quad\quad\quad\left.\left.  \cdot\int dv\,\nabla_v\widetilde{W}_{N,t}\left(\frac{x+x'}{2},v\right)\,e^{iv\cdot\frac{(x-x')}{\e}}\,\right|^2   \right\}^{1/2}\\
\end{split}
\ee
where we used \eqref{eq:v-der} to extract a factor $\e$. This will turn out to be very important to look at time scales of order one. Indeed, such a $\e$ will be used to cancel the factor $\e^{-1}$ in front of the time integral in \eqref{eq:est-tr}.

We denote the two integrals on the r.h.s. of \eqref{eq:J1-est} respectively by $\mathcal{J}_{1,1}$ and $\mathcal{J}_{1,2}$, and we observe that for every $x'\in\R^3$ and $\e>0$
\be\label{eq:mom-variables}
(1+x^2)\leq1+2\left(\frac{x+x'}{2}\right)^2+\e^2\left(\frac{x-x'}{\e}\right)^2.
\ee
Hence, using again \eqref{eq:v-der},
\be\label{eq:J1-div-var}
\begin{split}
\mathcal{J}_{1,1}&\leq \frac{C\,N\,\e}{r}\left\{ \int dx\int dx'\left|\int_0^1\frac{ds}{\sqrt{s}}\left(1+\left(\frac{x+x'}{2}\right)^2\right)\,\chi_{(r/\sqrt{s},y)}(x)\,\frac{(x-y)}{r/\sqrt{s}}\,\chi_{(r/\sqrt{1-s},y)}(x')   \right.\right.\\
&\quad\quad\quad\quad\quad\quad\quad\quad\quad\quad\left.\left.  \cdot\int dv\,\nabla_v\widetilde{W}_{N,t}\left(\frac{x+x'}{2},v\right)\,e^{iv\cdot\frac{(x-x')}{\e}}\,\right|^2   \right\}^{1/2}\\
&+\frac{C\,N\,\e^3}{r}\left\{ \int dx\int dx'\left|\int_0^1\frac{ds}{\sqrt{s}}\,\chi_{(r/\sqrt{s},y)}(x)\,\frac{(x-y)}{r/\sqrt{s}}\,\chi_{(r/\sqrt{1-s},y)}(x')   \right.\right.\\
&\quad\quad\quad\quad\quad\quad\quad\quad\quad\quad\left.\left.  \cdot\int dv\,\nabla^3_v\widetilde{W}_{N,t}\left(\frac{x+x'}{2},v\right)\,e^{iv\cdot\frac{(x-x')}{\e}}\,\right|^2   \right\}^{1/2}\\
\end{split}
\ee
and 
\be\label{eq:J2-div-var}
\begin{split}
\mathcal{J}_{1,2}&\leq \frac{C\,N\,\e}{r}\left\{ \int dx\int dx'\left|\int_0^1\frac{ds}{\sqrt{1-s}}\left(1+\left(\frac{x+x'}{2}\right)^2\right)\,\chi_{(r/\sqrt{s},y)}(x)\,\chi_{(r/\sqrt{1-s},y)}(x')\,\frac{(x'-y)}{r/\sqrt{1-s}}   \right.\right.\\
&\quad\quad\quad\quad\quad\quad\quad\quad\quad\quad\left.\left.  \cdot\int dv\,\nabla_v\widetilde{W}_{N,t}\left(\frac{x+x'}{2},v\right)\,e^{iv\cdot\frac{(x-x')}{\e}}\,\right|^2   \right\}^{1/2}\\
&+\frac{C\,N\,\e^3}{r}\left\{ \int dx\int dx'\left|\int_0^1\frac{ds}{\sqrt{1-s}}\,\chi_{(r/\sqrt{s},y)}(x)\,\chi_{(r/\sqrt{1-s},y)}(x')\,\frac{(x'-y)}{r/\sqrt{1-s}}   \right.\right.\\
&\quad\quad\quad\quad\quad\quad\quad\quad\quad\quad\left.\left.  \cdot\int dv\,\nabla^3_v\widetilde{W}_{N,t}\left(\frac{x+x'}{2},v\right)\,e^{iv\cdot\frac{(x-x')}{\e}}\,\right|^2   \right\}^{1/2}\\
\end{split}
\ee
We first estimate $\mathcal{J}_{1,1}$ and focus on the first integral on the r.h.s. of 
\eqref{eq:J1-div-var}
\be
\begin{split}
{\mathcal{J}}_{1,1}^{(1)}:=& \frac{C\,N\,\e}{r}\left\{ \int dx\int dx'\left|\int_0^1\frac{ds}{\sqrt{s}}\left(1+\left(\frac{x+x'}{2}\right)^2\right)\,\chi_{(r/\sqrt{s},y)}(x)\,\frac{(x-y)}{r/\sqrt{s}}\,\chi_{(r/\sqrt{1-s},y)}(x')   \right.\right.\\
&\quad\quad\quad\quad\quad\quad\quad\quad\quad\quad\left.\left.  \cdot\int dv\,\nabla_v\widetilde{W}_{N,t}\left(\frac{x+x'}{2},v\right)\,e^{iv\cdot\frac{(x-x')}{\e}}\,\right|^2   \right\}^{1/2}\\
\end{split}
\ee
By using Jensen's inequality with measure $\frac{ds}{\sqrt{s}}$ and performing the change of variable 
\be\label{eq:change-var}
X=\frac{x+x'}{2}\,,\quad\quad X'=\frac{x-x'}{\e}\,,
\ee
with Jacobian $J=8\,\e^3$, we obtain the bound
\be
\begin{split}
\mathcal{J}_{1,1}^{(1)}
&\leq \frac{C\,\sqrt{N}\,\e}{r}\left\{ \int dX\int dX'\int_0^1\frac{ds}{\sqrt{s}}\,\chi_{(r/\sqrt{2(1-s)},y)}(X-\e\,X'/2) \right.\\
&\quad\quad\quad\quad\quad\quad \times \chi_{(r/\sqrt{2s},y)}(X+\e\,X'/2)\,\frac{|X+\e\,X'/2-y|^2}{r^2/2s}\\
&\quad\quad\quad\quad\quad\quad\left. \times \int dv\int dv'\,(1+X^{2})^2\nabla_v\widetilde{W}_{N,t}(X,v)\cdot\nabla_{v'}\widetilde{W}_{N,t}(X,v')\,e^{i(v-v')\cdot X'}\right\}^{1/2}
\end{split}
\ee
where we have used the identity $\e^3=N^{-1}$.
H\"{o}lder inequality in the $X$ variable with $p=1$ and $q=\infty$ yields 
\be
\begin{split}
\mathcal{J}_{1,1}^{(1)}&\leq \frac{C\,\sqrt{N}\,\e}{r}\left\{ \int dX'\int_0^1\frac{ds}{\sqrt{s}}\,r^3\,s\,(1-s)\,e^{s(1-s)\e^2|X'|^2/r^2}\left(1+\frac{\e^4|X'|^4}{r^4/\sqrt{s^4(1-s)^4}}\right) \right.\\
&\quad\quad\quad\quad\quad\quad\left. \times \int dv\int dv'\,\sup_X\{(1+X^{2})^2\nabla_v\widetilde{W}_{N,t}(X,v)\cdot\nabla_{v'}\widetilde{W}_{N,t}(X,v')\}\,e^{i(v-v')\cdot X'}\right\}^{1/2}
\end{split}
\ee
We observe that the Gaussian $e^{-s(1-s)\e^2|X'|^2/r^2}$  and the function $e^{-s(1-s)\e^2|X'|^2/r^2}\frac{\e^4|X'|^4}{r^4/\sqrt{s^4(1-s)^4}}$ are point-wise bounded. Moreover, the integral in the $X'$ variable produces a Dirac delta, i.e. $\delta(v-v')=(2\pi)^{-3}\int dX'\,e^{i(v-v')\cdot X'}$.  $\mathcal{J}_{1,1}^{(1)}$ is therefore bounded by
\be
\mathcal{J}_{1,1}^{(1)}\leq C\,\sqrt{N}\,\e\,\sqrt{r}\,\left\{ \int dv\,\sup_X\,(1+X^{2})^2|\nabla_v\widetilde{W}_{N,t}(X,v)|^2\right\}^{1/2}
\ee 
We now focus on 
\be
\begin{split}
\mathcal{J}_{1,1}^{(2)}&:= \frac{C\,N\,\e^3}{r}\left\{ \int dx\int dx'\left|\int_0^1\frac{ds}{\sqrt{s}}\,\chi_{(r/\sqrt{s},y)}(x)\,\frac{(x-y)}{r/\sqrt{s}}\,\chi_{(r/\sqrt{1-s},y)}(x')   \right.\right.\\
&\quad\quad\quad\quad\quad\quad\quad\quad\quad\quad\left.\left.  \cdot\int dv\,\nabla^3_v\widetilde{W}_{N,t}\left(\frac{x+x'}{2},v\right)\,e^{iv\cdot\frac{(x-x')}{\e}}\,\right|^2   \right\}^{1/2}
\end{split}
\ee
The same argument we used to estimate $\mathcal{J}_{1,1}^{(2)}$ leads to the bound
\be
\mathcal{J}_{1,1}^{(2)}\leq C\,\sqrt{N}\,\e^3\,\sqrt{r}\,\left\{ \int dv\,\sup_X\,|\nabla^3_v\widetilde{W}_{N,t}(X,v)|^2\right\}^{1/2}
\ee
The term $\mathcal{J}_{1,2}$ can be estimated following the same lines of the bound for $\mathcal{J}_{1,1}$. We therefore obtain
\be
\begin{split}
\mathcal{J}_{1,2}&\leq C\,\sqrt{N}\,\e\,\sqrt{r}\,\left\{ \int dv\,\sup_X\,(1+X^{2})^2|\nabla_v\widetilde{W}_{N,t}(X,v)|^2\right\}^{1/2} \\
&+C\,\sqrt{N}\,\e^3\,\sqrt{r}\,\left\{ \int dv\,\sup_X\,|\nabla^3_v\widetilde{W}_{N,t}(X,v)|^2\right\}^{1/2}
\end{split}
\ee
Collecting all the bounds we get
\begin{equation*}
\begin{split}
\|\mathcal{J}_1\|_{\rm HS}&\leq C\,\sqrt{N}\,\e\,\sqrt{r}\,\left\{ \int dv\,\sup_X\,(1+X^{2})^2|\nabla_v\widetilde{W}_{N,t}(X,v)|^2\right\}^{1/2} \\
&+C\,\sqrt{N}\,\e^3\,\sqrt{r}\,\left\{ \int dv\,\sup_X\,|\nabla^3_v\widetilde{W}_{N,t}(X,v)|^2\right\}^{1/2}
\end{split}
\end{equation*}
which concludes the proof.
\end{proof}
\begin{lemma}\label{lemma:J21-J31}
\begin{equation*}
\begin{split}
\|\mathcal{J}_{2,1}\|_{\rm HS}&\leq C\,\sqrt{N}\,\e\,\sqrt{r}\,\left\{ \int dv\,\sup_X\,(1+X^{2})^2\left[|v\,\widetilde{W}_{N,t}(X,v)|^2+|\nabla_v(v^2\widetilde{W}_{N,t}(X,v))|^2\right]\right\}^{1/2} \\
&+C\,\sqrt{N}\,\e^3\,\sqrt{r}\,\left\{ \int dv\,\sup_X\,\left[|v\,\nabla^2_v\widetilde{W}_{N,t}(X,v)|^2+|\nabla_v(v^2\nabla^2_v\widetilde{W}_{N,t}(X,v))|^2\right]\right\}^{1/2}\end{split}
\end{equation*}
\begin{equation*}
\begin{split}
\|\mathcal{J}_{3,1}\|_{\rm HS}&\leq C\,\sqrt{N}\,\e\,\sqrt{r}\,\left\{ \int dv\,\sup_X\,(1+X^{2})^2\left[|v\,\widetilde{W}_{N,t}(X,v)|^2+|\nabla_v(v^2\widetilde{W}_{N,t}(X,v))|^2\right]\right\}^{1/2} \\
&+C\,\sqrt{N}\,\e^3\,\sqrt{r}\,\left\{ \int dv\,\sup_X\,\left[|v\,\nabla^2_v\widetilde{W}_{N,t}(X,v)|^2+|\nabla_v(v^2\nabla^2_v\widetilde{W}_{N,t}(X,v))|^2\right]\right\}^{1/2}
\end{split}
\end{equation*}
\end{lemma}
The proof of Lemma \ref{lemma:J21-J31} can be obtained following line by line the proof of Lemma \ref{lemma:J1}.
\begin{lemma}\label{lemma:J-2k}
For every $\delta\in(0,3/4)$, $k=3,4$, and $f_t=\widetilde{W}_{N,t}+\e\nabla^2_{v,x}\widetilde{W}_{N,t}+v\cdot\nabla_v\widetilde{W}_{N,t}$, the following bound holds
\be\label{eq:J-2k}
\begin{split}
\left(\int dy \|\mathcal{J}_{2,k}\|_{\rm HS}^{\frac{5}{2}}\right)^{\frac{2}{5}}&\leq C\,\sqrt{N}\,\e^{\frac{7}{5}+\frac{3}{5}\delta}\,r^{\frac{1}{2}-\frac{6}{5}\delta}\left(\int dX\left(\int dv\,(1+2\,X^2)\,|f_t(X,v)|\right)^2\right)^{\frac{3}{10}-\frac{2}{5}\delta} \\
&\quad\quad\quad\quad\quad\quad\quad\quad\quad\times\,\left(\int dX\int dv\,(1+2\,X^2)^2|f_t(X,v)|^2\right)^{\frac{1}{10}+\frac{2}{5}\delta} \\
&\quad\quad\quad\quad\quad\quad\quad\quad\quad\times\,\left(\int dv\,\sup_X (1+2\,X^2)^2|f_t(X,v)|^2\right)^{\frac{1}{10}} \\
&+C\,\sqrt{N}\,\e^{\frac{17}{5}+\frac{3}{5}\delta}\,r^{\frac{1}{2}-\frac{6}{5}\delta}\left(\int dX\left(\int dv\,|\nabla_v^2 f_t(X,v)|\right)^2\right)^{\frac{3}{10}-\frac{2}{5}\delta} \\
&\quad\quad\quad\quad\quad\quad\quad\quad\quad\times\,\left(\int dX\int dv\,|\nabla^2_v f_t(X,v)|^2\right)^{\frac{1}{10}+\frac{2}{5}\delta}\\
&\quad\quad\quad\quad\quad\quad\quad\quad\quad\times\,\left(\int dv\,\sup_X |\nabla^2_v f_t(X,v)|^2\right)^{\frac{1}{10}} 
\end{split}
\ee
\end{lemma}
\begin{proof}
We first observe that
\be
\|\mathcal{J}_{2,3}\|_{\rm HS}\leq C\|\mathcal{J}_{2,3}^{(1)}\|_{\rm HS}+ C\|\mathcal{J}_{2,3}^{(2)}\|_{\rm HS}
\ee
where we denote
\begin{displaymath}
\begin{split}
\|\mathcal{J}_{2,3}^{(1)}\|_{\rm HS}:=&\frac{C\,N\,\e^2}{r^2}\left\{ \int dx\int dx'\left|\int_0^1 ds\,\,\chi_{(r/\sqrt{s},y)}(x)\,\frac{|x-y|^2}{r^2/s}\,\chi_{(r/\sqrt{1-s},y)}(x')\,\right. \right.\\
&\quad\quad\quad\quad\quad\quad\quad\quad\quad\quad\left.\left. \times\,\int dv\,\left(1+2\left(\frac{x+x'}{2}\right)^2\right)\,f_t\left(\frac{x+x'}{2},v\right)\,e^{iv\cdot\frac{(x-x')}{\e} }\right|^2\right\}^{1/2}
\end{split}
\end{displaymath}
\begin{displaymath}
\begin{split}
\|\mathcal{J}_{2,3}^{(2)}\|_{\rm HS}:=&\frac{C\,N\,\e^4}{r^2}\left\{ \int dx\int dx'\left|\int_0^1 ds\,\chi_{(r/\sqrt{s},y)}(x)\,\frac{|x-y|^2}{r^2/s}\,\,\chi_{(r/\sqrt{1-s},y)}(x')\right. \right.\quad\quad\quad\quad\quad\quad\quad\\
&\quad\quad\quad\quad\quad\quad\quad\quad\quad\quad\quad\left.\left. \times\,\int dv\,\nabla_v^2 f_t\left(\frac{x+x'}{2},v\right)\,e^{iv\cdot\frac{(x-x')}{\e} }\right|^2\right\}^{1/2}
\end{split}
\end{displaymath}
We first bound $\mathcal{J}_{2,3}^{(1)}$. Jensen's inequality leads to
\begin{equation*}
\begin{split}
\|\mathcal{J}_{2,3}^{(1)}\|_{\rm HS}&\leq\frac{C\,N\,\e^2}{r^2}\left\{ \int dx\int dx'\int_0^1 ds\,\,\chi_{(r/\sqrt{2s},y)}(x)\,\frac{|x-y|^4}{r^4/s^2}\,\chi_{(r/\sqrt{2(1-s)},y)}(x')\, \right.\\
&\quad\quad\left.\times\int dv\int dv'\left(1+2\left(\frac{x+x'}{2}\right)^2\right)^2f_t\left(\frac{x+x'}{2},v\right)f_t\left(\frac{x+x'}{2},v'\right)e^{i(v-v')\cdot\frac{(x-x')}{\e} }\right\}^{1/2}\end{split}
\end{equation*}
We start by estimating $\|\mathcal{J}_{2,3}^{1}\|_{\rm HS}^{1/2}$ in such a way that it results to be bounded uniformly in $y$. To this end, we perform the change of variables \eqref{eq:change-var}, use Young inequality in the $s$ variable and integrate in the $X$ variable using Proposition \ref{prop:gauss-int}, thus obtaining 
\be\label{eq:J-23-1/2}
\begin{split}
\|\mathcal{J}_{2,3}\|_{\rm HS}^{1/2}&\leq \frac{C\e}{r}\left\{\ \int dX'\int_0^1 ds\,r^3\,s(1-s)e^{-s(1-s)\e^2|X'|^2/r^2}\left(1+\frac{\e^4|X'|^4}{r^4/[s(1-s)]^2}\right)\right.\\
&\quad\quad\quad\quad\quad\quad\quad\quad\times\left. \int dv\int dv'\sup_X\,(1+2X^2)^2f_t(X,v)\,f_t(X,v')\,e^{i(v-v')\cdot X'} \right\}^\frac{1}{4}\\
&\leq C\,\e\,r^{-\frac{1}{4}}\,\left(\int dv\,\sup_X\,(1+X^2)^2|f_t(X,v)|^2\right)^\frac{1}{4}
\end{split}
\ee
where in the last inequality we used that the function $s(1-s)e^{-s(1-s)\e^2|X'|^2/r^2}\left(1+\frac{\e^4|X'|^4}{r^4/[s(1-s)]^2}\right)$ is uniformly bounded and  that $\int dX'\,e^{i(v-v')\cdot X'}$ is proportional to the Dirac delta in $(v-v')$.

Now we are interested in estimating $\int dy\,\|\mathcal{J}_{2,3}^{(1)}\|_{\rm HS}^2$.  
Therefore, we perform the change of variables \eqref{eq:change-var}, use Young inequality in the $s$ variable  and integrate in the $y$ variable as done in Proposition \ref{prop:gauss-int}. It results 
\be
\begin{split}
&\int dy \|\mathcal{J}_{2,3}^{(1)}\|_{\rm HS}^2\\
&\quad\quad\quad\leq \frac{C\,{N}\,\e^4}{r^4} \int dX\int dX'\int_0^1 ds\,r^3\,s(1-s)\,e^{-s(1-s)\e^2|X'|^2/r^2}\left(1+\frac{\e^4|X'|^4}{r^4/[s(1-s)]^2}\right)\\
&\quad\quad\quad\quad\quad\quad\quad\quad\quad\quad\quad\quad\quad\quad\quad\times\,\int dv\int dv' (1+2\,X^2)^2f_t(X,v)\,f_t(X,v')\,e^{i(v-v')\cdot X'}  
\end{split}
\ee
Now we observe that on the one hand the integral in the $X'$ variable is bounded by
\be\label{eq:interp-1}
\int dX'\,e^{-s(1-s)\e^2|X'|^2/r^2}\left(1+\frac{\e^4|X'|^4}{r^4/\sqrt{s^4(1-s)^4}}\right)\leq \frac{C\,r^3}{\e^3\,s^{3/2}(1-s)^{3/2}}
\ee
On the other hand the integrand is bounded uniformly in $r$ and $\e$, i.e. there exists $C\geq 0$ such that
\be\label{eq:interp-2}
e^{-s(1-s)\e^2|X'|^2/r^2}\left(1+\frac{\e^4|X'|^4}{r^4/\sqrt{s^4(1-s)^4}}\right)\leq C.
\ee
Therefore, interpolating between \eqref{eq:interp-1} and \eqref{eq:interp-2} we obtain
\be\label{eq:J-23-2}
\begin{split}
\int dy \|\mathcal{J}_{2,3}^{(1)}\|_{\rm HS}^2&\leq \frac{C\,N\,\e^4}{r}\int_0^1 ds\,\frac{r^{3\gamma}}{\e^{3\gamma}[s(1-s)]^{\frac{3}{2}\gamma-1}}\left(\int dX\left(\int dv\,(1+2\,X^2)\,|f_t(X,v)|\right)^{2}\right)^\gamma\\
&\quad\quad\quad\quad\quad\quad\quad\quad\quad\quad\quad\quad\quad\quad\,\quad\quad\quad\times\,\left(\int dv\,(1+2\,X^2)^2\,|f_t(X,v)|^2\right)^{1-\gamma}  
\end{split}
\ee
for $\gamma\in(0,1)$ to be chosen below. 
Therefore, gathering together \eqref{eq:J-23-1/2} and \eqref{eq:J-23-2}, we get the bound
\begin{equation*}
\begin{split}
\left(\int dy\,\|\mathcal{J}_{2,3}^{(1)}\|_{\rm HS}^2\|\mathcal{J}_{2,3}^{(1)}\|_{\rm HS}^\frac{1}{2}\right)^\frac{2}{5}&\leq C\,\sqrt{N}\,\e^{\frac{23}{10}-\frac{6}{5}\gamma}\,r^{\frac{6}{5}\gamma-\frac{2}{5}}\left(\int_0^1 \frac{ds}{[s(1-s)]^\frac{3}{2}\gamma-1}\right)^\frac{2}{5}\\
&\quad\quad\quad\quad\quad\quad\times\left( \int dv\,\sup_X\,(1+2X^2)^2|f_t(X,v)|^2 \right)^\frac{1}{10}\\
&\quad\quad\quad\quad\quad\quad\times\left(\int dX\left(\int dv\,(1+2X^2)\,|f_t(X,v)|\right)^2\right)^{\frac{2}{5}\gamma}\\
&\quad\quad\quad\quad\quad\quad\times\left(\int dX\int dv\,(1+2X^2)^2|f_t(X,v)|^2\right)^{\frac{2}{5}(1-\gamma)}
\end{split}
\end{equation*}
%

%
%
%
We choose $\gamma=\frac{3}{4}-\delta$, for $\delta\in(0,3/4)$. Thus the integral in the $s$ variable is bounded and we obtain 
\be
\begin{split}
\left(\int dy \|\mathcal{J}_{2,3}^{(1)}\|_{\rm HS}^\frac{5}{2}\right)^{\frac{2}{5}}&\leq C\,\sqrt{N}\,\e^{\frac{7}{5}+\frac{3}{5}\delta}\,r^{\frac{1}{2}-\frac{6}{5}\delta}\left(\int dv\,\sup_X\,(1+2X^2)^2|f_t(X,v)|^2\right)^\frac{1}{10}\\
&\quad\quad\quad\quad\quad\quad\quad\quad\quad\times\,\left(\int dX\left(\int dv\,(1+2\,X^2)|f_t(X,v)|\right)^2\right)^{\frac{3}{10}-\frac{2}{5}\delta} \\
&\quad\quad\quad\quad\quad\quad\quad\quad\quad\times\,\left(\int dX\int dv\,(1+2\,X^2)^2|f_t(X,v)|^2\right)^{\frac{1}{10}+\frac{2}{5}\delta} 
\end{split}
\ee
Following the exact same lines of the proof for the bound on $\mathcal{J}_{2,3}^{(1)}$, we obtain
\be
\begin{split}
\left(\int dy \|\mathcal{J}_{2,3}^{(2)}\|_{\rm HS}^\frac{5}{2}\right)^{\frac{2}{5}}&\leq C\,\sqrt{N}\,\e^{\frac{17}{5}+\frac{3}{5}\delta}\,r^{\frac{1}{2}-\frac{6}{5}\delta}\left(\int dv\,\sup_X\,|\nabla_v^2 f_t(X,v)|^2\right)^\frac{1}{10}\\
&\quad\quad\quad\quad\quad\quad\quad\quad\quad\times\,\left(\int dX\left(\int dv\,|\nabla_v^2 f_t(X,v)|\right)^2\right)^{\frac{3}{10}-\frac{2}{5}\delta} \\
&\quad\quad\quad\quad\quad\quad\quad\quad\quad\times\,\left(\int dX\int dv\,|\nabla^2_v f_t(X,v)|^2\right)^{\frac{1}{10}+\frac{2}{5}\delta} 
\end{split}
\ee
To bound  
\begin{displaymath}
\begin{split}
\mathcal{J}_{2,4}:=& \frac{2\,N\e^2}{r^2}\left|\int_0^1 ds\,(1+x^2)\,\chi_{(r/\sqrt{s},y)}(x)\chi_{(r/\sqrt{1-s})}(x')\,\int dv\,f_t\left(\frac{x+x'}{2},v\right)\,e^{iv\cdot\frac{(x-x')}{\e}}\right|
\end{split}
\end{displaymath}
we notice that the function $\chi_{(r/\sqrt{s},y)}(x)$ scales exactly as 
\be\label{eq:scale-chi}
\chi_{(r/\sqrt{s},y)}(x)\frac{(x-y)^j}{r^j/\sqrt{s^{j}}},
\ee
 for every $j\in\N$. Proposition \ref{prop:gauss-int} allows us to mimic here the estimates we have performed to achieve the bound on $\mathcal{J}_{2,3}$, thus leading to \eqref{eq:J-2k} for $k=4$.
\end{proof}

\begin{lemma}\label{lemma:J-25}
For $\delta\in(0,3/4)$ and the vector valued function $$g_t(x,v)=\nabla_x\widetilde{W}_{N,t}(x,v)+\frac{\e}{2}\nabla^3_{v,x,x}\widetilde{W}_{N,t}(x,v)+\nabla^2_{v,x}\left(v\,\widetilde{W}_{N,t}(x,v)\right),$$ it holds
\begin{displaymath}
\begin{split}
\left(\int dy\,\|\mathcal{J}_{2,5}\|^\frac{5}{2}_{\rm HS}\right)^{\frac{2}{5}}&\leq C\,\sqrt{N}\,\e^{\frac{7}{5}+\frac{3}{5}\delta}\,r^{\frac{3}{2}-\frac{6}{5}\delta}\left(\int dX\left(\int dv\,(1+2\,X^2)\,|g_t(X,v)|\right)^2\right)^{\frac{3}{10}-\frac{2}{5}\delta} \\
&\quad\quad\quad\quad\quad\quad\quad\quad\quad\times\,\left(\int dX\int dv\,(1+2\,X^2)^2|g_t(X,v)|^2\right)^{\frac{1}{10}+\frac{2}{5}\delta} \\
&\quad\quad\quad\quad\quad\quad\quad\quad\quad\times\,\left(\int dv\,\sup_X (1+2\,X^2)^2|g_t(X,v)|^2\right)^{\frac{1}{10}} \\
&+C\,\sqrt{N}\,\e^{\frac{17}{5}+\frac{3}{5}\delta}\,r^{\frac{3}{2}-\frac{6}{5}\delta}\left(\int dX\left(\int dv\,|\nabla_v^2 g_t(X,v)|\right)^2\right)^{\frac{3}{10}-\frac{2}{5}\delta} \\
&\quad\quad\quad\quad\quad\quad\quad\quad\quad\times\,\left(\int dX\int dv\,|\nabla^2_v g_t(X,v)|^2\right)^{\frac{1}{10}+\frac{2}{5}\delta}\\
&\quad\quad\quad\quad\quad\quad\quad\quad\quad\times\,\left(\int dv\,\sup_X |\nabla^2_v g_t(X,v)|^2\right)^{\frac{1}{10}} 
\end{split}
\end{displaymath}
\end{lemma}
\begin{proof}
The proof is analogous to the one of Lemma \ref{lemma:J-2k} with the following modifications: 
 when performing Jensen inequality, use the measure $({2\sqrt{s}})^{-1}ds$;
notice that $\chi_{(r/\sqrt{s},y)}(x)\,\frac{(x-y)^2}{r^2/s}$ scales as $\chi_{(r/\sqrt{s},y)}(x)\,\frac{(x-y)}{r/\sqrt{s}}$ and apply Proposition \ref{prop:gauss-int}.
\end{proof}
\begin{lemma}\label{lemma:J-22}
For $\delta\in(0,3/4)$, 
\begin{displaymath}
\begin{split}
\left(\int dy\,\|\mathcal{J}_{2,2}\|_{\rm HS}^\frac{5}{2}\right)^{\frac{2}{5}}&\leq C\,\sqrt{N}\,\e^{\frac{7}{5}+\frac{3}{5}\delta}\,r^{\frac{1}{2}-\frac{6}{5}\delta}\left(\int dv\,\sup_X\,(1+2X^2)^2|\widetilde{W}_{N,t}(X,v)|^2\right)^\frac{1}{10}\\
&\quad\quad\quad\quad\quad\quad\quad\quad\quad\times\,\left(\int dX\left(\int dv\,(1+2\,X^2)\,|\widetilde{W}_{N,t}(X,v)|\right)^2\right)^{\frac{3}{10}-\frac{2}{5}\delta} \\
&\quad\quad\quad\quad\quad\quad\quad\quad\quad\times\,\left(\int dX\int dv\,(1+2\,X^2)^2|\widetilde{W}_{N,t}(X,v)|^2\right)^{\frac{1}{10}+\frac{2}{5}\delta}  \\
&+C\,\sqrt{N}\,\e^{\frac{17}{5}+\frac{3}{5}\delta}\,r^{\frac{1}{2}-\frac{6}{5}\delta}\left(\int dv\,\sup_X\,|\nabla_v^2\widetilde{W}_{N,t}(X,v)|^2\right)^\frac{1}{10}   \\
&\quad\quad\quad\quad\quad\quad\quad\quad\quad\times\,\left(\int dX\left(\int dv\,|\nabla_v^2 \widetilde{W}_{N,t}(X,v)|\right)^2\right)^{\frac{3}{10}-\frac{2}{5}\delta} \\
&\quad\quad\quad\quad\quad\quad\quad\quad\quad\times\,\left(\int dX\int dv\,|\nabla^2_v \widetilde{W}_{N,t}(X,v)|^2\right)^{\frac{1}{10}+\frac{2}{5}\delta}
\end{split}
\end{displaymath}
\end{lemma}
\begin{proof}
We observe that, with respect to the other terms, $\mathcal{J}_{2,2}$ contains an extra $r^{-1}$, thus it is more singular. To get reed of such a singularity and reduce the problem to one of the terms that have been already estimated in the previous Lemmas, we notice that we do not need to extract an extra $\e$ using \eqref{eq:v-der} in $\mathcal{J}_{2,2}$ (indeed, the Laplacian comes with a factor $\e^2$ in front that is enough to our purposes). Therefore we observe that for every $y\in\R^3$
\be\label{eq:triangular}
|x-x'|\leq |x-y|+|x'-y|\,.
\ee
Then the proof proceeds essentially following the same lines of the proof of Lemma \ref{lemma:J-2k}, up to some subtile differences. More precisely, we use \eqref{eq:mom-variables} to define 
\begin{displaymath}
\begin{split}
&\|\mathcal{J}_{2,2}^{(1)}\|_{\rm HS}\\
&:=\frac{4\,N\,\e^2}{r^2}\left\{\int dx\int dx'\left|\int_0^1 \sqrt{s}\,\chi_{(r/\sqrt{s},y)}(x)\,\chi_{(r/\sqrt{1-s},y)}(x')\,\left[\frac{(x-y)^3}{r^3/\sqrt{s^3}}+2\frac{(x-y)}{r/\sqrt{s}}\right]\cdot\frac{(x-x')}{r}\right.\right.\\
&\quad\quad\quad\quad\quad\quad\quad\quad\quad\quad\quad\quad\quad\quad\quad  \left.\left. \int dv\left(1+2\left(\frac{x+x'}{2}\right)^2\right)\,\widetilde{W}_{N,t}\left(\frac{x+x'}{2},v\right)\,e^{iv\cdot\frac{(x-x')}{\e}}\right|^2\right\}^{\frac{1}{2}}    
\end{split}
\end{displaymath}
\begin{displaymath}
\begin{split}
&\|\mathcal{J}_{2,2}^{(2)}\|_{\rm HS} \\
&:=\frac{4\,N\,\e^4}{r^2}\left\{\int dx\int dx'\left|\int_0^1 \sqrt{s}\,\chi_{(r/\sqrt{s},y)}(x)\,\chi_{(r/\sqrt{1-s},y)}(x')\,\left[\frac{(x-y)^3}{r^3/\sqrt{s^3}}+2\frac{(x-y)}{r/\sqrt{s}}\right]\cdot\frac{(x-x')}{r}\right.\right.\\
&\quad\quad\quad\quad\quad\quad\quad\quad\quad\quad\quad\quad\quad\quad\quad\quad\quad\quad\quad\quad\quad\quad\quad  \left.\left. \int dv\,\nabla_v^2\widetilde{W}_{N,t}\left(\frac{x+x'}{2},v\right)\,e^{iv\cdot\frac{(x-x')}{\e}}\right|^2\right\}^{\frac{1}{2}} 
\end{split}
\end{displaymath}
so that
\be
\|\mathcal{J}_{2,2}\|_{\rm HS}\leq C\|\mathcal{J}_{2,2}^{(1)}\|_{\rm HS}+C\|\mathcal{J}_{2,2}^{(2)}\|_{\rm HS}
\ee
We first focus on $\mathcal{J}_{2,2}^{(1)}$. By using the triangular inequality \eqref{eq:triangular}, we obtain the bound
\be\label{eq:J-22-1}
\begin{split}
\mathcal{J}_{2,2}^{(1)}&\leq \sum_{j=2,4}\frac{C\,N\,\e^2}{r^2}\left|\int_0^1 ds
\,\chi_{(r/\sqrt{s},y)}(x)\,\frac{|x-y|^j}{r^j/\sqrt{s^j}}\,\chi_{(r/\sqrt{1-s},y)}(x')\right.\\
&\quad\quad\quad\quad\quad\quad\quad\quad\quad\left.\times\int dv\,\left(1+\left(\frac{x+x'}{2}\right)^2\right)\,\widetilde{W}_{N,t}\left(\frac{x+x'}{2},v\right)\,e^{iv\cdot\frac{(x-x')}{\e}}\right|\\
&+\sum_{j=1,3}\frac{C\,N\,\e^4}{r^2}\left|\int_0^1 ds\,\frac{\sqrt{s}}{\sqrt{1-s}}
\,\chi_{(r/\sqrt{s},y)}(x)\,\frac{|x-y|^j}{r^j/\sqrt{s^j}}\,\chi_{(r/\sqrt{1-s},y)}(x')\,\frac{|x'-y|}{r/\sqrt{1-s}}\right.\\
&\quad\quad\quad\quad\quad\quad\quad\quad\quad\left.\times\int dv\,\left(1+\left(\frac{x+x'}{2}\right)^2\right)\,\nabla_v^2\widetilde{W}_{N,t}\left(\frac{x+x'}{2},v\right)\,e^{iv\cdot\frac{(x-x')}{\e}}\right|\\
\end{split}
\ee
 The first sum on the r.h.s. in \eqref{eq:J-22-1} can be bounded following the same lines of the proof of Lemma \ref{lemma:J-2k} thanks to Proposition \ref{prop:gauss-int}. As for the second sum, we note that
 $$
 \int_0^1 ds\,\frac{\sqrt{s}}{\sqrt{1-s}}=2
 $$
 Therefore we apply Jensen's inequality with the measure $\frac{\sqrt{s}}{2\sqrt{1-s}}ds$. Proceeding as in Lemma \ref{lemma:J-2k}, we obtain
 \be
 \begin{split}
 \left(\int dy\,\|\mathcal{J}_{2,2}^{(1)}\|_{\rm HS}^\frac{5}{2}\right)^{\frac{2}{5}}&\leq  C\,\sqrt{N}\,\e^{\frac{7}{5}+\frac{3}{5}\delta}\,r^{\frac{1}{2}-\frac{6}{5}\delta}\left(\int dv\,\sup_X\,(1+2X^2)^2|\widetilde{W}_{N,t}(X,v)|^2\right)^\frac{1}{10}\\
&\quad\quad\quad\quad\quad\quad\quad\quad\quad\times\,\left(\int dX\left(\int dv\,(1+2\,X^2)\,|\widetilde{W}_{N,t}(X,v)|\right)^2\right)^{\frac{3}{10}-\frac{2}{5}\delta} \\
&\quad\quad\quad\quad\quad\quad\quad\quad\quad\times\,\left(\int dX\int dv\,(1+2\,X^2)^2|\widetilde{W}_{N,t}(X,v)|^2\right)^{\frac{1}{10}+\frac{2}{5}\delta}  
\end{split}
\ee
$\mathcal{J}_{2,2}^{(2)}$ can be handled analogously, hence the bound
 \be
 \begin{split}
 \left(\int dy\,\|\mathcal{J}_{2,2}^{(2)}\|_{\rm HS}^\frac{5}{2}\right)^{\frac{2}{5}}\leq
&C\,\sqrt{N}\,\e^{\frac{17}{5}+\frac{3}{5}\delta}\,r^{\frac{1}{2}-\frac{6}{5}\delta}\left(\int dv\,\sup_X\,|\nabla_v^2\widetilde{W}_{N,t}(X,v)|^2\right)^\frac{1}{10}   \\
&\quad\quad\quad\quad\quad\quad\quad\quad\quad\times\,\left(\int dX\left(\int dv\,|\nabla_v^2 \widetilde{W}_{N,t}(X,v)|\right)^2\right)^{\frac{3}{10}-\frac{2}{5}\delta} \\
&\quad\quad\quad\quad\quad\quad\quad\quad\quad\times\,\left(\int dX\int dv\,|\nabla^2_v \widetilde{W}_{N,t}(X,v)|^2\right)^{\frac{1}{10}+\frac{2}{5}\delta}
\end{split}
 \ee
\end{proof}
\begin{lemma}\label{lemma:J-32}
For $\delta\in(0,3/4)$ it holds
\be\label{eq:J-32}
\begin{split}
\left(\int dy\,\|\mathcal{J}_{3,2}\|^\frac{5}{2}_{\rm HS}\right)^{\frac{2}{5}}&\leq C\,\sqrt{N}\,\e^{\frac{7}{5}+\frac{3}{5}\delta}\,r^{\frac{1}{2}-\frac{6}{5}\delta}\left(\int dv\,\sup_X\,(1+2X^2)^2|\widetilde{W}_{N,t}(X,v)|^2\right)^\frac{1}{10}\\
&\quad\quad\quad\quad\quad\quad\quad\quad\quad\times\,\left(\int dX\left(\int dv\,(1+2\,X^2)\,|\widetilde{W}_{N,t}(X,v)|\right)^2\right)^{\frac{3}{10}-\frac{2}{5}\delta} \\
&\quad\quad\quad\quad\quad\quad\quad\quad\quad\times\,\left(\int dX\int dv\,(1+2\,X^2)^2|\widetilde{W}_{N,t}(X,v)|^2\right)^{\frac{1}{10}+\frac{2}{5}\delta}  \\
&+C\,\sqrt{N}\,\e^{\frac{17}{5}+\frac{3}{5}\delta}\,r^{\frac{1}{2}-\frac{6}{5}\delta}\left(\int dv\,\sup_X\,|\nabla_v^2\widetilde{W}_{N,t}(X,v)|^2\right)^\frac{1}{10}   \\
&\quad\quad\quad\quad\quad\quad\quad\quad\quad\times\,\left(\int dX\left(\int dv\,|\nabla_v^2 \widetilde{W}_{N,t}(X,v)|\right)^2\right)^{\frac{3}{10}-\frac{2}{5}\delta} \\
&\quad\quad\quad\quad\quad\quad\quad\quad\quad\times\,\left(\int dX\int dv\,|\nabla^2_v \widetilde{W}_{N,t}(X,v)|^2\right)^{\frac{1}{10}+\frac{2}{5}\delta}
\end{split}
\ee
\end{lemma}
\begin{proof}
As for the previous terms, we start by splitting the integral into two parts according to \eqref{eq:mom-variables}, thus defining 
\begin{displaymath}
\begin{split}
&\|\mathcal{J}_{3,2}^{(1)}\|_{\rm HS}\\
&:=\sum_{j=0,2}\frac{C\,N\,\e^2}{r^2}\left\{\int dx\int dx'\left|\int_0^1 \frac{2\,ds}{\sqrt{1-s}}\,s\,\chi_{(r/\sqrt{s},y)}(x)\,\frac{(x-y)^j}{r^j/\sqrt{s^j}}\right.\right.\\
&\quad\quad\quad\quad\quad\quad\quad\quad\quad\quad\quad\quad\quad\quad  \left.\left.\times\,\chi_{(r/\sqrt{1-s},y)}(x')\,\frac{(x'-y)}{r/\sqrt{1-s}}\cdot\frac{(x-x')}{r}\right.\right.\\
&\quad\quad\quad\quad\quad\quad\quad\quad\quad\quad\quad\quad\quad\quad  \left.\left. \int dv\left(1+2\left(\frac{x+x'}{2}\right)^2\right)\,\widetilde{W}_{N,t}\left(\frac{x+x'}{2},v\right)\,e^{iv\cdot\frac{(x-x')}{\e}}\right|^2\right\}^{\frac{1}{2}}    
\end{split}
\end{displaymath}
\begin{displaymath}
\begin{split}
&\|\mathcal{J}_{3,2}^{(2)}\|_{\rm HS} \\
&:=\sum_{j=0,2}\frac{C\,N\,\e^4}{r^2}\left\{\int dx\int dx'\left|\int_0^1 \frac{2\,ds}{\sqrt{1-s}}\,s\,\chi_{(r/\sqrt{s},y)}(x)\,\frac{(x-y)^j}{r^j/\sqrt{s^j}}\right.\right.\\
&\quad\quad\quad\quad\quad\quad\quad\quad\quad\quad\quad\quad\quad\quad\quad\quad\quad\quad\quad\quad\quad  \left.\left.\times\,\chi_{(r/\sqrt{1-s},y)}(x')\,\frac{(x'-y)}{r/\sqrt{1-s}}\cdot\frac{(x-x')}{r}\right.\right.\\
&\quad\quad\quad\quad\quad\quad\quad\quad\quad\quad\quad\quad\quad\quad\quad\quad\quad\quad\quad\quad\quad  \left.\left. \int dv\,\nabla_v^2\widetilde{W}_{N,t}\left(\frac{x+x'}{2},v\right)\,e^{iv\cdot\frac{(x-x')}{\e}}\right|^2\right\}^{\frac{1}{2}} 
\end{split}
\end{displaymath}
so that
\be
\|\mathcal{J}_{3,2}\|_{\rm HS}\leq C\|\mathcal{J}_{3,2}^{(1)}\|_{\rm HS}+C\|\mathcal{J}_{3,2}^{(2)}\|_{\rm HS}
\ee
We proceed as in Lemma \ref{lemma:J-22}.
For every $y\in\R^3$ and $s\in(0,1)$
\be\label{eq:convexinterp+1}
|x-x'|\leq |x-y|+|x'-y|
\ee
Thus, plugging \eqref{eq:convexinterp+1} into $\mathcal{J}_{3,2}^{(1)}$ and $\mathcal{J}_{3,2}^{(2)}$ we obtain respectively
\begin{displaymath}
\begin{split}
\mathcal{J}_{3,2}^{(1)}&\leq\sum_{j=0,2}\frac{C\,N\,\e^2}{r^2}\left|\int_0^1 \frac{2\,ds}{\sqrt{1-s}}\,\sqrt{s}\,\chi_{(r/\sqrt{s},y)}(x)\,\frac{|x-y|^{j+1}}{r^{j+1}/\sqrt{s^{j+1}}}\,\chi_{(r/\sqrt{1-s},y)}(x')\,\frac{|x'-y|}{r/\sqrt{1-s}}\right.\\
&\quad\quad\quad\quad\quad\quad\quad\quad\quad\quad\quad\quad\quad\quad  \left. \int dv\left(1+2\left(\frac{x+x'}{2}\right)^2\right)\,\widetilde{W}_{N,s}\left(\frac{x+x'}{2},v\right)\,e^{iv\cdot\frac{(x-x')}{\e}}\right| \\
&+\sum_{j=0,2}\frac{C\,N\,\e^2}{r^2}\left|\int_0^1 ds\,2\,s\,\chi_{(r/\sqrt{s},y)}(x)\,\frac{|x-y|^{j}}{r^{j}/\sqrt{s^{j}}}\,\chi_{(r/\sqrt{1-s},y)}(x')\,\frac{|x'-y|^2}{r^2/(1-s)}\right.\\
&\quad\quad\quad\quad\quad\quad\quad\quad\quad\quad\quad\quad\quad\quad  \left. \int dv\left(1+2\left(\frac{x+x'}{2}\right)^2\right)\,\widetilde{W}_{N,t}\left(\frac{x+x'}{2},v\right)\,e^{iv\cdot\frac{(x-x')}{\e}}\right|   
\end{split}
\end{displaymath}
and
\begin{displaymath}
\begin{split}
\mathcal{J}_{3,2}^{(2)}&\leq\sum_{j=0,2}\frac{C\,N\,\e^4}{r^2}\left|\int_0^1 \frac{2\,ds}{\sqrt{1-s}}\,\sqrt{s}\,\chi_{(r/\sqrt{s},y)}(x)\,\frac{|x-y|^{j+1}}{r^{j+1}/\sqrt{s^{j+1}}}\,\chi_{(r/\sqrt{1-s},y)}(x')\,\frac{|x'-y|}{r/\sqrt{1-s}}\right.\\
&\ \ \quad\quad\quad\quad\quad\quad\quad\quad\quad\quad\quad\quad\quad\quad\quad\quad\quad\quad\quad\quad\quad  \left. \int dv\,\nabla_v^2\widetilde{W}_{N,t}\left(\frac{x+x'}{2},v\right)\,e^{iv\cdot\frac{(x-x')}{\e}}\right| \\
&+\sum_{j=0,2}\frac{C\,N\,\e^4}{r^2}\left|\int_0^1 ds\,2\,s\,\chi_{(r/\sqrt{s},y)}(x)\,\frac{|x-y|^{j}}{r^{j}/\sqrt{s^{j}}}\,\chi_{(r/\sqrt{1-s},y)}(x')\,\frac{|x'-y|^2}{r^2/(1-s)}\right.\\
&\ \ \quad\quad\quad\quad\quad\quad\quad\quad\quad\quad\quad\quad\quad\quad\quad\quad\quad\quad\quad\quad\quad  \left. \int dv\,\nabla_v^2\widetilde{W}_{N,t}\left(\frac{x+x'}{2},v\right)\,e^{iv\cdot\frac{(x-x')}{\e}}\right|
\end{split}
\end{displaymath}
For $k=1,2$, the first sum in $\mathcal{J}_{3,2}^{(k)}$ can be treated as in Lemma \ref{lemma:J-2k} using Jensen's inequality with measure $\frac{\sqrt{s}}{2\sqrt{1-s}}ds$, while we can deal with the second sum in $\mathcal{J}_{3,2}^{(k)}$ by performing Jensen's inequality with measure $ds$ and then proceed by following the same lines in the proof of Lemma \ref{lemma:J-2k}. The key observation to adapt the proof of Lemma \ref{lemma:J-2k} to the terms we are considering now is that the functions
\be\label{eq:scale-j}
\chi_{(r/\sqrt{s},y)}(x)\frac{(x-y)^j}{r^j/\sqrt{s^j}}\quad\mbox{and}\quad\chi_{(r/\sqrt{1-s},y)}(x')\frac{(x'-y)^j}{r^j/\sqrt{(1-s)^j}} 
\ee
scale in the same way for all $j\in\N$. Proposition \ref{prop:gauss-int} allows to proceed as in the proof of Lemma \ref{lemma:J-2k}.\\
To sum up, we collect all the terms and obtain the desired bound \eqref{eq:J-32}.
\end{proof}
\begin{lemma}\label{lemma:J-34}
For every $\delta\in(0,3/4)$ and $f_t=\widetilde{W}_{N,t}+\e\nabla^2_{v,x}\widetilde{W}_{N,t}+v\cdot\nabla_v\widetilde{W}_{N,t}$, the following bound holds
\be\label{eq:J-34}
\begin{split}
\left(\int dy \|\mathcal{J}_{3,4}\|_{\rm HS}^\frac{5}{2}\right)^{\frac{2}{5}}&\leq C\,\sqrt{N}\,\e^{\frac{7}{5}+\frac{3}{5}\delta}\,r^{\frac{1}{2}-\frac{6}{5}\delta}\left(\int dv\,\sup_X\,(1+2X^2)^2|f_{t}(X,v)|^2\right)^\frac{1}{10}\\
&\quad\quad\quad\quad\quad\quad\quad\quad\quad\times\,\left(\int dX\left(\int dv\,(1+2\,X^2)\,|f_{t}(X,v)|\right)^2\right)^{\frac{3}{10}-\frac{2}{5}\delta} \\
&\quad\quad\quad\quad\quad\quad\quad\quad\quad\times\,\left(\int dX\int dv\,(1+2\,X^2)^2|f_{t}(X,v)|^2\right)^{\frac{1}{10}+\frac{2}{5}\delta}  \\
&+C\,\sqrt{N}\,\e^{\frac{17}{5}+\frac{3}{5}\delta}\,r^{\frac{1}{2}-\frac{6}{5}\delta}\left(\int dv\,\sup_X\,|\nabla_v^2 f_{t}(X,v)|^2\right)^\frac{1}{10}   \\
&\quad\quad\quad\quad\quad\quad\quad\quad\quad\times\,\left(\int dX\left(\int dv\,|\nabla_v^2 f_{t}(X,v)|\right)^2\right)^{\frac{3}{10}-\frac{2}{5}\delta} \\
&\quad\quad\quad\quad\quad\quad\quad\quad\quad\times\,\left(\int dX\int dv\,|\nabla^2_v f_{t}(X,v)|^2\right)^{\frac{1}{10}+\frac{2}{5}\delta}
\end{split}
\ee
\end{lemma}
\begin{proof}
The bound is obtained by repeating the same proof of Lemma \ref{lemma:J-2k}, thanks to the observation \eqref{eq:scale-j} and Proposition \ref{prop:gauss-int}.
\end{proof}
\begin{lemma}\label{lemma:J-33}
For $\delta\in(0,3/4)$ and the vector valued function $$g_t(x,v)=\nabla_x\widetilde{W}_{N,t}(x,v)+\frac{\e}{2}\nabla^3_{v,x,x}\widetilde{W}_{N,t}(x,v)+\nabla^2_{v,x}\left(v\,\widetilde{W}_{N,t}(x,v)\right),$$ it holds
\begin{displaymath}
\begin{split}
\left(\int dy\,\|\mathcal{J}_{3,3}\|^\frac{5}{2}_{\rm HS}\right)^{\frac{2}{5}}&\leq C\,\sqrt{N}\,\e^{\frac{7}{5}+\frac{3}{5}\delta}\,r^{\frac{3}{2}-\frac{6}{5}\delta}\left(\int dv\,\sup_X\,(1+2X^2)^2|g_{t}(X,v)|^2\right)^\frac{1}{10}\\
&\quad\quad\quad\quad\quad\quad\quad\quad\quad\times\,\left(\int dX\left(\int dv\,(1+2\,X^2)\,|g_{t}(X,v)|\right)^2\right)^{\frac{3}{10}-\frac{2}{5}\delta} \\
&\quad\quad\quad\quad\quad\quad\quad\quad\quad\times\,\left(\int dX\int dv\,(1+2\,X^2)^2|g_t(X,v)|^2\right)^{\frac{1}{10}+\frac{2}{5}\delta}  \\
&+C\,\sqrt{N}\,\e^{\frac{17}{5}+\frac{3}{5}\delta}\,r^{\frac{3}{2}-\frac{6}{5}\delta}\left(\int dv\,\sup_X\,|\nabla_v^2 g_{t}(X,v)|^2\right)^\frac{1}{10}   \\
&\quad\quad\quad\quad\quad\quad\quad\quad\quad\times\,\left(\int dX\left(\int dv\,|\nabla_v^2 g_t(X,v)|\right)^2\right)^{\frac{3}{10}-\frac{2}{5}\delta} \\
&\quad\quad\quad\quad\quad\quad\quad\quad\quad\times\,\left(\int dX\int dv\,|\nabla^2_v g_t(X,v)|^2\right)^{\frac{1}{10}+\frac{2}{5}\delta}
\end{split}
\end{displaymath}
\end{lemma}
\begin{proof}
Thanks to \eqref{eq:scale-j} and Proposition \ref{prop:gauss-int}, we follow the same lines in the proof of Lemma \ref{lemma:J-25}, thus achieving the desired bound.
\end{proof}


\subsection{Estimates on the subleading term}\label{subsect:subleading}
\begin{proposition}\label{prop:B-tr}
Under the same assumptions of Theorem \ref{thm:trace}, there exists a positive constant $C$ depending on $T$, $\|\widetilde\rho_0\|_{L^1}$ and $(H3)$ such that
$$\tr\ |B_s|\leq C\,N\,\e^2\left[\|\widetilde{W}_{N,s}\|_{H^2_4}+\e\,\|\widetilde{W}_{N,s}\|_{H^3_4}+\e^2\,\|\widetilde{W}_{N,s}\|_{H^4_4}+\e^3\,\|\widetilde{W}_{N,s}\|_{H^5_4}+\e^4\|\widetilde{W}_{N,s}\|_{H^6_4}\right]$$
\end{proposition}
\begin{proof}
We proceed as done in Proposition \ref{prop:1term-tr} Eq. \eqref{eq:trick} to get
\begin{equation}\label{eq:tr-B} 
\tr \, |B_s| \leq \| (1-\eps^2\Delta)^{-1} (1+x^2)^{-1} \|_{\rm HS} \, \| (1+x^2) (1-\eps^2 \Delta) B_s \|_{\rm HS} \leq C \sqrt{N} \| (1+x^2) (1-\eps^2\Delta) B_s \|_{\rm HS} 
\end{equation}
We denote  by $U_s$ the convolution of the interaction with the spatial density at time $s$
$$U_s:= \frac{1}{|\cdot|^\alpha}* \widetilde{\rho}_s.$$ 
With this notation, the kernel of the operator $\widetilde{B} := (1-\eps^2 \Delta) B_s$ reads 
\be\label{eq:termsB}
\begin{split}
\widetilde{B} &(x;x')\\
 &= N \left[U_s(x) - U_s(x') - \nabla U_s\left(\frac{x+x'}{2}\right)\cdot(x-x')\right]\int \widetilde{W}_{N,s}\left(\frac{x+x'}{2},v\right)e^{i\,v\cdot\frac{(x-x')}{\e}}dv   \\
& - N\e^2\left[\Delta U_s(x) - \frac{1}{4}\Delta \nabla U_s\left(\frac{x+x'}{2}\right)\cdot(x-x') - \frac{1}{2} \Delta U_s \left(\frac{x+x'}{2}\right)\right] \int \widetilde{W}_{N,s}\left(\frac{x+x'}{2},v\right)e^{i\,v\cdot\frac{(x-x')}{\e}}dv \\
& - \frac{N\e^2}{4} \left[U_s(x) - U_s(x') - \nabla U_s\left(\frac{x+x'}{2}\right)\cdot(x-x')\right] \int (\Delta_1 \widetilde{W}_{N,s}) \left(\frac{x+x'}{2},v\right)e^{i\,v\cdot\frac{(x-x')}{\e}}dv \\
& + N \left[U_s(x) - U_s(x') - \nabla U_s\left(\frac{x+x'}{2}\right)\cdot(x-x')\right] \int \widetilde{W}_{N,s}\left(\frac{x+x'}{2},v\right) v^2 e^{i\,v\cdot \frac{(x-x')}{\e}}dv \\
& - \frac{N\e^2}{2} \left[\nabla U_s(x) - \frac{1}{2} \nabla^2 U_s\left(\frac{x+x'}{2}\right) (x-x') - \nabla U_s\left(\frac{x+x'}{2}\right)\right] \int (\nabla_1 \widetilde{W}_{N,s}) \left(\frac{x+x'}{2},v\right) e^{i\,v\cdot \frac{(x-x')}{\e}}dv \\
& - N\e \left[\nabla U_s(x) - \frac{1}{2} \nabla^2 U_s\left(\frac{x+x'}{2}\right)(x-x') - \nabla U_s\left(\frac{x+x'}{2}\right)\right] \int \widetilde{W}_{N,s}\left(\frac{x+x'}{2},v\right) v e^{i\,v\cdot \frac{(x-x')}{\e}}dv \\
& - N\e \left[U_s(x) - U_s(x') - \nabla U_s\left(\frac{x+x'}{2}\right)\cdot(x-x') \right] \int (v\cdot \nabla_1 \widetilde{W}_{N,s}) \left(\frac{x+x'}{2},v\right) e^{i\,v\cdot \frac{(x-x')}{\e}}dv\\
 & =: \sum_{j=1}^7 \widetilde{B}_j (x;x')  
\end{split}
\ee
where $\nabla_1$ and $\Delta_1$ stand for derivatives with respect to the first variable.

The terms $\widetilde B_1, \widetilde B_4, \widetilde B_6, \widetilde B_7$ come with powers of $\e$ which are not enough to our purpose (actually, we need in the end $\e^a$ where $a>1$). For this reason we write  
\begin{equation*}
\begin{split} 
&U_s (x) - U_s (x') - \nabla U_s \left( \frac{x+x'}{2} \right) \cdot (x-x') \\ 
&= \int_0^1 d\lambda \left[ \nabla U_s \left(\lambda x + (1-\lambda) x'\right) - \nabla U_s \left((x+x')/2\right) \right] \cdot (x-x')  \\ 
&= \sum_{i,j=1}^3\int_0^1 d\lambda\left(\lambda-\frac{1}{2}\right) \int_0^1 d\mu\,\partial_i \partial_j U_s\left(\mu (\lambda x + (1-\lambda) x') + (1-\mu) \frac{(x+x')}{2}\right) (x-x')_i (x-x')_j.
\end{split} 
\end{equation*}
Being $U_s$ a convolution, we can perform derivatives on $\widetilde{\rho}_s$, that is bounded thanks to Assumptions $(H1)-(H5)$. Therefore, the generalised Feffermann - de la Llave representation formula for $\frac{1}{\ |\,\cdot\,|^\alpha}$ leads to
\be\label{eq:U-taylor}
\begin{split}
U_s &(x) - U_s (x') - \nabla U_s \left( \frac{x+x'}{2} \right) \cdot (x-x') \\
&= \sum_{i,j=1}^3\int_0^1 d\lambda\,\left(\lambda-\frac{1}{2}\right) \int_0^1 d\mu\,\int_0^\infty \frac{dr}{r^{1+\alpha}}\\
&\ \times\int dy\,\chi_{(r,y)}(\mu (\lambda x + (1-\lambda) x') + (1-\mu) (x+x')/2)\partial_i\partial_j\widetilde{\rho_s}(y)\, (x-x')_i (x-x')_j.
\end{split}
\ee
Plugging \eqref{eq:U-taylor} into the definition of $\widetilde{B}_1$ and using twice Eq. \eqref{eq:v-der} and Young inequality, we get
\begin{equation*}
\begin{split}
|\widetilde{B}_1&(x;x')|\\
\leq & C\,N\,\e^2\sum_{i,j=1}^3\int_0^1d\lambda\,\left|\lambda-\frac{1}{2}\right|\int_0^1 d\mu\left| \int_0^\infty\frac{dr}{r^{1+\alpha}}\right.\\
\times&\left.\int dy\,\chi_{(r,y)}(\mu(\lambda x+(1-\lambda)x')+(1-\mu)(x+x')/2)\partial^2_{i,j}\widetilde{\rho}_s(y)\int dv\,\partial^2_{v_i,v_j}\widetilde{W}_{N,s}\left(\frac{x+x'}{2},v\right)e^{iv\cdot\frac{(x-x')}{\e}} \right|
\end{split}
\end{equation*}
We aim to bound
\be
\begin{split}
\|(1+x^2)&\widetilde{B}_1\|_{\rm HS}^2\\
=& N\e^4\int dX\int dX'\left[1+X^2+\e^2 (X')^2\right]^2\left|\sum_{i,j=1}^3\int_0^1d\lambda\,\left(\lambda-\frac{1}{2}\right)\int_0^1 d\mu \int_0^\infty\frac{dr}{r^{1+\alpha}}\right.\\
\times&\left.\int dy\,\chi_{(r,y)}(X+\e\mu(\lambda-1/2)X')\partial^2_{v_i,v_j}\widetilde{\rho}_s(y)\int dv\,\partial^2_{v_i,v_j}\widetilde{W}_{N,s}\left(X,v\right)e^{iv\cdot X'} \right|^2
\end{split}
\ee
where we performed the change of variables \eqref{eq:change-var}.
For a fixed $k>0$, we split the integral in the $r$ variable into two parts, $r\in(0,k)$ and $r\in(k,\infty)$, and we estimate them separately. Thus, for $r\in(0,k)$ we use Young inequality, recall that Assumptions $(H1)-(H5)$ imply $\|\nabla^2\rho_s\|_{L^\infty}\leq C$ and then integrate in the $y$ variable to extract $r^3$ which cancels the singularity, thus leading to the bound 
\be\label{eq:HS-B1-r0}
CN\e^4\int dX\int dv(1+X^2)^2|\nabla^2_v\widetilde{W}_{N,s}(X,v)|^2+CN\e^8\int dX\int dv\,|\nabla^4_v\widetilde{W}_{N,s}(X,v)|^2
\ee
where $C$ depends on $\|\nabla^2\widetilde{\rho}_s\|_{L^\infty}$.\\
For $r\in(k,\infty)$, we integrate by parts in the $y$ variable twice and recall that  $e^{-|z-y|^2/r^2}(1+|z-y|^2/r^2)\leq C$ for every $z\in\R^3$. Since $\widetilde{\rho}_s\in L^1(\R^3)$ we get the bound
\be\label{eq:HS-B1-rinf}
CN\e^4\int dX\int dv(1+X^2)^2|\nabla^2_v\widetilde{W}_{N,s}(X,v)|^2+CN\e^8\int dX\int dv\,|\nabla^4_v\widetilde{W}_{N,s}(X,v)|^2
\ee
where $C$ depends on $\|\widetilde{\rho}_s\|_{L^1}$.\\
Whence, gathering together the two estimates \eqref{eq:HS-B1-r0}, \eqref{eq:HS-B1-rinf}, we get
\be\label{eq:bound-HS-tildeB1}
\|(1+x^2)\widetilde{B}_1\|_{\rm HS}
\leq C\sqrt{N}\e^2\|\widetilde{W}_{N,s}\|_{H_2^2}
+C\sqrt{N}\e^4\|\widetilde{W}_{N,s}\|_{H^4}
\ee
where $C=C(\|\widetilde{\rho}_s\|_{L^1},\|\nabla^2\widetilde{\rho}_s\|_{L^\infty})$.

The Hilbert-Schmidt norms $\|(1+x^2)\widetilde{B}_3\|_{\rm HS}$, $\|(1+x^2)\widetilde{B}_4\|_{\rm HS}$ and $\|(1+x^2)\widetilde{B}_7\|_{\rm HS}$ can be estimated in the exact same way, thus obtaining 
\be
\|(1+x^2)\widetilde{B}_3\|_{\rm HS}\leq C\sqrt{N}\e^4\|\widetilde{W}_{N,s}\|_{H_4^4}
+C\sqrt{N}\e^6\|\widetilde{W}_{N,s}\|_{H_4^6}
\ee
\be
\|(1+x^2)\widetilde{B}_4\|_{\rm HS}\leq C\sqrt{N}\e^2\|\widetilde{W}_{N,s}\|_{H_4^2}
+C\sqrt{N}\e^4\|\widetilde{W}_{N,s}\|_{H_4^4}
\ee
\be
\|(1+x^2)\widetilde{B}_7\|_{\rm HS}\leq C\sqrt{N}\e^3\|\widetilde{W}_{N,s}\|_{H_3^2}
+C\sqrt{N}\e^5\|\widetilde{W}_{N,s}\|_{H_2^5}
\ee
We now estimate $\widetilde{B}_6$, in which we have to deal with higher order derivatives of $U$. Proceeding as for $\widetilde{B}_1$, we first use \eqref{eq:U-taylor} and then divide the integral in the $r$ variables into two parts:
\be\label{eq:tildeB6}
\begin{split}
|\widetilde{B}_6&(x;x')|\\
\leq & C\,N\,\e^3\int_0^1d\lambda\,\left|\lambda-1/2\right|\int_0^1 d\mu\left| \int_0^k\frac{dr}{r^{1+\alpha}}\right.\\
&\quad\quad\quad\quad\quad\quad\times\left.\int dy\,\nabla\chi_{(r,y)}(\mu(\lambda x+(1-\lambda)x')+(1-\mu)(x+x')/2)\nabla^2\widetilde{\rho}_s(y)\right.\\
&\quad\quad\quad\quad\quad\quad\times\left.\int dv\,\nabla^2_{v}\widetilde{W}_{N,s}\left(\frac{x+x'}{2},v\right)e^{iv\cdot\frac{(x-x')}{\e}} \right|\\
+&C\,N\,\e^3\int_0^1d\lambda\,\left|\lambda-1/2\right|\int_0^1 d\mu\left| \int_k^\infty\frac{dr}{r^{1+\alpha}}\right.\\
&\quad\quad\quad\quad\quad\quad\times\left.\int dy\,\nabla^3\chi_{(r,y)}(\mu(\lambda x+(1-\lambda)x')+(1-\mu)(x+x')/2)\widetilde{\rho}_s(y)\right.\\
&\quad\quad\quad\quad\quad\quad\times\left.\int dv\,\nabla^2_{v}\widetilde{W}_{N,s}\left(\frac{x+x'}{2},v\right)e^{iv\cdot\frac{(x-x')}{\e}} \right|\\
\end{split}
\ee
where in the second term we have integrated by parts twice in the $y$ variable. 

We aim to bound
$\|(1+x^2)\widetilde{B}_6\|_{\rm HS}^2$. Therefore we consider the first term in \eqref{eq:tildeB6}, perform the change of variables \eqref{eq:change-var} and choose $k$ so that $\int_0^k r^{-\alpha}dr=1$. Then we can apply Young inequality with measure  $r^{-\alpha}dr$ and we get the bound
\be\label{eq:HS-B6-r0}
\begin{split}
C\,N\,\e^6&\int dX\int dX'[1+X^2+\e^2(X')^2]^2\int_0^1d\lambda\,|\lambda-1/2|^2\int_0^1 d\mu \int_0^{k}\frac{dr}{r^\alpha}\frac{1}{r^4}\\
&\quad\quad\times\int dy\,\chi_{(r,y)}(X+\e\mu(\lambda-1/2)X')\frac{|X+\e\mu(\lambda-1/2)X'-y|}{r}\\
&\quad\quad\times\int dy'\,\chi_{(r,y')}(X+\e\mu(\lambda-1/2)X')\frac{|X+\e\mu(\lambda-1/2)X'-y'|}{r}\\
&\quad\quad\times\iint dv\,dv'\,\nabla^2_{v}\widetilde{W}_{N,s}(X,v)\,\nabla^2_{v'}\widetilde{W}_{N,s}(X,v')\,e^{i(v-v')\cdot X'} \\
\leq C\,N&\,\e^6\|\widetilde{W}_{N,s}\|_{H^2_2}^2+C\,N\,\e^{10}\|\widetilde{W}_{N,s}\|_{H^4}^2
\end{split}
\ee
where $C$ depends on $\|\nabla^2\widetilde{\rho}_s\|_{L^\infty}$.\\
For $r\in(k,\infty)$, we perform the change of variables \eqref{eq:change-var} and recall that  $e^{-|z-y|^2/r^2}(|z-y|^k/r^k)\leq C$ for every $z\in\R^3$ and $k\in\N$. Since $\widetilde{\rho}_s\in L^1(\R^3)$ we get the bound
\be\label{eq:HS-B6-rinf}
C\,N\,\e^6\|\widetilde{W}_{N,s}\|_{H^2_2}^2+C\,N\,\e^{10}\|\widetilde{W}_{N,s}\|_{H^4}^2
\ee
where $C$ depends on $\|\widetilde{\rho}_s\|_{L^1}$.\\
Gathering together the two estimates \eqref{eq:HS-B6-r0}, \eqref{eq:HS-B6-rinf}, we get
\be\label{eq:bound-HS-tildeB6}
\|(1+x^2)\widetilde{B}_6\|_{\rm HS}
\leq C\sqrt{N}\e^3\|\widetilde{W}_{N,s}\|_{H_2^2}
+C\sqrt{N}\e^5\|\widetilde{W}_{N,s}\|_{H^4}
\ee
where $C=C(\|\widetilde{\rho}_s\|_{L^1},\|\nabla^2\widetilde{\rho}_s\|_{L^\infty})$.

The norms $\|(1+x^2)\widetilde{B}_j\|_{\rm HS}$, $j=2,5$, can be dealt analogously, thus obtaining
\be\label{eq:bound-HS-tildeB6}
\|(1+x^2)\widetilde{B}_2\|_{\rm HS}
\leq C\sqrt{N}\e^4\|\widetilde{W}_{N,s}\|_{H_2^2}
+C\sqrt{N}\e^6\|\widetilde{W}_{N,s}\|_{H^4}
\ee
and 
\be\label{eq:bound-HS-tildeB6}
\|(1+x^2)\widetilde{B}_5\|_{\rm HS}
\leq C\sqrt{N}\e^4\|\widetilde{W}_{N,s}\|_{H_2^4}
+C\sqrt{N}\e^6\|\widetilde{W}_{N,s}\|_{H^6}
\ee
where $C=C(\|\widetilde{\rho}_s\|_{L^1},\|\nabla^2\widetilde{\rho}_s\|_{L^\infty})$.

The final result is:
\be \| (1 + x^{2}) \widetilde{B} \|_\text{HS} \leq C \sqrt{N} \left[ \eps^2 \| \widetilde{W}_{N,s} \|_{H^2_{4}} + \eps^3  \| \widetilde{W}_{N,s} \|_{H^3_{4}} + \eps^4 \| \widetilde{W}_{N,s} \|_{H^4_{4}} + \eps^5 \| \widetilde{W}_{N,s} \|_{H^5_4}+\eps^6 \| \widetilde{W}_{N,s} \|_{H^6_4} \right] 
\ee
Assumptions $(H1)-(H5)$ imply
\be\label{eq:estimateB}
\tr \, |B_s| \leq C  N \eps^2 \left[ \|W_N\|_{H_4^2}  + \eps \| W_N \|_{H^3_{4}} + \eps^2 \| W_N \|_{H^4_4} + \eps^3 \| W_N \|_{H_4^5}+\e^6\|W_N\|_{H_4^6} \right]
\ee
where $C$ depends on $T$, $\|\widetilde\rho_0\|_{L^1}$ and assumption $(H3)$ for $k=0$ and $l=3$.

\end{proof}

\subsection{Theorem \ref{thm:trace}: end of the proof}\label{subsect:end}

To conclude we apply Proposition \ref{prop:vlasov-pde} to bound the weighted Sobolev norms appearing in the estimates proved in Proposition \ref{prop:1term-tr} and Proposition \ref{prop:B-tr} in terms of  the initial data $W_N$ and constants depending on $T$. Let us denote by $C_i'$, for $i=1,\dots,8$, the bounds on $C_i$ (defined in Proposition \ref{prop:1term-tr}) obtained through Proposition \ref{prop:vlasov-pde}. Analogously, let us denote by $D_i=\|\widetilde{W}_{N,s}\|_{H_4^{2+i}}$ for $i=0,\dots,4$, and by $D_i'$ the bounds on $D_i$ obtained through Proposition \ref{prop:vlasov-pde}. Hence, gathering together Proposition \ref{prop:1term-tr} and Proposition \ref{prop:B-tr}, Lemma \ref{lem:est-tr} leads to
\begin{equation*}
\begin{split}
\tr\ |\omega_{N,t}-\widetilde{\omega}_{N,t}|\leq& C_1'\int_0^t \tr\ |\omega_{N,s}-\widetilde{\omega}_{N,s}|\\
&+ N\,\e^{\frac{2}{5}}\left[C_2'\e^\frac{3}{5}+C_3'\e^{\frac{3}{5}\delta}+C_4'\e^{2+\frac{3}{5}}+C_5'\e^{\frac{9}{10}}+C_6'\e^\frac{7}{5}+C_7'\e^{\frac{29}{10}}+C_8'\e^\frac{17}{5}\right.\\
&\quad\quad\quad\quad\quad\quad\quad\quad\quad\quad+\left. D_0'\e^\frac{3}{5}+D_1'\e^\frac{8}{5}+D_2'\e^\frac{13}{5}+D_3'\e^\frac{18}{5}+D_4'\e^\frac{23}{5}\right]
\end{split}
\end{equation*}
Hence, applying Gr\"{o}nwall lemma  we get
\begin{equation*}
\begin{split}
\tr\ |\omega_{N,t}-\widetilde{\omega}_{N,t}|\leq e^{C_1't}\,N\,\e^\frac{2}{5}&\left[C_2'\e^{\frac{3}{5}}+C_3'\e^{\frac{3}{5}\delta}+C_4'\e^{2+\frac{3}{5}\delta}+C_5'\e^{\frac{9}{10}}+C_6'\e^{\frac{7}{5}}+C_7'\e^{\frac{29}{10}}+C_8'\e^{\frac{17}{5}}\right.\\
&\quad\quad\quad\quad\quad\quad\quad\quad\left.+D_0'\e^{\frac{3}{5}}+D_1'\e^{\frac{8}{5}}+D_2'\e^{\frac{13}{5}}+D_3'\e^{\frac{18}{5}}+D_4'\e^{\frac{23}{5}}     \right]
\end{split}
\end{equation*}
which concludes the proof.

\section{Hilbert-Schmidt norm convergence: proof of Theorem \ref{thm:HS}}\label{section:HS}

Recalling \eqref{eq:omega-omegatilde} and taking its Hilbert-Schmidt norm, we easily obtain
\be
\begin{split}
\|\omega_{N,t}-\tilde{\omega}_{N,t}\|_{\rm HS}&\leq\frac{1}{\e}\int_0^t ds\,\int dy\,|\rho_s(y)-\tilde{\rho}_s(y)|\int\frac{dr}{r^{1+\alpha}}\|[\chi_{r,y},\tilde{\omega}_{N,s}]\|_{\rm HS}+\frac{1}{\e}\int_0^t ds\,\|B_s\|_{\rm HS}\\
&\leq \frac{1}{N\e}\int_0^t ds\,\tr\ |\omega_{N,s}-\tilde{\omega}_{N,s}|\int\frac{dr}{r^{1+\alpha}}\|[\chi_{r,y},\tilde{\omega}_{N,s}]\|_{\rm HS}+\frac{1}{\e}\int_0^t ds\,\|B_s\|_{\rm HS}
\end{split}
\ee
where we have used Proposition \ref{prop:FDLL} in the first inequality and \eqref{eq:L1-tr} in the second bound. For $r$ close to zero we use Lemma \ref{lemma:fund-est}, which implies  
\be\label{eq:commutator-HS}
\|[\chi_{r,y},\widetilde{\omega}_{N,s}]\|_{\rm HS}\leq \|\mathcal{J}_1\|_{\rm HS}\leq \sqrt{N}\e\sqrt{r}(C+C'\e^2)
\ee
where in the second inequality we used Lemma \ref{lemma:J1} and $C,\,C'$ are  positive constants depending only on $T>0$ and on the initial data $W_N$. For $r$ at infinity we bound $\|\mathcal{J}_1\|_{\rm HS}$ as in \eqref{eq:J1-infty}. Moreover, the Hilbert-Schmidt norm of $B_s$ in \eqref{eq:HS-esti} can be handled following the same steps in the proof of Proposition \ref{prop:B-tr}. Of course here we are in a favourable situation because we have two derivatives less than in the case of Proposition \ref{prop:B-tr} and no spatial moments to control, thus the bound on $\|B_s\|_{\rm HS}$ basically reduced in a bound on $\|B_1\|_{\rm HS}$ (where we have used the same notations as in the proof of Proposition \ref{prop:B-tr}). 

The assumptions on $W_N$ allow us to apply Theorem \ref{thm:trace}  and to get
\be\label{eq:HS-esti}
\begin{split}
\|\omega_{N,t}-\widetilde{\omega}_{N,t}\|_{\rm HS}\leq e^{C_1't}\,\sqrt{N}\,\e^{\frac{2}{5}}&\left[ C_2'\e^{\frac{3}{5}}+C_3'\e^{\frac{3}{5}\delta}+C_4'\e^{2+\frac{3}{5}\delta}+C_5'\e^{\frac{9}{10}}+C_6'\e^{\frac{7}{5}}+C_7'\e^{\frac{29}{10}}+C_8'\e^{\frac{17}{5}}\right.\\
&\quad+D_0'\e^{\frac{3}{5}}+D_1'\e^{\frac{8}{5}}+D_2'\e^{\frac{13}{5}}+D_3'\e^{\frac{18}{5}}+D_4'\e^{\frac{23}{5}}\\
 &\quad+E_1\e^{4+\frac{3}{5}\delta}+E_2\e^{\frac{49}{10}}+E_3\e^{\frac{27}{5}}+E_4\e^{\frac{13}{5}}+E_5'\e^{\frac{18}{5}}\\
 &\quad\left.+E_6'\e^\frac{23}{5}+E_7'\e^\frac{28}{5}+E_8'\e^\frac{33}{5}
     \right]
\end{split}
\ee
where $C_i'$, $D_j'$, $E_i'$, $i=1,\dots,8$, $j=0,\dots,4$,  are constants depending on $T>0$, on the assumptions on the initial data $W_N$ and its derivatives up to order 6.

\section{$L^2$-norm convergence: proof of Theorem \ref{thm:L2}}\label{section:L2}

To estimate $\|W_{N,t}-W_t\|_{2}$ we use $\widetilde{W}_{N,t}$, solution to \eqref{eq:Valpha} with initial data $W_N$, as an intermediate step. Indeed, by triangular inequality we have 
\be\label{eq:l2-estimate}
\| W_{N,t}-W_t\|_{L^2}\leq \|W_{N,t}-\widetilde{W}_{N,t}\|_{L^2}+\|\widetilde{W}_{N,t}-W_t\|_{L^2}
\ee
Since $\|W_{N,t}-\widetilde{W}_{N,t}\|_{L^2}=\frac{C}{\sqrt{N}}\|\omega_{N,t}-\widetilde{\omega}_{N,s}\|_{\rm HS}$, the first term on the r.h.s. of \eqref{eq:l2-estimate} is bounded by applying Theorem \ref{thm:HS}, that is
\be\label{eq:l2-1}
\begin{split}
\|W_{N,t}-\widetilde{W}_{N,t}\|_{L^2}\leq  e^{C_1't}\,\e^{\frac{2}{5}}&\left[ C_2'\e^{\frac{3}{5}}+C_3'\e^{\frac{3}{5}\delta}+C_4'\e^{2+\frac{3}{5}\delta}+C_5'\e^{\frac{9}{10}}+C_6'\e^{\frac{7}{5}}+C_7'\e^{\frac{29}{10}}+C_8'\e^{\frac{17}{5}}\right.\\
&\quad+D_0'\e^{\frac{3}{5}}+D_1'\e^{\frac{8}{5}}+D_2'\e^{\frac{13}{5}}+D_3'\e^{\frac{18}{5}}+D_4'\e^{\frac{23}{5}}\\
 &\quad+E_1\e^{4+\frac{3}{5}\delta}+E_2\e^{\frac{49}{10}}+E_3\e^{\frac{27}{5}}+E_4\e^{\frac{13}{5}}+E_5'\e^{\frac{18}{5}}\\
 &\quad\left.+E_6'\e^\frac{23}{5}+E_7'\e^\frac{28}{5}+E_8'\e^\frac{33}{5}
     \right].
\end{split}
\ee 
The second term on the r.h.s. of \eqref{eq:l2-estimate} can be handled using a stability argument for solutions of the Vlasov Eq. \eqref{eq:Valpha} satisfying assumptions $(H1)-(H5)$. Such a statement has been proved in the context of regular interactions in \cite{BPSS}. In the case of singular potentials minor modifications of such a proof lead to
\be\label{eq:l2-2}
\|\widetilde{W}_{N,t}-W_t\|_{L^2}\leq C(k_{N,1}+k_{N,2})
\ee
where $C$ is a positive constant depending on $T>0$ and on the initial data.

The key observation to obtain such a bound consists in noticing that  $(H3)$  implies boundedness on derivatives of the Lagrangian flow associated to \eqref{eq:Valpha}, as proved in Appendix \ref{appendix:regularity} Remark \ref{rk:v-der-cond}.

Gathering together \eqref{eq:l2-1} and \eqref{eq:l2-2}, we obtain the desired bound. 

\appendix

\section{Regularity estimates for the Vlasov equation}\label{appendix:regularity}

We will treat the Vlasov Eq. with inverse power law potential $1/|x|^{\alpha}$, $\alpha\in(0,1/2)$, as a Vlasov-Poisson system in which the Poisson Eq. is replaced by the $p$-Poisson Eq. with $p=\frac{3-\alpha}{2}>1$. This observation allows to reproduce entirely the proofs in \cite{LP} up to minor modifications. \\
To this end, we rephrase here in an extensive way the assumptions of Theorem \ref{thm:trace}:
\begin{itemize}
\item[ ] {\bf Assumption 1.} (Existence condition) $\tilde{W}_{N,0}\geq 0$, $\tilde{W}_{N,0}\in L^1\cap L^\infty(\R^3\times\R^3)$ and $\mathcal{H}_0$ (see \eqref{eq:energy}) finite;
\item[ ] {\bf Assumption 2.} (Existence condition) There exists $m_0>\frac{3\alpha}{2-\alpha}$ such that if $m<m_0$
\be\label{eq:hyp-moments}
\iint |v|^m\,\tilde{W}_{N,0}(x,v)\,dx\,dv<+\infty\,.
\ee
\item[ ] {\bf Assumption 3.} (Uniqueness condition) For all $R,\,T>0$
\be\label{eq:hyp-sup} 
\sup\{| \tilde{W}_{N,0}(y+vt,w)|\,:\,|y-x|\leq R\,t^2,\,|w-v|\leq R\,t\}\in L^\infty((0,T)\times\R^3_x;L^1(\R^3_v))\,,
\ee
\be\label{eq:hyp-supp}
\ \ \ \sup\{|\nabla \tilde{W}_{N,0}(y+vt,w)|\,:\,|y-x|\leq r,\,|w-v|\leq r\}\in L^\infty((0,T)\times\R^3_x;L^1\cap L^2(\R^3_v))\,,
\ee
where $\nabla \tilde{W}_{N,0}$ denotes the $x$ and $v$ gradient of $\tilde{W}_{N,0}$
\item[ ] {\bf Assumption 4.} (Regularity assumption) For all $T,\,R>0$, 
\be\label{eq:regularity-1}
\begin{split}
\sup\{(1+x^2+v^2)^k|\nabla^l\widetilde{W}_0|(y+tw)\,:\,|y-x|\leq &R\,t^2,\,|w-v|\leq R\,t\}\\
&\in L^\infty((0,T)\times\R^3_x;L^1\cap L^2(\R^3_v))\,,
\end{split}
\ee
where $l=0,\dots, 6$ and $k=0,2$
\item[ ] {\bf Assumption 5.} (Regularity assumption) There exist two positive constants $C$ and $C'$  such that
\be\label{eq:regularity-2}
\|\widetilde{W}_0\|_{H^6_4}\leq C,\quad\quad\quad \|(1+x^8+v^n)\widetilde{W}_{0}\|_{W^{l,1}}\leq C'
\ee
where $l=0,1$ and $n=3,5$.
\end{itemize}



\begin{remark}\label{rk:v-der-cond}
Observe that condition \eqref{eq:hyp-supp} implies $\nabla_v\widetilde{W}_{N,t}\in L^\infty([0,T]\times\R^3_x;L^2(\R^3_v))$. 
Indeed, consider the solution $s\mapsto (x(s), v(s))$ of 
        \begin{equation} \label{ode2}
        \left\{     \begin{split} (\partial_s x)(s) &= v(s),\hspace{0.5cm} &x(t)=x\in\mathbb{R}^3, \\
                          (\partial_s v)(s) &= E(s, x(s)),\hspace{0.5cm} &v(t)=v\in\mathbb{R}^3,       
                    \end{split}
        \right.
        \end{equation} 
so that the solution $\widetilde{W}_{N,t}$ of the Vlasov Eq. \eqref{eq:Valpha} is given pointwise by 
\be\label{eq:lagrangian}
\widetilde{W}_{N,t}(x,v)= \widetilde{W}_{0}(x(0),v(0)),\;\forall\; (t,x,v)\in [0,T]\times\mathbb{R}^6
\ee
where $ t\mapsto (x(0),v(0))\in\mathbb{R}^6$ is a shorthand notation for the backward flow.  Moreover, we recall that when $\alpha=1$ it has been proved  in \cite{LP} that $E,\nabla_x E\in L^{\infty}([0,T], C^{0,\beta})$ s.t. 
        \be\label{eq:holder-est} \|E\|_{L^{\infty}([0,T], C^{0,\beta})},\; \|\nabla_x E\|_{L^{\infty}([0,T], C^{0,\beta})}\leq R\ee 
for some fixed $ R>0$, $\beta\in(0,1)$. For $\alpha\in(0,1/2)$ Sobolev embeddings apply to $E(t,x)=\left(\nabla\frac{1}{|\cdot|^\alpha}*\rho\right)(t,x)$, since $\frac{1}{|x|^{\alpha+1}}$ can be treated as a Riesz potential.  Thus  \eqref{eq:holder-est} holds as well for $\alpha\in(0,1/2)$.

We first bound $\frac{\partial x(s)}{\partial v} $ and $\frac{\partial v(s)}{\partial v}  $ in norm, uniformly in the initial data $ x,v\in \mathbb{R}^3$ and in time. We integrate (\ref{ode2}) and differentiate w.r.t. $v$ s.t.
        \[\begin{split}
        \Big\|\frac{\partial x(s)}{\partial v} \Big\|_{\max}&\leq (s-t) + \int_{t}^{s}d\sigma\;\int_{t}^{\sigma}d\tau\; \Big\|\nabla_x E(\tau, x(\tau))\cdot \frac{\partial x(\tau)}{\partial v}\Big\|_{\max} \\
        &\leq  T + RT \int_{t}^{s}d\tau\; \Big\|\frac{\partial x(\tau)}{\partial v}\Big\|_{\max}
        \end{split}\]
as well as
        \[\Big\|\frac{\partial v(s)}{\partial v} \Big\|_{\max}\leq  R\int_{t}^{s}d\tau\; \Big\|\frac{\partial x(\tau)}{\partial v}\Big\|_{\max}\]
The first bound implies that $$ \sup_{s\in[0,T]}\left\|\frac{\partial x(s)}{\partial v} \right\|_{\max}\leq C_{R,T}\exp(C_{R,T}T)$$ for some constant $ C_{R,T}~>~0$, by Gronwall's inequality. Inserting this into the second bound, we also conclude that $\frac{\partial v(s)}{\partial v} $ is uniformly bounded in norm. 

To conclude  we use the chain rule s.t.
        \[\begin{split} (\nabla_v \widetilde{W}_{N,t})(x,v)  =&\; (\nabla_x \widetilde{W}_0)\big (x+ tv + (x(0)-x-tv), v+ (v(0)-v) \big)\cdot \frac{\partial x(0)}{\partial v} \\
        &+ (\nabla_v \widetilde{W}_0)\big (x+ tv + (x(0)-x-tv), v+ (v(0)-v) \big)\cdot \frac{\partial v(0)}{\partial v}  \end{split}\]
As a simple consequence of integrating (\ref{ode2}) and using the uniformly in time boundedness of the electric field $E$, we get 
        \[|x(0)-x-tv|\leq Rt^2, \hspace{0.5cm} |v(0)-v|\leq R t   \]
With the uniform bounds on the matrix norms of $\partial x(s)/\partial v $ and $\partial v(s)/\partial v  $, this implies 
        \[\big|(\nabla_v \widetilde{W}_{N,t})(x,v)\big|\leq \widetilde {C}_{R,T} \sup_{y,w\in\mathbb{R}^3}\Big \{\big|\nabla_z \widetilde{W}_0(y+tv, w)\big|: |y-x|\leq Rt^2, |w-v|\leq Rt  \Big\}\]
for some constant $\widetilde {C}_{R,T} $ depending only on $R$ and $T$. Now,  the function that maps $ (t,x,v)\in [0,T]\times\mathbb{R}^6$ to the r.h.s. of the last line is an element that lies in particular in $L^{\infty}([0,T]\times \mathbb{R}^3_x, L^{2}(\mathbb{R}^3_v))$. Hence, we conclude 
        \[\nabla_v \widetilde{W}_{N,t}\in L^{\infty}([0,T]\times \mathbb{R}^3_x, L^{2}(\mathbb{R}^3_v)).\]

Analogously, \eqref{eq:regularity-1} with $k=0$ implies $\nabla^l_v\widetilde{W}_{N,t}\in L^\infty([0,T]\times\R^3_x;L^2(\R^3_v))$ for $l=1,\dots,5$. Indeed, the bounds \eqref{eq:holder-est} for higher order derivatives are guaranteed by Assumption 4 and Schauder estimates (see \cite{krylov}).
\end{remark}

\begin{remark}
We underline that the constraint $m_0>3\alpha/(2-\alpha)$ comes from Sobolev embeddings. Proceeding as in \cite{Pallard}, it is possible to relax this assumption. However finding minimal assumptions on the initial data falls outside the purpose of this paper. 
\end{remark}

Before stating our regularity result, we need the following interpolation inequalities 

\begin{proposition}\label{prop:interpolation}
Let $m\geq 0$ and $\widetilde{W}_{N,t}\in L^1\cap L^\infty(\R^3\times \R^3)$ solution to \eqref{eq:Valpha}, such that assumptions 2 and 4 with $k=0$ hold for all $r,\,T>0$.
Then there exist a constant $c>0$, which only depends on $m$ and $\|\tilde{W}_{N,t}\|_{L^\infty_{x,v}}$, and a constant $C>0$, which only depends on $m$ and  $\|\nabla^l_v\tilde{W}_{N,t}\|_{L^\infty_{t,x}(L^2_v)}$, such that
\be\label{eq:interpolation-Linfty}
\left\|\widetilde{\rho}_t\right\|_{L^{\frac{m+3}{3}}}\leq c\left(\iint |v|^m |\widetilde{W}_{N,t}(x,v)|\,dx\,dv\right)^{\frac{3}{m+3}}\,.
\ee
and
\begin{equation}\label{eq:interpolation}
\left\|\int|\nabla^l_v\widetilde{W}_{N,t}(\cdot,v)|\,dv\,\right\|_{L^{\frac{2m+3}{3}}}\leq C\left(\iint |v|^m |\nabla^l_v\tilde{W}_{N,t}(x,v)|\,dx\,dv\right)^{\frac{3}{2m+3}}\,.
\end{equation}
\end{proposition}
\begin{proof}
We first prove \eqref{eq:interpolation-Linfty}. By definition, 
\begin{equation}\label{eq:norm-rho}
\left\|  \widetilde{\rho}_t \right\|_{L^{\frac{m+3}{3}}}=\left( \int \left|\int |\widetilde{W}_{N,t}(x,v)|\,dv\right|^{\frac{m+3}{3}}dx \right)^{\frac{3}{m+3}}\,.
\end{equation}
Fix $R>0$ and split the integral in the $v$ variable into two pieces 
\begin{equation*}
\begin{split}
\int |\widetilde{W}_{N,t}(x,v)|\,dv&=\int_{|v|\leq R}|\widetilde{W}_{N,t}(x,v)|\,dv+\int_{|v|>R}|\widetilde{W}_{N,t}(x,v)|\,dv\\
&\leq R^3 \|\widetilde{W}_{N,t}\|_{L^{\infty}}+\frac{1}{R^m}\int |v|^m |\widetilde{W}_{N,t}(x,v)|\,dv\,.
\end{split}
\end{equation*} 
By optimising in $R$ in the last line of the above inequality,  we get
\begin{equation}\label{eq:bound-rhox}
\int |\widetilde{W}_{N,t}(x,v)|\,dv\leq c\left(\int |v|^m |\widetilde{W}_{N,t}(x,v)|\,dv\right)^\frac{3}{m+3}\,.
\end{equation}
where $c$ depends on powers of the $L^\infty$ norm of $\widetilde{W}_s$. We plug \eqref{eq:bound-rhox} in \eqref{eq:norm-rho} and we obtain
\begin{equation}
\left\|\int|\widetilde{W}_{N,t}(\cdot,v)|\,dv\,\right\|_{L^{\frac{m+3}{3}}}\leq c\left(\iint |v|^m |\widetilde{W}_{N,t}(x,v)|\,dv\,dx\right)^{\frac{3}{m+3}}\,.
\end{equation}

To prove \eqref{eq:interpolation} we proceed analogously. By definition, 
\begin{equation}\label{eq:norm-nablaW}
\left\|  \int|\nabla^l_v\widetilde{W}_{N,t}(\cdot,v)|\,dv\, \right\|_{L^{\frac{2m+3}{3}}}=\left( \int \left|\int |\nabla^l_v\widetilde{W}_{N,t}(x,v)|\,dv\right|^{\frac{2m+3}{3}}dx \right)^{\frac{3}{2m+3}}\,.
\end{equation}
Fix $R>0$ and split the integral in the $v$ variable into two pieces 
\begin{equation*}
\begin{split}
\int |\nabla^l_v\widetilde{W}_{N,t}(x,v)|\,dv&=\int_{|v|\leq R}|\nabla^l_v\widetilde{W}_{N,t}(x,v)|\,dv+\int_{|v|>R}|\nabla^l_v \widetilde{W}_{N,t}(x,v)|\,dv\\
&\leq R^{\frac{3}{2}}\|\nabla^l_v\widetilde{W}_{N,t}\|_{L^{\infty}_{t,x}(L^2_v)}+\frac{1}{R^m}\int |v|^m |\nabla^l_v\widetilde{W}_{N,t}(x,v)|\,dv\,.
\end{split}
\end{equation*} 
By optimising in $R$ in the last line of the above inequality,  we get
\begin{equation}\label{eq:bound-nablarhox}
\int |\nabla^l_v\widetilde{W}_{N,t}(x,v)|\,dv\leq C\left(\int |v|^m |\nabla^l_v\widetilde{W}_{N,t}(x,v)|\,dv\right)^\frac{3}{2m+3}\,.
\end{equation}
We plug \eqref{eq:bound-nablarhox} in \eqref{eq:norm-nablaW} and we obtain
\begin{equation}
\left\|\int|\nabla^l_v\widetilde{W}_{N,t}(\cdot,v)|\,dv\,\right\|_{L^{\frac{2m+3}{3}}}\leq C\left(\iint |v|^m |\nabla^l_v\widetilde{W}_{N,t}(x,v)|\,dv\,dx\right)^{\frac{3}{2m+3}}\,.
\end{equation}
\end{proof}

\begin{proposition}\label{prop:vlasov-pde}
Let Assumptions 1 and 2 hold true.
\begin{itemize}
\item[i)] Then, there exists a solution to the Vlasov system \eqref{eq:V-P-alpha} $\widetilde{W}_{N,t}\in\mathcal{C}(\R^+;L^p(\R^3\times\R^3))\cap L^{\infty}(\R^+;L^\infty(\R^3\times\R^3))$, for all $p\in [1,\infty)$.
\item[ii)] Suppose that also Assumption 3 holds, then the solution $\widetilde{W}_{N,t}$ is unique, with $E\in L^{\infty}((0,T);\mathcal{C}^{1,\beta}(\R^3_x))$, for all $\beta<1$.
\item[iii)] If moreover Assumptions 4 and 5 hold, then there exists a positive constant $C$ depending on $R$ and $T$ such that
\begin{eqnarray}
&&\int \sup_x\,(1+x^2+v^2)^4|\nabla_v\widetilde{W}_{N,t}(x,v)|^2\,dv\leq C; \label{eq:bound1}\\
&&\iint\,(1+x^2+v^2)^4|\nabla^l_v\widetilde{W}_{N,t}(x,v)|\,dv\,dx\leq C,\quad \quad\mbox{ for } l=0,1;\label{eq:bound2}\\
&&\sum_{|\beta|\leq 5}\iint (1+x^2+v^2)^4|\nabla^\beta \widetilde{W}_{N,t}(x,v)|^2\,dx\,dv\leq C\label{eq:bound3}
\end{eqnarray}
where the integrals are taken over $\R^3$.
\end{itemize}
\end{proposition}
\begin{remark}\label{rk:kin-energy}
We notice that, by conservation of energy \eqref{eq:energy-vlasov}, to require the initial energy to be bounded immediately implies the kinetic energy to be bounded, i.e. there exists a positive constant $C$ independent on time such that
\begin{equation*}
\int_{\R^3\times\R^3}{|v|^2}\,\widetilde{W}_t(x,v)dxdv\leq C.
\end{equation*} 
Proposition \ref{prop:interpolation} then implies $\rho_t\in L^{\frac{5}{3}}(\R^3)$.
\end{remark}
\begin{remark}
We recall that in our context the initial datum $W_N$ is obtained as the Wigner transform of a given trace class fermionic operator $\omega_N$ such that $0\leq \omega_N\leq 1$. This condition does not imply in general $W_N\geq 0$. This problem has been fixed in Appendix A of \cite{BPSS}, where the well-posedness of the Vlasov Eq. for signed measures is established. By standard approximation arguments, Appendix A in \cite{BPSS} applies to the case $V(x)=\frac{1}{|x|^\alpha}$. Namely, on the one hand we fix a sequence of mollify functions $\{\eta^\e\}_\e\subset \mathcal{C}^\infty_c(\R^3)$ and define the smooth potential $V^\e=\frac{1}{|\cdot|^{\alpha}} * \eta^\e$, which satisfies the assumptions in Appendix A of \cite{BPSS}. On the other hand, estimates \`a la Lions and Perthame (\cite{LP}) guarantee that all bounds are uniform $\e$, thus allowing to perform the limit $\e\to 0$. Details on such standard procedure can be found in \cite{golse-notes}.        
\end{remark}

{\it Sketch of the proof of Proposition \ref{prop:vlasov-pde}.}
Proposition \ref{prop:vlasov-pde} can be essentially deduced from \cite{LP}. For this reason we summarise here only the main but minor differences with respect to the proofs in \cite{LP} or we will give the key observations to proceed as in \cite{LP}.

 To prove $i)$, we follow line by line the proof of Theorem 1 in \cite{LP}. We replace the Poisson Eq. $-\Delta U=\rho$ by the $p$-Poisson Eq. $-\Delta^p U=\rho$, $p=\frac{3-\alpha}{2}$. Sobolev embeddings hold as well in this case (\cite{Sobolev}), since ${1}/{|x|^\alpha}$ can be seen as a Riesz potenital. Indeed, for $a\in(0,3)$, where $3$ represents the spacial dimension, we define the Riesz potential
 $I_a(x)=\frac{1}{|x|^{3-a}}$ and the operator $I_a$ acting on some Lebesgue measurable functions $g$ by convolution, i.e. for all $x\in\R^3$
 $$I_a g(x)=(I_a*g)(x)=\int_{\R^3}\frac{g(y)}{|x-y|^{3-a}}\,dy.$$
 Once $i)$ is proved, Assumption 3 guarantees boundedness on the spatial density $\rho_t$, thus uniqueness of the solution $\widetilde{W}_{N,t}$ obtained in $i)$. The proof is a simple adaptation of 
 Section 3 in \cite{LP}. In particular, $m_0>6$ implies by Sobolev embedding $E\in\mathcal{C}(\R^+;\mathcal{C}^{0,\gamma}(\R^3_x))$ for $\gamma\in(0,1)$. A bootstrap argument and Schauder estimates (see for instance \cite{krylov}) lead to $E\in L^\infty((0,T);\mathcal{C}^{1,\beta}(\R^3_x))$, for all $\beta<1$. 

The proof of $iii)$ makes use of Assumption 4 and 5 to gain regularity on the Lagrangian flow and on the solution $\widetilde{W}_{N,t}$. More precisely, 
Assumption 4 implies that 
\be\label{eq:bound-infty-2}
\|(1+x^2)^2\nabla_v(|v|^m\widetilde{W}_{N,s})\|_{L^2_v(L^\infty_x)}\leq C
\ee
for $m=0,1,2$. To see this, we need to use Young inequality in $(1+x^2)^2|v|^{2m}\leq C(1+|x|^2)^4+|v|^{4m}$ and to bound this latter by $C_1(1+|x|^8+|v|^{4m})$, for some constant $C_1$. Then, we recall that a solution of the Vlasov Eq. is transported along the flow \eqref{ode2}, that is \eqref{eq:lagrangian} holds.

We plug \eqref{eq:lagrangian} into \eqref{eq:bound-infty-2}. Then Eq. \eqref{eq:bound1} follows from Assumption 4 and the fact that higher order derivatives are bounded, in the same spirit of Remark \ref{rk:v-der-cond}. 

Following the same lines of Appendix B in \cite{BPSS}, the first bound in Assumption 5 implies $\|\widetilde{W}_{N,s}\|_{H_4^6}\leq C$, that is \eqref{eq:bound3}. To this end, we need to prove boundedness of derivatives of the flow. This follows from Assumption 4 and Schauder estimates. 

Moreover, the second bound in Assumption 5 guarantees 
$$\|(1+|x|^8+|v|^n)\nabla^l\widetilde{W}_{N,s}\|_{L^1}\leq C\quad\quad \mbox{for } n=3,5\ \mbox{and } l=0,1$$
and \eqref{eq:bound2} is proved.

\section{Useful integrals}\label{ap:int}

In this section, we prove a key estimate we have been using throughout Section \ref{section:trace}.
\begin{proposition}\label{prop:gauss-int}
Let $z,w\in\R^3$, $s\in[0,1]$ and $r\geq 0$. For some $j,\,k\,\in\,\N$ and some constant $C>0$, it holds
\begin{equation}\label{eq:integrals}
\begin{split}
\int_{\R^3} du\,e^{-s|z-u|^2/r^2}&\frac{|z-u|^j}{r^j/\sqrt{s^j}}\,e^{-(1-s)|w-u|^2/r^2}\frac{|w-u|^k}{r^k/\sqrt{(1-s)^k}}\\
&\leq C\,r^3\,s\,(1-s)\,e^{-s(1-s)|z-w|^2/r^2}\left(1+\frac{|z-w|^{j+k}}{r^{j+k}/\sqrt{s^{j+k}(1-s)^{j+k}}}\right)
\end{split}
\end{equation}
\end{proposition} 
\begin{proof}
First we compute explicitly the case $j=k=0$. Standard computations show that
\begin{equation*}
\int du\,e^{-s|z-u|^2/r^2}\,e^{-(1-s)|w-u|^2/r^2}=e^{-s(1-s)|z-w|^2/r^2}\int du\,e^{-|sz+(1-s)w-u|^2/r^2}.
\end{equation*}
The change of variable $z=(sz+(1-s)w-u)/r$ then gives 
$$
\int du\,e^{-s|z-u|^2/r^2}\,e^{-(1-s)|w-u|^2/r^2}=C\,r^3\,e^{-s(1-s)|z-w|^2/r^2}
$$
where $C$ is an explicit exactly computable constant. 

We now look at the case $j,\,k\in\N$. Using the above computations, we get
\begin{equation*}
\mbox{l.h.s. of } \eqref{eq:integrals}= e^{-s(1-s)|z-w|^2/r^2}\int du\,e^{-|sz+(1-s)w-u|^2/r^2}\frac{|z-u|^j}{r^j/\sqrt{s^j}}\,\frac{|w-u|^k}{r^k/\sqrt{(1-s)^k}}. 
\end{equation*}
Using that 
\begin{equation*}
\begin{split}
&|z-u|\leq |sz+(1-s)w-u|+(1-s)|z-w|\\
&|w-u|\leq|sz+(1-s)w-u|+s|z-w|
\end{split}
\end{equation*}
we obtain the bound
\begin{equation*}
\begin{split}
\int &du\,e^{-|sz+(1-s)w-u|^2/r^2}\frac{|z-u|^j}{r^j/\sqrt{s^j}}\,\frac{|w-u|^k}{r^k/\sqrt{(1-s)^k}}\\
&\leq C\int du\,e^{-|sz+(1-s)w-u|^2/r^2}\left(\frac{|sz+(1-s)w-u|^j}{r^j/\sqrt{s^j}}+(1-s)^j\frac{|z-w|^j}{r^j/\sqrt{s^j}}\right)\\
&\quad\quad\quad\quad\quad\quad\quad\quad\quad\quad\quad\quad\times\left(\frac{|sz+(1-s)w-u|^k}{r^k/\sqrt{(1-s)^k}}+s^k\frac{|z-w|^k}{r^k/\sqrt{(1-s)^k}}\right)
\end{split}
\end{equation*}
Therefore, straightforward computations lead to
\begin{equation*}
\begin{split}
\int &du\,e^{-|sz+(1-s)w-u|^2/r^2}\frac{|z-u|^j}{r^j/\sqrt{s^j}}\,\frac{|w-u|^k}{r^k/\sqrt{(1-s)^k}}\\
&\leq C\int du\,e^{-|sz+(1-s)w-u|^2/r^2}\frac{|sz+(1-s)w-u|^{j+k}}{r^{j+k}}\\
&+C\sqrt{s^{j+k}}\frac{|z-w|^k}{r^k/\sqrt{s^k(1-s)^k}}\int du\,e^{-|sz+(1-s)w-u|^2/r^2}\frac{|sz+(1-s)w-u|^{j}}{r^{j}}\\
&+C\sqrt{(1-s)^{j+k}}\frac{|z-w|^j}{r^j/\sqrt{s^j(1-s)^j}}\int du\,e^{-|sz+(1-s)w-u|^2/r^2}\frac{|sz+(1-s)w-u|^{k}}{r^{k}}\\
&+C\sqrt{s^{k}(1-s)^j}\frac{|z-w|^{j+k}}{r^{j+k}/\sqrt{s^{j+k}(1-s)^{j+k}}}\int du\,e^{-|sz+(1-s)w-u|^2/r^2}
\end{split}
\end{equation*}
and therefore the bound
\begin{equation*}
\begin{split}
\int &du\,e^{-|sz+(1-s)w-u|^2/r^2}\frac{|z-u|^j}{r^j/\sqrt{s^j}}\,\frac{|w-u|^k}{r^k/\sqrt{(1-s)^k}}\\
&\leq C\,r^3\,s(1-s)\,e^{-s(1-s)|z-w|^2/r^2}\left(1+\frac{|z-w|^{j+k}}{r^{j+k}/\sqrt{s^{j+k}(1-s)^{j+k}}}\right)
\end{split}
\end{equation*}
which concludes the proof.
\end{proof}

{\bf Acknowledgement.}\\
The author is supported by the grant SNSF Ambizione PZ00P2\_161287/1.

\adresse


\begin{thebibliography}{10}






\bibitem{AKN1} L. Amour, M. Khodja and J. Nourrigat.
The semiclassical limit of the time dependent Hartree-Fock equation: the Weyl symbol of the solution.
\emph{Anal. PDE} {\bf 6} (2013), no. 7, 1649--1674. 

\bibitem{AKN2}  L. Amour, M. Khodja and J. Nourrigat.
The classical limit of the Heisenberg and time dependent Hartree-Fock equations: the Wick symbol of the solution. {\it Math. Res. Lett.} {\bf 20} (2013), no. 1, 119--139. 

\bibitem{APPP} A. Athanassoulis, T. Paul, F. Pezzotti and M. Pulvirenti. Strong Semiclassical Approximation of Wigner Functions for the Hartree Dynamics. \emph{Rend. Lincei Mat. Appl.} \textbf{22}  (2011), 525--552.

\bibitem{BD} C. Bardos and P. Degond, {Global existence for the Vlasov--Poisson equation in 3 space variables with small initial data}, \emph{Ann. Inst. H. Poincar\'e Anal. Non Lin\'eaire}, \textbf{2} (1985), 101--118.

\bibitem{BJPSS} N.~{Benedikter}, V.~Jaksic, M.~{Porta}, C. Saffirio and B.~{Schlein}. Mean-field Evolution of Fermionic Mixed States. \emph{Comm. Pure Appl. Math.} \textbf{69} (2016), 2250--2303. 

\bibitem{BPSS} N.~Benedikter, M.~Porta, C.~ Saffirio, B.~Schlein. From the Hartree-Fock dynamics to the
Vlasov equation. \emph{Arch. Ration. Mech. Anal.} \textbf{221} (2016), no. 1, 273--334.

\bibitem{BPS13} N.~{Benedikter}, M.~{Porta} and B.~{Schlein}. {Mean-field evolution of fermionic systems}. \emph{Comm. Math. Phys.} {\bf 331} (2014), 1087--1131.

\bibitem{castella} F. Castella {Propagation of space moments in the Vlasov-Poisson Equation and further results}, {\it Ann. Inst. Henri Poincar\'e} \textbf{16} (1999), no. 4, 503--533.
 
 \bibitem{Dobr} R.~L.~{Dobrushin}. {Vlasov equations}. \emph{Functional Analysis and Its Applications.} {\bf 13} (1979), no.~2, 115--123.
     
\bibitem{DMS} L. Desvillettes, E. Miot and C. Saffirio, {Polynomial propagation of moments and global existence for a Vlasov--Poisson system with a point charge}, {\it Ann. Inst. H. Poincar\'e (C) Anal. Non Lin\'eaire} \textbf{32} (2015), no. 2, 373--400.

\bibitem{EESY}
A.~{Elgart}, L.~{Erd{\H{o}}s}, B.~{Schlein} and H.-T.~{Yau}. {Nonlinear {H}artree equation as the mean field limit of weakly coupled fermions}. \emph{J. Math. Pures Appl. (9)} \textbf{83} (2004), no.~10,
  1241--1273.
\bibitem{Evans} 
L.~C.~Evans.
 \emph{Partial Differential Equations,}
Graduate Studies in Mathematics, {\bf 19}. American Mathematical Society, Providence, RI (1998).

\bibitem{FDLL} ChL.~Fefferman, R.~de la Llave. Relativistic stability of matter--I. \emph{Rev. Mat. Iberoam.} \textbf{2} no. 2 (1986), 119--213.

\bibitem{FigalliLigaboPaul} A.~Figalli, M.~Ligab\`o, T.~Paul. Semiclassical limit for mixed states with singular and rough potentials. \emph{Indiana Univ. Math. J.} \textbf{61} no. 1 (2012), 193--222.

\bibitem{GIMS} I. Gasser, R. Illner, P.A. Markowich and C. Schmeiser. Semiclassical, $t \to
          \infty$ asymptotics and dispersive effects for HF systems. \emph{Math. Modell. Numer. Anal.} {\bf 32} (1998), 699--713.

\bibitem{golse-notes} F. Golse, \emph{Mean Field Kinetic Equations} (2013)\\ available online at http://www.cmls.polytechnique.fr/perso/golse/M2/PolyKinetic.pdf

\bibitem{GolsePaul1} F.~Golse, T.~Paul. The Schr\"odinger Equation in the Mean-Field and Semiclassical Regime. \emph{Arch. Rational Mech. Anal.} \textbf{223} (2017), 57--94.

\bibitem{GolsePaul2} F.~Golse, T.~Paul. Empirical Measures and Quantum Mechanics: Applications to the Mean-Field Limit. arXiv:1711.08350

\bibitem{GolsePulvirentiPaul} F.~Golse, M.~Pulvirenti, T.~Paul. On the Derivation of the Hartree Equation in the Mean Field Limit: Uniformity in the Planck Constant. To appear in \emph{J. Funct. Anal.}


\bibitem{GMP}
S. Graffi, A. Martinez and M. Pulvirenti. Mean-Field approximation of quantum systems and classical limit. \emph{Math. Models Methods Appl. Sci.} {\bf 13} (2003), no. 1, 59--73. 

\bibitem{HS} C.~Hainzl, R.~ Seiringer. General decomposition of radial functions on $\R^n$ and applications to $N$-body quantum systems. \emph{Lett. Math. Phys.} \textbf{61} no. 1 (2002), 75--84.


\bibitem{HoldingMiot} T. Holding, E. Miot. {Uniqueness and stability for the Vlasov-Poisson system with spatial density in Orlicz spaces}. arXiv:1703.03046v1

\bibitem{Iordanski} S. V. Iordanskii. {The Cauchy problem for the kinetic equation of plasma}. {\it Trudy Mat. Inst. Steklov.} \textbf{60} (1961), 181--194.

\bibitem{krylov} N. V. Krylov. \emph{Lectures on elliptic and parabolic equations in H\"{o}lder spaces}, Graduate Studies in Mathematics AMS \textbf{12} (1996) 

\bibitem{L} E.~H.~Lieb. Thomas-Fermi and related theories of atoms and molecules. \emph{Rev. Mod. Phys.} \textbf{53} no. 4  (1981), 603--641.

\bibitem{LL} E.~H.~Lieb, M.~Loss. Analysis: second edition. Graduate Studies in Mathematics \textbf{14} (2001), American Mathematical Society, Providence, RI, 2001. xxii+346 pp. ISBN: 0-8218-2783-9. 

\bibitem{LSi} E.~H.~Lieb, B.~Simon. TheThomas-Fermi theory of atoms, molecules and solids. \emph{Adv. Math.} \textbf{23} (1977), 22--116.

\bibitem{LionsPaul} P.-L. Lions, T. Paul. Sur les mesures de Wigner. {\it Rev. 
Mat. Iberoamericana} {\bf 9} (1993), 553--618. 

 \bibitem{LP} P.-L. Lions, B. Perthame. {Propagation of moments
and regularity  for the 3-dimensional Vlasov-Poisson system.} {\it Invent.
Math.} \textbf{105} (1991), 415--430.

\bibitem{Loeper} G. Loeper. {Uniqueness of the solution to the Vlasov-Poisson system with bounded density}. {\it J. Math. Pures Appl.} (9) \textbf{86} (2006), no. 1, 68--79.

\bibitem{MM} P. A. Markowich, N. J. Mauser. The Classical Limit of a Self-Consistent Quantum Vlasov Equation. {\it Math. Models Methods Appl. Sci.} {\bf 3} (1993), no. 1, 109--124.

\bibitem{Miot} E. Miot. {A uniqueness criterion for unbounded solutions to the Vlasov-Poisson system}. {\it Commun. Math. Phys.} \textbf{345} (2016), no. 2, 469--482.

\bibitem{NS} H.~{Narnhofer}, G.~L.~{Sewell}. {Vlasov hydrodynamics of a quantum
  mechanical model}. \emph{Comm. Math. Phys.} \textbf{79} (1981), no.~1, 9--24.

\bibitem{OkabeUkai} S. Okabe, T. Ukai. {On classical solutions in the large in time of the two-dimensional Vlasov equation}. {\it Osaka J. Math.} \textbf{15} (1978), 245--261.

\bibitem{Pallard} C. Pallard. {Moment propagation for weak solutions to the Vlasov-Poisson system}. \emph{Commun. Partial Differ. Equations} \textbf{37} (2012), no. 7, 1273--1285.

\bibitem{PP} F. Pezzotti, M. Pulvirenti. Mean-field limit and Semiclassical Expansion of a Quantum Particle System. {\it Ann. H. Poincar\'e} {\bf 10} (2009), no. 1, 145--187.

\bibitem{Pf} K. Pfaffelm{o}ser. {Global existence of the Vlasov-Poisson system in three dimensions for general initial data}, {\it J. Differ. Equ.} \textbf{95} (1992), 281--303.

\bibitem{PRSS} M.~Porta, S.~Rademacher, C.~Saffirio, B.~Schlein. Mean field evolution of fermions with
Coulomb interaction. \emph{J. Stat. Phys.} \textbf{166} (2017), 1345--1364.

\bibitem{S18} C.~Saffirio. Mean-field evolution of fermions with singular interaction. arXiv:1801.02883

\bibitem{S:H-VP} C.~Saffirio. {\it In preparation.}

\bibitem{Sch} J. Schaeffer. {Global existence of smooth solutions to the Vlasov-Poisson system in three dimensions}. {\it Commun. Partial Differ. Equations} \textbf{16} (1991), no. 8--9, 1313--1335.

\bibitem{Sobolev} S. L. Sobolev. { On a theorem of functional analysis}. {\it Mat. Sb.} (4) {\bf 46} (1938), 471--497 (translated into English in Transl. Amer. Math. Soc. {\bf 34}, 39--68).

\bibitem{Spohn} H.~{Spohn}. {On the {V}lasov hierarchy}, \emph{Math. Methods Appl. Sci.} \textbf{3}
  (1981), no.~4, 445--455.

\end{thebibliography}
\end{document}